%% file: paper_current.tex
\documentclass[a4paper,10pt]{article}

\usepackage{amsthm,amsmath,amssymb,amsfonts,xcolor,mathtools}
\usepackage[normalem]{ulem}
\usepackage{cancel}
\usepackage{color}
\usepackage{tikz}
\usepackage{appendix}

\input{mac.tex}

\title{ 
Global Results on the Classification of Two-Component Integrable Evolutionary Systems}
\author{}
\author{Alexander V. Mikhailov$^{\star}$, Vladimir Novikov$^{\ddagger}$ and Jing Ping Wang$^\dagger$ 
\\
$\dagger $ School of Mathematics and Statistics, Ningbo University, China
\\
$\ddagger$ Department of Mathematical Sciences, Loughborough University,  UK
\\
$\star$ School of Mathematics, University of Leeds, UK
}
\date{}  

\begin{document}

\maketitle

\begin{abstract}

We derive necessary and sufficient integrability conditions for two-component polynomial evolutionary systems of odd order in $(1+1)$ dimensions. Integrable systems are members of infinite hierarchies of commuting symmetries, which are characterised by their spectral invariants. We prove that there are precisely $24$ possible spectral classes of integrable hierarchies. As an application, we obtain a complete classification of integrable homogeneous hierarchies whose lowest-order equations are of order 3 and 5. The resulting classification naturally splits into two classes. The C--integrable systems are reduced, by means of differential substitutions, to linear--triangular form, while the S--integrable systems are related, through linear changes of variables and differential substitutions, to canonical Drinfeld--Sokolov KdV-type systems associated with affine Lie algebras of rank two.

\end{abstract}


\section{Introduction}

One of the central problems in the theory of integrable systems is to obtain a comprehensive description of integrable equations. This includes establishing strong, effectively verifiable necessary and sufficient criteria for integrability and, ultimately, achieving a complete classification of integrable systems. By integrability of a partial differential equation we mean the existence of an infinite hierarchy of symmetries.

The classification of integrable partial differential systems remains a challenging problem, even for low-order equations. Nevertheless, substantial progress has been made for systems of a fixed order (see, for example, the survey papers \cite{mr86i:58070,mr89e:58062,mr93b:58070,mr95j:35010,sokolovbook}). By \emph{global} classification, we refer to the derivation of integrability conditions that hold uniformly for equations of arbitrary order, enabling the classification of entire integrable hierarchies rather than individual equations of a fixed order. The symmetry approach, based on symbolic representation, was developed precisely for this purpose and has proven to be an effective tool in the global classification of integrable evolutionary hierarchies \cite{mr99g:35058, mr2001h:37147, mr99i:35005, mr1829636, wang98, kamp02a}.

The first global classification result was obtained for scalar homogeneous evolutionary equations
\[
u_t=u_n+F(u_{n-1},\ldots,u),\qquad u_k=\partial_x^k(u)
\]
with nonnegative weight of the dependent variable $u$ \cite{mr99g:35058,mr2001h:37147}. It can be formulated as follows: {\sl  Any integrable equation of this type is a member of an infinite hierarchy of symmetries whose lowest-order equation is of order 2, 3, or 5}. Consequently, the classification problem reduces to the classification of integrable equations of orders 2, 3, and 5. Since there are only finitely many such equations (namely, ten), there are also only finitely many associated hierarchies of symmetries.

The symbolic approach has proved to be an effective tool for studying the integrability of non-evolutionary and integro-differential equations \cite{mr1908645, mnw07, nw07, mn2, mnw2022}, multi-component systems \cite{mr99i:35005,  mr1829636, kamp02a, MR2070382, Peter09}, and multi-dimensional  equations \cite{wang21} (see also review papers \cite{mnw09,sw2009}). 
More recently, the symmetry approach based on symbolic representation has been extended to address the integrability problem for evolutionary differential--difference equations of arbitrary order in both the Abelian \cite{mnw2022} and non-Abelian \cite{nw2024} settings.

Integrable systems and their associated hierarchies can be naturally divided into two categories, corresponding to the S-- and C--integrable systems in the terminology of F.~Calogero \cite{Calogero1991}. C--integrable systems can be linearised via differential substitutions, allowing their integration through direct methods. In contrast, the integration of S--integrable systems relies on the existence of a Lax representation and the spectral transform method. This latter class exhibits a far richer structure and is of greater interest from the perspective of integrable systems theory. S--integrable systems possess infinitely many local conservation laws, multi-Hamiltonian structures, recursion operators, as well as Darboux and Bäcklund transformations. These transformations not only generate integrable differential–difference and fully discrete systems, but also provide effective tools for constructing explicit solutions, including multi-soliton solutions.

In this paper, we investigate the integrability problem for two-component evolutionary systems 
\begin{equation}\label{eveqn0}
 \left\{\begin{aligned}
u_t &= \lambda_1  u_n
      + F_1\!\left( u_{n-1},v_{n-1} ,\ldots,  u,v\right), \\
v_t &= \lambda_2 v_n
      + F_2\!\left( u_{n-1},v_{n-1} ,\ldots,  u,v\right),
\end{aligned}\right.
\qquad  u_{k}=\partial_x^k(u),\, v_{k}=\partial_x^k(v) ,\quad 
\lambda_1,\lambda_2 \in \mathbb{C}, \quad n>1.
\end{equation}
While several classification results exist for specific low-order values of n, our aim is to achieve global classification results.
 
A complete classification of second-order ($ n=2$) systems~\eqref{eveqn0} 
admitting infinitely many symmetries and local conservation laws was obtained in \cite{mr87h:35312, mr87h:35313, mr89g:58092}. 
In this case, it was shown that the functions $F_1$ and $F_2$  must be polynomial in variables $ u_1$ and $v_1$. 
Moreover, it was proved that the existence of local conservation laws forces $\lambda_2=-\lambda_1$ 
for all even-order systems of the form~\eqref{eveqn0}.
The condition $\lambda_2= -\lambda_1$ is no longer necessary if one allows for C--integrable systems. Indeed, a second-order system~\eqref{eveqn0} may possess an infinite hierarchy of symmetries even if $\lambda_2\ne -\lambda_1$. 
The case $\lambda_2 =\lambda_1$ was classified in \cite{mr90a:35212}, while a classification of homogeneous polynomial integrable systems~\eqref{eveqn0} with  $n=2,\, \lambda_2\ne -\lambda_1$ was obtained in \cite{MR2070382}.

For third-order systems~\eqref{eveqn0}, the classification problem was investigated primarily by Meshkov and Balakhnev in \cite{meshkov94,KuleminMeshkov1997,Mesh08,MeshB08,Balak23}. In these works, the authors did not assume the polynomiality of the functions $F_1$ and $F_2$, but restricted their analysis to several specific cases, namely
\[
\frac{\lambda_1}{\lambda_2}=-\frac12;
\qquad \frac{\lambda_1}{\lambda_2}=\frac{\pm3\sqrt5-7}{2}; \qquad
 \lambda_1=1,\ \lambda_2=0.
\]
In the latter two  cases, several integrable  systems were missed in their classification lists.
A detailed comparison of these results with our classification will be given in Section~\ref{secclass}.
An attempt to classify homogeneous polynomial third-order systems~\eqref{eveqn0} with $W(u)=W(v)=2$ admitting a fixed-order generalised symmetry was undertaken by Foursov \cite{Foursov03}. However, this integrability criterion employed there lacks a rigorous justification.

Currently, no classification results are known for integrable systems~\eqref{eveqn0} with $n\geq 5$.  Only a few isolated examples have been reported in the literature \cite{mnw09, Talati15, TW17, Gerdjikov20}.

In contrast to previous studies, which focused on equations of a fixed order, we consider the classification problem at the level of \emph{integrable hierarchies} within the class of systems~\eqref{eveqn0} with polynomial nonlinearities $F_1, F_2$ and odd order~$n$. Systems related by invertible transformations, which in this case consist of scaling of dependent and independent variables and the involution $u\leftrightarrow v$, are considered equivalent. The generalised symmetries, members of the hierarchy  associated with system~\eqref{eveqn0} are evolutionary systems of the form
\begin{equation}
\label{sym_intr}
\left\{
\begin{array}{l}
u_{t_m}=\mu_{1}(m) u_m+G_1(u_{m-1},v_{m-1},\ldots,u,v),\\
v_{t_m}=\mu_{2}(m) v_m+G_2(u_{m-1},v_{m-1},\ldots,u,v),
\end{array}\right. 
\end{equation} 
The sequence 
$$ \left\{\frac{\mu_{2}(m_k)}{\mu_{1}(m_k)}\,;\, m_{k+1}>m_k,\, m_k\in\N  \right\}$$ 
is referred to as the \emph{spectral invariant of the hierarchy}; in particular, it contains the ratio $\lambda_2/\lambda_1$. 

The natural degree grading of polynomials makes it possible to analyse integrability in a systematic manner, starting from quadratic terms and continuing successively with higher-degree terms. This leads to the notions of approximate symmetries and approximate integrability \cite{mr1908645, mnw07}. Within this setting, van der Kamp \cite{Peter09} attempted a classification of two-component evolutionary equations admitting infinitely many approximate symmetries at the quadratic level. 
 However, this result does not solve the full classification problem, as the author noted that ``an equation
may have infinitely many approximate symmetries of degree $2$, but fail to have any
symmetries. This problem involves conditions of higher grading and is left open.''

We derive integrability conditions that are necessary and sufficient for the existence of a hierarchy of symmetries of system~\eqref{eveqn0}. These conditions make it possible to determine the spectral invariant of the hierarchy and to construct its members~\eqref{sym_intr} explicitly. We show that the spectral invariants of integrable hierarchies fall into exactly 24 distinct classes, listed in Theorem~\ref{thmmn} and summarised in Table~\ref{table2}. Furthermore, we prove that, for an integrable system~\eqref{eveqn0}, every nontrivial approximate symmetry of degree $2$ extends uniquely to a full generalised symmetry of the system. When the quadratic terms of $F_1$ and $F_2$ vanish, the analysis of cubic terms is required.
We also prove that if system~\eqref{eveqn0} admits a nontrivial generalised symmetry, then it possesses infinitely many symmetries of the form~\eqref{sym_intr}, and is therefore integrable. These results provide explicit, effectively verifiable, necessary and sufficient conditions for integrability.

The integrability conditions obtained in this work are suitable for addressing two global problems: (I) determining whether a given odd-order system~\eqref{eveqn0} is integrable and predicting the orders of the symmetries that must belong to its hierarchy; and (II) classifying odd-order integrable hierarchies.

The first problem is relatively straightforward. For a given system~\eqref{eveqn0}, once its order and the spectral invariant are known, the corresponding case in Table~\ref{table2} can be identified, thus determining the orders of the symmetries that the system must possess if it is integrable. Then one verifies the existence of a symmetry of sufficiently low order by directly computing its symbol (Theorem~\ref{symsys}), or by other means.

The second problem is considerably more challenging, as it requires the classification of integrable hierarchies corresponding to all $24$ cases listed in Theorem~\ref{thmmn}. For Cases~$1-15$, the smallest guaranteed orders of systems in the corresponding hierarchies are drawn from the set $\{3, 5, 7, 11, 13, 31\}$. By contrast, the hierarchies associated with Cases~$16-24$ may begin with systems of 
arbitrarily large odd order. 

In this paper, we present a complete classification of integrable hierarchies of systems starting from equations of order $3$ and $5$, under natural technical assumptions. These correspond precisely to the Cases $1$, $4$, $10$, $11$, $13$, $14$, $16$, $18$, $22$ and $23$ in Table~\ref{table2}. We conjecture that any odd order integrable system  \eqref{eveqn0} either admits symmetries of order $3$ or $5$, or is C--integrable. For scalar integrable equations, a similar statement has been rigorously proved in \cite{mr99g:35058}. 

Our approach is based on the symbolic representation of differential algebras, introduced into the theory of integrable systems by Gelfand and Dickey \cite{mr58:22746}. The symbolic representation allows us to derive the explicit conditions for the existence of hierarchy of (approximate) symmetries. Within this framework, the integrability conditions can be reformulated as factorisation and divisibility properties for certain multivariable polynomials. The analysis of these properties is based on the results obtained by F. Beukers. Certain special cases of these polynomials were investigated by Cauchy and Liouville \cite{cauchy-liouville}, and later by Mirimanoff \cite{Mirimanoff}. 

The hierarchies obtained in our classification naturally split into two classes: S--integrable hierarchies and linearisable C--integrable hierarchies. We show that all genuinely non-triangular S--integrable hierarchies can be explicitly related to the Drinfeld--Sokolov hierarchies associated with affine Lie algebras of rank~2.
The hierarchies whose lowest-order equations are of order $3$ are connected, through linear changes of variables or differential substitutions, to the canonical Drinfeld--Sokolov hierarchies corresponding to the affine Lie algebras $A_3^{(2)}, A_4^{(2)}$ and $B_2^{(1)}$ . Similarly, hierarchies starting from fifth-order systems are related to the canonical hierarchies associated with the affine Lie algebras $G_2^{(1)}$ and $D_4^{(3)}$. 
In addition, there are two hierarchies that can be reduced by simple differential substitutions to triangular form. In these cases, one equation is the KdV equation, while the other is a linear equation arising from the associated Lax representation.
 
The paper is organised as follows. We begin by introducing the algebraic framework used to study the symmetry structures of two-component evolutionary systems. In Section~2, we define the notions of symmetries and approximate symmetries, and introduce the concept of approximate integrability.

In Section~\ref{symb}, we develop the symbolic representation and reformulate the concepts introduced in the previous section within this framework. We then derive criteria for approximate integrability in terms of symbolic representation. For a given evolutionary system, we obtain recursive formulas for the coefficients of its symmetries (Theorem~\ref{symsys}), demonstrating that they are uniquely determined by the linear part of the symmetry. This result provides an effective method for verifying the existence of symmetries of a prescribed order and for deriving integrability conditions when the orders of the symmetries are known.

The main result of Section~\ref{hom} is the derivation of necessary and sufficient integrability conditions for hierarchies of odd-order homogeneous systems of the form~\eqref{eveqn0}.
This section contains two principal results of the paper, Theorems~\ref{thmmn} and~\ref{bbt}, which together provide the foundation for the classification carried out in the subsequent sections.
Theorem~\ref{bbt} establishes that, for a $2$-approximately integrable system, the existence of a single nontrivial symmetry guarantees the existence of infinitely many symmetries. Consequently, it provides necessary and sufficient conditions for the integrability of such systems.
Theorem~\ref{thmmn}, in turn, gives a complete classification of all possible cases in which a system of the form~\eqref{eveqn0} is $2$-approximately integrable. In particular, it identifies $24$ distinct classes listed in Table~\ref{table2}. For each class, we specify the spectral invariant of the corresponding hierarchy, the admissible orders of the approximate symmetries, and, where applicable, the associated values of the ratio $\mu_2(m)/\mu_1(m)$.

Section~\ref{secclass} contains the complete classification of integrable hierarchies of homogeneous systems whose lowest-order equations are of order $3$ or $5$.

In Section~\ref{LaxTrans}, we analyse the systems obtained in the classification. We show that all S--integrable systems can be related to canonical Drinfeld–Sokolov KdV-type systems and provide explicit differential substitutions that link them. For the C--integrable systems, we construct differential substitutions that reduce them to a linear–triangular form.

In the final section, Section~\ref{dissrem}, we discuss possible extensions that may lead to a more complete understanding of the problem. We also present several observations  arising from our study.

We would like to thank Dr. G.~Zhao for carrying out the classification of third- and fifth-order integrable systems in the case of equal weights of the variables $u$ and $v$, presented in Sections~\ref{3order} and~\ref{5thord}. These results were previously obtained as part of Dr. Zhao's PhD thesis \cite{Zhao2020}, completed under the supervision of one of the authors.

\section{Symmetries and approximate symmetries of evolutionary systems}\label{Sec2}
The main objects of our study are two-component evolutionary equations of the form \eqref{eveqn0} and their symmetry algebras. In this section, we introduce the basic definitions and notation needed for the two-component case, including the algebra of differential polynomials and the concepts of symmetries and approximate symmetries. A detailed treatment of these topics can be found in \cite{mnw07}. The setup presented here can be straightforwardly generalised to systems with more than two dependent variables.

\subsection{Algebra of Differential Polynomials}\label{difalg}

In what follows, we adopt the notation that $u_n$ and $v_n$ denote the $n$-th $x$-derivatives $\partial_x^{\,n} u$ and $\partial_x^{\,n} v$ of the dependent variables $u$ and $v$, respectively. In particular, $u_0$ and $v_0$ coincide with the functions $u$ and $v$ themselves. When no confusion arises, we omit the index~$0$ and simply write $u$ and $v$.

A \emph{$uv$-monomial} is a finite product of the form
\begin{equation}\label{umon}
u^\alpha v^\beta := u_0^{\alpha_0} u_1^{\alpha_1} \cdots
u_n^{\alpha_n} \,
v_0^{\beta_0} v_1^{\beta_1} \cdots 
v_m^{\beta_m},\qquad \alpha_k,\beta_k\in\zp
\end{equation}
for some nonnegative integers $n$ and $m$. The total degree of the monomial is
\begin{equation}\label{albe}
|\alpha| + |\beta|
= (\alpha_0 + \cdots + \alpha_n) + (\beta_0 + \cdots + \beta_m).
\end{equation}

A \emph{differential polynomial} is a finite linear combination of such monomials with complex coefficients. All differential polynomials form a \emph{differential algebra}
\[
\mathcal{R} = \big( \mathbb{C};\, u, v;\, D_x \big),
\]
with usual operations of addition, multiplication, and the derivation $D_x: \mathcal{R} \to \mathcal{R}$  defined by
\begin{equation}\label{Dx}
D_x = \sum_{k \ge 0}
\left(
u_{k+1} \frac{\partial}{\partial u_k}
+ v_{k+1} \frac{\partial}{\partial v_k}
\right).
\end{equation}

The monomials $u^{\alpha}v^{\beta}$ are eigenvectors of the following commuting linear operators:
\begin{equation}\label{Opdeg} \begin{array}{lll}
D_u=\sum_{k\ge 0}u_{k}\frac{\partial}{\partial u_k}, \qquad &D_v=\sum_{k\ge 0}v_{k}\frac{\partial}{\partial v_k}, \qquad &X=\sum_{k\ge 1}k\left(u_{k}\frac{\partial}{\partial
u_k}+v_{k}\frac{\partial}{\partial v_k}\right),
\end{array}
\end{equation}
with
\begin{equation*} \begin{array}{lll}
D_u(u^{\alpha}v^{\beta})=|\alpha|u^{\alpha}v^{\beta},\quad& 
D_v(u^{\alpha}v^{\beta})=|\beta|u^{\alpha}v^{\beta},\quad&
X(u^{\alpha}v^{\beta})=\sum_{k>0}k(\alpha_k+\beta_k) u^{\alpha}v^{\beta}.
\end{array}
\end{equation*}

The algebra $\mathcal{R}$ is bi-graded and decomposes into a direct sum of eigenspaces
\[
\mathcal{R}
= \bigoplus_{i,j} \mathcal{R}^{i,j},
\qquad
\mathcal{R}^{i,j}
= \left\{ f \in \mathcal{R} \mid D_u(f) = i\;f\ (i\in \mathbb{Z}_{\ge 0}),\;\ D_v(f) = j\;f\ (j\in \mathbb{Z}_{\ge 0}) \right\}.
\]
Here $\mathcal{R}^{i,j}$ is a linear subspace of homogeneous differential polynomials in which each monomial has degree $i$ in the variable $u$ and its derivatives, and degree $j$ in the variable $v$ and its derivatives. 
The algebra $\mathcal{R}$ also has a natural grading with respect to the total degree: 
\[ \mathcal{R}= \bigoplus_{q\in \mathbb{Z}_{\ge 0}}{\mathcal R}^q \, ,
\quad {\mathcal R}^q=\{ f\in \mathcal{R} \, |\, (D_u+D_v) f= q f \}=\bigoplus_{i+j=q} \mathcal{R}^{i,j}\, .
\]

Consequently, every element $f\in\cR$ can be uniquely decomposed into homogeneous components:
$$
f=\sum_{q\ge 0}f_q=\sum_{i, j}f_{i,j},\quad f_q\in \cR^q, \  f_{i,j}\in\cR^{i,j}.
$$
We denote by $\pr{p,q}: \mathcal{R}\to \mathcal{R}^{p,q} $ the projection operator defined by 
\begin{equation}\label{projr}
\pr{p,q}(f)=\pr{p,q}\left(\sum_{i, j}f_{i,j}\right)=f_{p,q}.
\end{equation}

The eigenvalue of the operator $X$ counts the total number of derivatives in a monomial. The algebra $\mathcal{R}$ can also be graded with respect to this operator:
\[ \ring = \bigoplus_{p\in \mathbb{Z}_{\ge 0} } \cR_p\, ,
\quad {\cal R}_p=\{ f\in \ring \, |\, X(f)= p f \}\, .
\]
Since $D_u$, $D_v$, and $X$ commute, these gradings are compatible with multiplication, and hence 
\[
\ring = \bigoplus_{p, q\in \mathbb{Z}_{\ge 0} }{\cal R}^q_p, \qquad
\mathcal{R}^{q}_p\mathcal{R}^{q'}_{p'}\subset \mathcal{R}^{q+q'}_{p+p'}.
\]
It is convenient to introduce weighted homogeneous
polynomials when we addressing classification problems.
Let $\mu, \nu$ be two positive rational numbers\footnote{In general, the
weights may be negative or zero, but in this paper
we restrict ourselves to positive weights.}, which we call the
weights of $u$ and $v$, respectively, and set $W(u)=\mu,\
W(v)=\nu$. We say that  $f\in\ring$ is {\it a homogeneous
polynomial of weight $p$}  (and write $W(f)=p$) if 
\begin{equation}\label{X}
X_{\mu\nu}(f)=p f,\quad  X_{\mu\nu}=\mu D_u+\nu D_v +X , 
\end{equation}
where $D_u$, $D_v$ and $X$ are defined in \eqref{Opdeg}.
The eigenvalues of the operator $X_{\mu\nu}$ on $\cR$ form a countable set denoted by $\cal W$.
Since both $\mu$ and $\nu$ are positive, the subspaces of homogeneous polynomials of weight $p$
are finite dimensional for $p\in \cal W$.

\subsection{Symmetries and approximate symmetries}
A two-component system of partial differential equations of the form
\begin{equation}\label{ev0}
{\bf u}_t={\bf f},\qquad 
\mathbf{u}
=
\begin{pmatrix}
u \\
v
\end{pmatrix},\quad
\mathbf{f}
=
\begin{pmatrix}
f_1 \\
f_2
\end{pmatrix}
\in \mathcal{R} \times \mathcal{R},
\end{equation}
is called an \emph{evolutionary system}. It corresponds to the \emph{evolutionary derivation} $D_{\mathbf{f}} : \mathcal{R} \to \mathcal{R}$ defined by
\[
D_{\mathbf{f}}
=
\sum_{k \ge 0}
\left(
D_x^{\,k}(f_1)\frac{\partial}{\partial u_k}
+
D_x^{\,k}(f_2)\frac{\partial}{\partial v_k}
\right),
\]
that is, a derivation commuting with $D_x$. Any evolutionary derivation is uniquely determined by an element of the linear $\mathbb{C}$--vector space
$
\mathcal{L} = \mathcal{R} \times \mathcal{R},
$
called the characteristics of this derivation.
For example, the derivations $D_u$ and $D_v$ given by \eqref{Opdeg} are evolutionary and correspond to the vectors in $\mathcal{L}$:
\[{\bf d}_u=
\begin{pmatrix}
u \\
0
\end{pmatrix},\qquad {\bf d}_v=\begin{pmatrix}
0 \\
v
\end{pmatrix},
\]
respectively. In contrast, the derivation $X$ is not evolutionary since
\[
[X, D_x] = D_x .
\]

Let $D_{\mathbf{f}}$ and $D_{\mathbf{g}}$ be two evolutionary derivations. Then their commutator is again an evolutionary derivation:
\[
[D_{\mathbf{f}},  D_{\mathbf{g}}]=D_{\mathbf{f}} D_{\mathbf{g}} - D_{\mathbf{g}} D_{\mathbf{f}} = D_{\mathbf{h}},
\]
where $\mathbf{h} = [\mathbf{f}, \mathbf{g}]$ is the Lie bracket defined by
\begin{equation}\label{liesys}
[\mathbf{f}, \mathbf{g}]
=
\sum_{k \ge 0}
\begin{pmatrix}
\dfrac{\partial g_1}{\partial u_k} D_x^{\,k}(f_1)
+ \dfrac{\partial g_1}{\partial v_k} D_x^{\,k}(f_2)
-\dfrac{\partial f_1}{\partial u_k} D_x^{\,k}(g_1)
- \dfrac{\partial f_1}{\partial v_k} D_x^{\,k}(g_2)
\\[1.2ex]
\dfrac{\partial g_2}{\partial u_k} D_x^{\,k}(f_1)
+ \dfrac{\partial g_2}{\partial v_k} D_x^{\,k}(f_2)
-\dfrac{\partial f_2}{\partial u_k} D_x^{\,k}(g_1)
- \dfrac{\partial f_2}{\partial v_k} D_x^{\,k}(g_2)
\end{pmatrix}.
\end{equation}

The Lie bracket $[\cdot,\cdot] : \mathcal{L} \times \mathcal{L} \to \mathcal{L}$ equips $\mathcal{L}$ with the structure of an infinite-dimensional Lie algebra.
The Lie algebra $\mathcal{L}$ is bi-graded and decomposes into a direct sum of eigenspaces of ${\bf d}_u$ and ${\bf d}_v$:
\[
\mathcal{L}
= \bigoplus \mathcal{L}^{n,m},
\qquad
\mathcal{L}^{n,m}
= \{{\bf f}  \in \mathcal{L} \mid [{\bf d}_u, {\bf f}] = n {\bf f},\; [{\bf d}_v, {\bf f}] = m {\bf f}, \ n,m \in \mathbb{Z}_{\ge -1}\}.
\]
The gradings of $\cL$ can also be defined by inheriting the 
bi-graded structure of differential algebra $\cR$, that is,
\[  \cL^{n,m}=\cR^{n+1,m}\times\cR^{n,m+1}, \qquad [\cL^{n,m},\cL^{i,j}]\subseteq \cL^{n+i,m+j}\, .
\]
Similarly to algebra $\ring$, the Lie algebra $\mathcal{L}$ also has a natural grading with respect to the eigenspaces of ${\bf d}_u+{\bf d}_v$: 
\begin{equation}\label{gradLie}
\mathcal{L}= \bigoplus_{n\in \bbbz_{\ge -1}}{\mathcal L}^n \, ,
\quad {\mathcal L}^n=\{ {\bf f}\in \mathcal{L} \, |\, [{\bf d}_u+{\bf d}_v, {\bf f}] = n {\bf f}\}=\bigoplus_{i+j=n} \mathcal{L}^{i,j}\, .
\end{equation}
For example, $\cL^{-1,-1}={\bf 0}$, where ${\bf 0}$ is the zero column vector, and
\[\cL^{-1}=\cL^{-1,0}\oplus\cL^{0,-1}\, , \quad \cL^0=\cL^{-1,1}\oplus\cL^{0,0}\oplus\cL^{1,-1}\, ,\quad
\cL^1=\cL^{-1,2}\oplus\cL^{0,1}\oplus\cL^{1,0}\oplus\cL^{2,-1} .\] 
For $p \in \bbbz_{\ge 0}$,
\[ \cL^{-1,1}=\spc\left\{\left(\begin{array}{l}v_p\\0\end{array}\right)\right\},\quad
\cL^{0,0}=\spc\left\{\left(\begin{array}{l}u_p\\0\end{array}\right),\left(\begin{array}{l}0\\v_p\end{array}\right)\right\},\quad
\cL^{1,-1}=\spc\left\{\left(\begin{array}{l}0\\u_p\end{array}\right)\right\},
\]
and, for $i, j \in \bbbz_{\ge 0}$,
\[ \cL^{-1,2}=\spc\left\{\left(\begin{array}{l}v_i v_j\\0\end{array}\right)\right\},\quad
\cL^{0,1}=\spc\left\{\left(\begin{array}{l}u_i v_j\\0\end{array}\right)\, ,\,
\left(\begin{array}{l}0\\v_i v_j\end{array}\right)\right\},
\ \mbox{etc}.\]
We denote by $\pl{k}$ and $\pl{{p,q}}$ the projection operators 
$$
\pl{k}\left(\sum_{i\ge -1}\bff_i\right)=\bff_k,\quad \bff_i\in\cL^i,
$$
and
$$
\pl{p,q}\left(\sum_{i,j\ge -1}\bff_{i,j}\right)=\bff_{p,q},\quad \bff_{p,q}\in\cL^{p,q}.
$$
On homogeneous components, the projection operators on $\cL$ can be expressed in terms of those on $\cR$. Specifically,
\begin{equation}\label{projLR}
 \pl{k} (\bff)=\begin{pmatrix} \pr{k+1}(f_1)\\ \pr{k+1}(f_2) \end{pmatrix} \quad \mbox{and} \quad   \pl{p,q}(\bff)=\begin{pmatrix} \pr{p+1,q}(f_1)\\ \pr{p,q+1}(f_2) \end{pmatrix} .
\end{equation}

There is a canonical way to define a filtration for the grading given in  \eqref{gradLie}.
For each $l\in\mathbb{Z}_{\ge 0}$, let 
$F^l \mathcal{L}= \bigoplus_{q\geq l}{\mathcal L}^q $. This leads to
\begin{equation}\label{filterlie}
{\mathcal L}\supset F^0 \mathcal{L}\supset F^1 \mathcal{L}\cdots \supset F^{l-1} \mathcal{L} \supset F^l \mathcal{L} \supset F^{l+1} \mathcal{L} \cdots
\end{equation}
satisfying $\mathcal{L}^{l}=F^l \mathcal{L}/F^{l+1} \mathcal{L}$. We define the quotient map $\pi_l: \mathcal{L} \to \mathcal{L}/F^{l} \mathcal{L}$ by
\begin{equation}\label{projf}
    \pi_l({\bf f})={\bf f}+F^{l} \mathcal{L}, \qquad {\bf f} \in \mathcal{L} .
\end{equation}
It is clear that $\ker(\pi_l)=F^{l} \mathcal{L}$, that is, $\pi_l({\bf f})=0$ if and only if the degrees of both its components of $\bf f$ are greater than $l+1$.

Notice that the subspaces $\cL^{i,j}$ are infinite dimensional. We define the  weighted homogeneous vectors following the definition of weighted homogeneous
polynomials in previous section.
Let $W(u)=\mu$ and $W(v)=\nu$. We say that  $\bff\in\cL$ is homogeneous
 of weight $p$  (and write $W(\bff)=p$) if $[X_{\mu\nu}, D_{\bff}]=p D_{\bff}$, where $X_{\mu\nu}$ is defined by \eqref{X}. This is equivalent to
 $$
 W(f_1)=p+\mu \quad \mbox{and} \quad W(f_2)=p+\nu \quad \mbox{for} \quad \bff=(f_1, f_2)^T\in \cal L.
 $$
We denote all possible eigenvalues of  ${\rm ad}_{X_{\mu\nu}}$ by $\cal W$.
The {\sl weighted} grading of $\cL$ is
\[ \cL = \bigoplus_{p\in {\cal W}}{\cal L}_p \, ,
\quad {\cal L}_p=\{ f\in \cL\, |\, [X_{\mu\nu}, D_{\bff}]=p D_{\bff} \}\, .
\]
The algebra $\cL$ admits a degree-weighted grading
\[
\cL= \bigoplus_{p\in {\cal W}, q\in \mathbb{Z}_{\ge -1}}{\cal L}^q_p, \qquad
[\mathcal{L}^{q}_p, \mathcal{L}^{q'}_{p'}]\subset \mathcal{L}^{q+q'}_{p+p'}.
\]
The subspaces $\cL_p$ are finite dimensional for $p\in \cal W$ when both $\mu$ and $\nu$ are positive. For example, when $p=3$, $\mu=2$ and $\nu=\frac{5}{2}$, we have
\[ \cL_3=\spc\left\{\left(\begin{array}{l}u_3\\0\end{array}\right), \left(\begin{array}{l}0\\v_3\end{array}\right),
\left(\begin{array}{l}u u_1\\0 \end{array}\right), \left(\begin{array}{l}v^2\\0 \end{array}\right), \left(\begin{array}{l}0\\u_1 v\end{array}\right),
\left(\begin{array}{l}0\\uv_1 \end{array}\right)\right\}\]
\begin{Def}
An evolutionary equation \eqref{ev0} is said to be homogeneous of weight $p$ if there exist $W(u)=\mu$ and $W(v)=\nu$ such that $\bff\in \cL_p$.
\end{Def}

For every element $f\in\cR$ we define its Fr\'echet derivative by
\begin{eqnarray}\label{Fder}
f_*=\left(f_{*u},f_{*v}\right),\quad f_{*u}=\sum_{i\ge 0}\frac{\partial f}{\partial u_i}D_x^i,\quad f_{*v}=\sum_{i\ge 0}\frac{\partial f}{\partial v_i}D_x^i,   
\end{eqnarray}
and set $\ord\,f_*=\max\{\ord\,f_{*u}, \ord\,f_{*v}\}$.

The Fr\'echet derivative $f_*$ defines a linear map $\cL\to\cR$ with action given by
$$
f_*(\bg)=\left(f_{*u},f_{*v}\right)\left(
\begin{array}{l}g_1\\ g_2
\end{array}
\right)=f_{*u}(g_1)+f_{*v}(g_2)=D_{\bg}(f),\qquad \bg=(g_1, g_2)^T\in\cL.
$$
Thus, in terms of Fr\'echet derivatives, the Lie bracket~\eqref{liesys} takes the form
\begin{eqnarray}\label{liesys1}
[\bff,\bg]=\left(\begin{array}{l}g_{1,*u}(f_1)+g_{1,*v}(f_2)-f_{1,*u}(g_1)-f_{1,*v}(g_2)\\ g_{2,*u}(f_1)+g_{2,*v}(f_2)-f_{2,*u}(g_1)-f_{2,*v}(g_2)\end{array}\right). 
\end{eqnarray}

\begin{Def}
An element $\mathbf{g} \in \mathcal{L}$ is called a \emph{symmetry} (or a \emph{generator of an infinitesimal symmetry}) of the evolutionary system~\eqref{ev0} if
\[
[\mathbf{g}, \mathbf{f}] = 0 .
\]
\end{Def}
The set of all symmetries of the system~\eqref{ev0}, denoted by
\[
\mathcal{C}_{\mathcal{L}}(\mathbf{f})
=
\{ \mathbf{g} \in \mathcal{L} \mid [\mathbf{g}, \mathbf{f}] = 0 \},
\]
is the \emph{centraliser} of $\mathbf{f}$ and a Lie subalgebra of $\mathcal{L}$.

For any $\mathbf{f}$ that does not explicitly depend on the variable $x$, the system~\eqref{ev0} admits a symmetry
\[
\begin{pmatrix}
u_1 \\
v_1
\end{pmatrix}\in \mathcal{C}_{\mathcal{L}}(\mathbf{f}).
\]
This generator corresponds to the classical point symmetry given by translations in the independent variable $x$. Symmetries that do not arise from point transformations are referred to as \emph{generalised symmetries}.

A common approach is to associate with each element $\bg\in\cL$ an evolutionary derivation $D_{\bg}$.
An element $\bg\in\cL$ is a symmetry of the evolutionary system \eqref{ev0} if and only if the corresponding evolutionary derivations $D_{\mathbf{f}}$ and $D_{\bg}$ commute,
that is, $[D_{\mathbf{f}},D_{\bg}]=0$.

The \emph{order} of an element $F\in\cR$ is, by definition, the order of the Fr\'echet derivative $F_*$ as a differential operator on $\cL$: $\ord\,F=\ord\,F_*$. The order of an element $\bff=(f_1, f_2)^T\in\cL$ is defined as $\ord\,\bff=\max\{\ord\,f_1, \ord\,f_2\}$.
By \emph{order} of the system $\bu_t=\bff$ we mean $\ord\,\bff$.

\begin{Def}\label{defint}
System \eqref{ev0} is said to be {\sl integrable} if its algebra of symmetries
$\mathcal{C}_{\mathcal{L}}(\mathbf{f})$ is infinite dimensional and contains symmetries of arbitrary high order, i.e. for every $N\in\N$ the system possesses a symmetry of order greater than $N$.
\end{Def}

In addition to the concept of symmetries, we define the notion of approximate symmetries in a manner similar to that in \cite{mnw07}.
\begin{Def}\label{defaps}
An element $\bg\in\cL$ is called an \emph{approximate symmetry of degree $k$} (or $k$-approximate symmetry) of system $\bu_t=\bff$ if
$$
\pi_{k-1} \left([\bff,\bg]\right)=0,
$$
where projection $\pi_k$ is defined by \eqref{projf}. 
\end{Def}
We denote the set of $k$-approximate symmetries of system $\bu_t=\bff$ by
$$
\mathcal{C}^k_{\mathcal{L}}(\mathbf{f})=\{\bg\in\cL\,|\, \pi_{k-1} \left([\bff,\bg]\right)=0\}.
$$
\begin{Def}
System \eqref{ev0} is said to be {\sl $k$-approximate integrable} if its algebra $\mathcal{C}^k_{\mathcal{L}}(\mathbf{f})$ is infinite dimensional and contains $k$-approximate symmetries of arbitrary high order.
\end{Def}

Every symmetry is an approximate symmetry of any degree. It is clear that a $k$-approximate symmetry is also a $p$-approximate symmetry for any $1\le p<k$, so
$$
\mathcal{C}^1_{\mathcal{L}}(\mathbf{f})\supset \mathcal{C}^2_{\mathcal{L}}(\mathbf{f})\supset \cdots\supset \mathcal{C}^k_{\mathcal{L}}(\mathbf{f})\supset\cdots,\quad \mathcal{C}^\infty_{\mathcal{L}}(\mathbf{f})=\mathcal{C}_{\mathcal{L}}(\mathbf{f}).  
$$
Moreover, if $\bg_1, \bg_2\in\mathcal{C}^k_{\mathcal{L}}(\mathbf{f})$ and $\alpha, \beta\in\C$, then $\alpha \bg_1+\beta\bg_2\in\mathcal{C}^k_{\mathcal{L}}(\mathbf{f})$, and $[\bg_1,\bg_2]\in\mathcal{C}^{k}_{\mathcal{L}}(\mathbf{f})$;
hence $\mathcal{C}^k_{\mathcal{L}}(\mathbf{f})$ is a Lie subalgebra of $\cL$.

\section{Symbolic representation}\label{symb}
In this section, we reformulate the basic definitions and notation introduced in the previous section using symbolic representation, which can be regarded as a simplified notation of a Fourier
transform. In this framework, differential polynomials are mapped into symmetric multivariate polynomials, and the action of derivation $D_x$ is transformed into multiplication by a linear factor. 
Symbolic representation will serve as the main tool throughout the paper.
\subsection{Symbolic representations of differential algebra $\cR$ and Lie algebra $\cL$}

To define the symbolic representation $\hR$ of the graded differential algebra $\cR$,
we first establish an isomorphism of $\C$-linear spaces $\phi:\,\cR^{m,n}\mapsto\hR^{m,n}$ and then extend it to the differential algebra isomorphism by defining the multiplication and the action of the derivation $D_x$ on $\cR$. This construction induces a Lie algebra isomorphism $\tilde{\phi}:\cL\to\hL$, where $\hL=\hR\times\hR$. 
The isomorphism $\phi$ is uniquely determined by its action on monomials.

\begin{Def}\label{symbmon} 
The symbolic representation of a $uv$-monomial is defined as
$$
\phi:\quad u_{i_1}\cdots u_{i_m}v_{j_1}\cdots v_{j_n}\mapsto \hu^m\hv^n\langle\xi_1^{i_1}\cdots\xi_m^{i_m}\rangle_{S_m^{\xi}}\langle\zeta_1^{j_1}\cdots\zeta_n^{j_n}\rangle_{S_n^{\zeta}}\in\hR^{m,n},
$$
where $\langle\cdot\rangle_{S_m^{\xi}}$ and $\langle\cdot\rangle_{S_n^{\zeta}}$ 
denote the symmetrisation (average) over the permutation groups of the $m$ symbolic variables $\xi_1, \ldots, \xi_m$ and the $n$ variables $\zeta_1, \ldots, \zeta_n$, respectively:
$$
\langle a(\xi_1,\ldots,\xi_m)\rangle_{{S_m^{\xi}}}=\frac{1}{m!}\sum_{\sigma\in S_m^{\xi}}a(\xi_{\sigma(1)},\ldots,\xi_{\sigma(m)}),\quad \langle b(\zeta_1,\ldots,\zeta_n)\rangle_{{S_n^{\zeta}}}=\frac{1}{n!}\sum_{\sigma\in S_n^{\zeta}}b(\zeta_{\sigma(1)},\ldots,\zeta_{\sigma(n)}).
$$
\end{Def}
Note that in the above definition  we do not assume the indices $i_1, \ldots, i_m$ and $j_1, \ldots, j_n$ to be distinct. We reserve $\xi_1, \ldots, \xi_m, \zeta_1, \ldots, \zeta_n$ for the symbolic variables in the symbolic representation, while $\hu^m\hv^n$ indicates the appropriate linear space to which the resulting symbol belongs. 

By definition, to the linear combination of monomials corresponds the linear combination of their symbols
$$
\phi(\alpha f+\beta g)=\alpha\phi(f)+\beta\phi(g)=\alpha {\hat f} +\beta {\hat g},\quad f\in\cR^{m,n},\,g\in\cR^{m',n'},\quad \alpha,\beta\in\C.
$$

\begin{Ex} For example
\begin{eqnarray*}
&&u_k\raph\hu\xi_1^k,\quad v_k\raph\hv\zeta_1^k,\quad u_pu_q\raph\frac{1}{2}\hu^2\left(\xi_1^p\xi_2^q+\xi_1^q\xi_2^p\right),\\
&&\alpha u_1u_2v_3+\beta u_1v_2v_3\raph \frac{\alpha}{2}\hu^2\hv\left(\xi_1\xi_2^2+\xi_2\xi_1^2\right)\zeta_1^3+\frac{\beta}{2}\hu\hv^2\xi_1\left(
\zeta_1^2\zeta_2^3+\zeta_2^2\zeta_1^3\right)\,\,\alpha,\beta\in\C.
\end{eqnarray*}
\end{Ex}

Under the isomorphism $\phi$, a homogeneous polynomial $f\in\cR^{m,n}$ is mapped to
\begin{equation}\label{symbf}
 \phi(f)=\hu^m\hv^n a(\xi_1,\ldots,\xi_m,\zeta_1,\ldots,\zeta_n) \in \hR^{m,n},   
\end{equation}
where the polynomial $a$ is invariant under the action of the direct product group $S_m^{\xi}\times S_n^{\zeta}$, i.e. the polynomial $a$ is symmetric separately in the variables $\xi_1, \ldots, \xi_m$ and 
$\zeta_1, \ldots, \zeta_n$.

Recall that the projection operator $\pr{i,j}$ extracts the homogeneous component of a differential polynomial $f\in\cR$ lying in $\cR^{i,j}$. 
In the symbolic representation, we  define the corresponding projection operator $\hat{\Pi}_{\hat{\cR}}^{i,j}$ induced by the commutativity condition $\phi\,\pr{i,j}=\hat{\Pi}_{\hat{\cR}}^{i,j}\,\phi$. 
For $f\in \cR$ we have, on the one hand
$$
f=\sum_{i+j>0}f_{i,j},\quad \pr{i,j}(f)=f_{i,j},\quad \phi(f_{i,j})=\hu^i\hv^j a_{i,j}(\xi_1,\ldots,\xi_i,\zeta_1,\ldots,\zeta_j), 
$$
and, on the other hand
$$
\hat{f}=\phi(f)=\sum_{i+j>0}\hu^i\hv^j a_{i,j}(\xi_1,\ldots,\xi_i,\zeta_1,\ldots,\zeta_j),\quad \hat{\Pi}_{\hat{\cR}}^{i,j}(\hat{f})=\hu^i\hv^j a_{i,j}(\xi_1,\ldots,\xi_i,\zeta_1,\ldots,\zeta_j). 
$$ 

We define multiplication in the symbolic representation as follows:
Let $f\in\cR^{m,n}$ with its symbol given by \eqref{symbf}, and let $g\in\cR^{p,q}$ with the symbol
\begin{equation}\label{symbg}
 \phi(g)=\hu^p\hv^q b(\xi_1,\ldots,\xi_p,\zeta_1,\ldots,\zeta_q).  
\end{equation}
We set the product of their symbols to be the symbol of their product: $\phi(f)*\phi(g)=\phi(fg)$. This gives
\begin{equation}\label{symbmulti} 
 \phi(f)*\phi(g)=\hu^{m+p}\hv^{n+q}\langle\langle a(\xi_1,\ldots,\xi_m,\zeta_1,\ldots,\zeta_n)b(\xi_{m+1},\ldots,\xi_{m+p},\zeta_{n+1},\ldots,
\zeta_{n+q})\rangle_{S^{\xi}_{m+p}}\rangle_{S^{\zeta}_{n+q}}.   
\end{equation}
The inner symmetrisation is taken over permutations of the $\xi$-variables, and the outer symmetrisation is taken over permutations of the $\zeta$-variables. 
This multiplication extends to arbitrary elements in $\cR$ by linearity. In particular, the symbolic representation of monomials (Definition \ref{symbmon}) follows from $\phi(u_k)=\hu\xi_1^n,\,\phi(v_p)=\hv\zeta_1^p$ together with the multiplication rule \eqref{symbmulti}.

Let $\eta$ denote the symbol of the derivative $D_x$, acting by
$$
\eta\bigg(\hu^m\hv^na(\xi_1,\ldots,\xi_m,\zeta_1,\ldots,\zeta_n)\bigg)=\left(\sum_{i=1}^m \xi_i+\sum_{j=1}^n \zeta_j\right) a(\xi_1,\ldots,\xi_m,\zeta_1,\ldots,\zeta_n).
$$
One verifies directly that $\phi(D_x(f))=\eta(\phi(f))$ for all $f\in\cR^{m,n}$. Extension of this rule by linearity defines the action of  $\eta$ on $\hat{\cR}$. 
The algebra $\hat{\cR}$ equipped with the multiplication rule $*$ and action of $\eta$ is isomorphic to the differential algebra $\cR$.

We now describe the symbolic representation of the Fr\'echet derivative.  For any $f\in\cR^{m,n}$ with the symbol given by \eqref{symbf}, note that
$$
\phi\left(\frac{\partial f}{\partial u_i}\right)=\frac{\partial^{i+1} \phi(f)}{\partial u \partial\xi_{m}^i}(\xi_1, \ldots, \xi_{m-1}, 0, \zeta_1, \ldots, \zeta_n),
\quad
\phi\left(\frac{\partial f}{\partial v_j}\right)=\frac{\partial^{j+1} \phi(f)}{\partial v \partial\zeta_{n}^j}(\xi_1, \ldots, \xi_{m}, \zeta_1, \ldots, \zeta_{n-1}, 0) .
$$
Thus, the symbol of the Fr\'echet derivative
$
\phi(f_*)=\hat{f}_*=\left(\hat{f}_{*u},\hat{f}_{*v}\right)
$
 is given by
$$
\hat{f}_{*u}=m\hu^{m-1}\hv^na(\xi_1,\ldots,\xi_{m-1},\eta,\zeta_1,\ldots,\zeta_n),\quad \hat{f}_{*v}=n\hu^{m}\hv^{n-1}a(\xi_1,\ldots,\xi_{m},\zeta_1,\ldots,\zeta_{n-1},\eta).
$$
The symbol $\hat{f}_{*u}$ is symmetric in the variables $\xi_1, \ldots, \xi_{m-1}, \eta$ and separately in $\zeta_1, \ldots, \zeta_n$, while $\hat{f}_{*v}$ is symmetric in $\xi_1, \ldots, \xi_{m}$ and separately in $\zeta_1, \ldots, \zeta_{n-1}, \eta$.

For any monomial $g\in\cR^{p,q}$ with symbol given by \eqref{symbg}, the action of the operators $f_{*u}$ and $f_{*v}$ on $g$ is represented as 
\begin{eqnarray*}
&&\phi\left(f_{*u}(g)\right)=\hat{f}_{*u}(\hat{g})=\hu^{m+p-1}\hv^{n+q}\\ &&\quad\langle\langle m a(\xi_1, \ldots, \xi_{m-1}, \sum_{i=1}^p\xi_{m+i-1}+\sum_{j=1}^q\zeta_{n+j}, \zeta_1, \ldots, \zeta_n)
b(\xi_{m},\ldots,\xi_{m+p-1},\zeta_{n+1},\ldots,\zeta_{n+q})\rangle_{S^{\xi}_{m+p-1}}\rangle_{S^{\zeta}_{n+q}};\\
&&\phi(f_{*v}(g))=\hat{f}_{*v}(\hat{g})=\hu^{m+p}\hv^{n+q-1}\\ &&\quad\langle\langle n a(\xi_1, \ldots, \xi_{m}, \zeta_1, \ldots, \zeta_{n-1}, \sum_{i=1}^p\xi_{m+i}+\sum_{j=1}^q\zeta_{n+j-1})
b(\xi_{m+1},\ldots,\xi_{m+p},\zeta_{n},\ldots,\zeta_{n+q-1})\rangle_{S^{\xi}_{m+p}}\rangle_{S^{\zeta}_{n+q-1}} .
\end{eqnarray*}
By linearity, these formulas immediately extend to arbitrary elements of $\cR$.

We next extend the symbolic representation to the Lie algebra $\cL$ by defining $\tilde{\phi}:\,\cL\to\hat{\cL}=\hat{\cR}\times\hat{\cR}$ by
$$
\tilde{\phi}(\bff)=\hat{\bff}=(\phi(f_1),\phi(f_2))^T=(\hat{f}_1,\hat{f}_2)^T,\quad \bff=(f_1,f_2)^T\in\cL.
$$
The induced Lie bracket between $\hat{\bff}$ and $\hat{\bg}=(\hat{g}_1,\hat{g}_2)^T$ on $\hat{\cL}$ obtained from \eqref{liesys1} takes the form 
$$
[\hat{\bff},\hat{\bg}]=\left(\begin{array}{l}\hat{g}_{1,*u}(\hat{f}_1)+\hat{g}_{1,*v}(\hat{f}_2)-
\hat{f}_{1,*u}(\hat{g}_1)-\hat{f}_{1,*v}(\hat{g}_2)\\
\hat{g}_{2,*u}(\hat{f}_1)+\hat{g}_{2,*v}(\hat{f}_2)-
\hat{f}_{2,*u}(\hat{g}_1)-\hat{f}_{2,*v}(\hat{g}_2)
\end{array}\right).
$$
By construction, $\tilde{\phi}([\bff,\bg])=[\tilde{\phi}(\bff), \tilde{\phi}(\bg)]$, and hence $\tilde{\phi}:\,\cL\to\hat{\cL}$ is a Lie algebra isomorphism. Clearly, $\hat{\cL}$ inherits the multigraded structure of $\cL$. We denote the corresponding projection operators in $\hat{\cL}$ by adding hats to the notation used in $\cL$, e.g., $\hat{\Pi}_{\hat{\cL}}^k$ and $\hat{\Pi}_{\hat{\cL}}^{p,q}$.

\subsection{Symmetries and approximate symmetries in the symbolic representation}
Consider a two-component evolutionary system of the form
\begin{eqnarray}
\label{sys1} {\bf u}_t={\bff}=\bff^0+\bff^{(1)}, \quad {\bf u}=\begin{pmatrix} u, v   
\end{pmatrix}^T,\quad \bff^0=\begin{pmatrix}
    L_1(u), L_2(v)
\end{pmatrix}^T\in \cL^{0,0}, \quad \bff^{(1)}\in F^1\cL, 
\end{eqnarray}
where $L_s, s=1,2$ are linear differential operators with constant coefficients.  In the special case
 $L_1=\lambda_1 D_x^n$ and $ L_2=\lambda_2 D_x^n$,  this  system reduces to (\ref{eveqn0}).
In the symbolic representation, system \eqref{sys1} takes the form
\begin{eqnarray}
\label{sys1sys}
\hat{\bff}=\begin{pmatrix}
    \hat{f}_1\\\hat{f}_2
\end{pmatrix}
=\begin{pmatrix} \hat{u}\omega_1(\xi_1)+\sum_{i+j\ge 2}\hat{u}^i\hat{v}^{j}a^1_{i,j}(\xi_1,\ldots,\xi_i,\zeta_1,\ldots,\zeta_{j})\\
\hat{v}\omega_2(\zeta_1)+\sum_{i+j\ge 2}\hat{u}^i\hat{v}^{j}a^2_{i,j}(\xi_1,\ldots,\xi_i,\zeta_1,\ldots,\zeta_{j})
\end{pmatrix}, \quad i, j\in \mathbb{Z}_{\ge 0},
\end{eqnarray}
where $\hu\omega_1(\xi_1)=\phi(L_1(u))$, $\hv\omega_2(\zeta_1)= \phi(L_2(v))$, and $\hu^i\hv^ja_{i,j}^s(\xi_1,\ldots,\xi_i,\zeta_1,\ldots,\zeta_{j})=\phi(\Pi_{\cR}^{i,j}(f_s))$ for $i+j\ge 2$ and $s=1,2$ are symbolic representations of nonlinear terms.

When studying generalised symmetries $\bu_{\tau}=\bg$ of \eqref{sys1}, we assume $\bg\in F^0\cL$, thus excluding constant terms. The latter with order $-\infty$ requires separate treatment.

\begin{Pro}\label{prolinpart} Assume that $L_1\neq L_2$ and that both operators are of order greater than $1$. If $\bu_{\tau}=\bg\in F^0\cL$ is a generalised symmetry of \eqref{sys1}, then the symmetry $\bg$ possesses a nonzero linear term $\bg^0\in {\cal L}^{0,0}$.
\end{Pro}

\begin{proof} Suppose that $\bg=(g_1,g_2)^T$ is a symmetry of \eqref{sys1}, and assume that neither component contains a linear term. Let 
\[
\hat g_1=\hat u^j\hat v^k A^1_{j,k}+\cdots, \quad j+k\ge2;
\qquad\quad
\hat g_2=\hat u^{j'}\hat v^{k'} A^2_{j',k'}+\cdots, \quad j'+k'\ge2,
\]
where at least one of the lowest-degree coefficients $A^1_{j,k}$ and $A^2_{j',k'}$ is nonzero.
The symmetry condition implies
\begin{equation}\label{gfun}
\begin{pmatrix}0\\ 0\end{pmatrix}= \Big{[}\begin{pmatrix}\hu\omega_1(\xi_1)\\ \hv\omega_2(\zeta_1)\end{pmatrix}, \begin{pmatrix} \hat{u}^j\hat{v}^{k}A^1_{j,k}\\ \hat{u}^{j'}\hat{v}^{k'} A^2_{j',k'} \end{pmatrix} \Big{]}
=-\begin{pmatrix} G^{\omega_1}_{j,k} \ \hat{u}^j\hat{v}^{k} A^1_{j,k}\\G^{\omega_2}_{j',k'}  \hat{u}^{j'}\hat{v}^{k'}A^2_{j',k'}\end{pmatrix},   
\end{equation}
where
\begin{eqnarray}\label{gomega}
G^{\omega_i}_{j,k}(\xi_1,\ldots,\xi_j,\zeta_1,\ldots,\zeta_k)=\omega_i(\sum_{s=1}^j
\xi_s+\sum_{r=1}^k \zeta_r)-\sum_{s=1}^j \omega_1(\xi_s)-\sum_{r=1}^k \omega_2(\zeta_r),\quad
i=1,2 .
\end{eqnarray}
Since $A^1_{j,k}$ and/or $A^2_{j',k'}$ are nonzero, \eqref{gfun} implies $G^{\omega_1}_{j,k}=0$ and/or $G^{\omega_2}_{j',k'}=0$.
By the assumptions on $\omega_1$ and $\omega_2$, the identities $G^{\omega_1}_{j,k}=0$ and $G^{\omega_2}_{j',k'}=0$ hold if and only if $(j,k)=(1,0)$ and $(j',k')=(0,1)$, respectively. This contradicts the assumption that the lowest-degree terms of $\hat g_1$ and $\hat g_2$ have degree at least two.
Therefore, the generalised symmetry $\bg$ must contain a nonzero linear term $\bg^0\in {\cal L}^{0,0}$.
\end{proof}
Following Proposition \ref{prolinpart}, the symbolic representation of a generalised symmetry of \eqref{sys1} is of the form
\begin{eqnarray}
\label{sys1sym}
\hat{\bg}=\begin{pmatrix}
 \hat{u}\Omega_1(\xi_1)+\sum_{i+j\ge 2}\hat{u}^i\hat{v}^{j}A^1_{i,j}(\xi_1,\ldots,\xi_i,\zeta_1,\ldots,\zeta_{j})\\  
 \hat{v}\Omega_2(\zeta_1)+\sum_{i+j\ge 2}\hat{u}^i\hat{v}^{j}A^2_{i,j}(\xi_1,\ldots,\xi_i,\zeta_1,\ldots,\zeta_{j})
\end{pmatrix}, \quad \Omega_1\neq 0,\ \Omega_2\neq 0,\quad  i, j\in \mathbb{Z}_{\ge 0}.
\end{eqnarray}
Before constructing the nonlinear part of a symmetry, we introduce the notation for the nonlinear components of the equation and its symmetry:
 \begin{eqnarray}\label{nonlinear}
 \hat{F}_s=\sum_{i+j\ge 2}\hat{u}^i\hat{v}^{j}a^s_{i,j}(\xi_1,\ldots,\xi_i,\zeta_1,\ldots,\zeta_{j}),
 \quad \hat{G}_s=\sum_{i+j\ge 2}\hat{u}^i\hat{v}^{j}A^s_{i,j}(\xi_1,\ldots,\xi_i,\zeta_1,\ldots,\zeta_{j}), \quad s=1,2.
 \end{eqnarray}
\begin{The}\label{symsys} 
Assume that $L_1\neq L_2$ and that both orders are greater than $1$.  Let $\bu_{\tau}=\bg$ with symbolic representation \eqref{sys1sym} be a generalised symmetry of system \eqref{sys1} with symbolic representation \eqref{sys1sys}. Then the coefficients $A^i_{j,k}$ are determined by
\begin{eqnarray}
\label{Aijkg}
&&A^i_{j,k}(\xi_1,\ldots,\xi_j,\zeta_1,\ldots,\zeta_k)=\frac{G^{\Omega_i}_{j,k}(\xi_1,\ldots,\xi_j,\zeta_1,\ldots,\zeta_k)}
{G^{\omega_i}_{j,k}(\xi_1,\ldots,\xi_j,\zeta_1,\ldots,\zeta_k)}a^i_{j,k}(\xi_1,\ldots,\xi_j,\zeta_1,\ldots,\zeta_k)+\\
\nonumber && \qquad \qquad+\frac{R^i_{j,k}}{G^{\omega_i}_{j,k}(\xi_1,\ldots,\xi_j,\zeta_1,\ldots,\zeta_k)},
\qquad i=1,2,\,\,\, j, k \in \mathbb{Z}_{\ge0}, \,\, j+k=2, 3,\ldots ,
\end{eqnarray}
where 
\begin{eqnarray}\label{GOmega}
G^{\Omega_i}_{j,k}(\xi_1,\ldots,\xi_j,\zeta_1,\ldots,\zeta_k)=\Omega_i(\sum_{s=1}^j
\xi_s+\sum_{r=1}^k \zeta_r)-\sum_{s=1}^j \Omega_1(\xi_s)-\sum_{r=1}^k \Omega_2(\zeta_r),\quad
i=1,2
\end{eqnarray}
and $G^{\omega_i}_{j,k}$ is defined by \eqref{gomega}.
The quantities $R^i_{j,k}$ are given recursively through lower degree terms $a^i_{s,p}$ and $ A^i_{s,p} $ with $
s+p<j+k$ by
$$
R^i_{j,k}=\left\{\begin{array}{ll} 0, & j+k=2,\\
\hat{\Pi}_{\hat{\cR}}^{j,k}\left(\hat{G}_{i,*u}(\hat{F}_1)+\hat{G}_{i,*v}(\hat{F}_2)-\hat{F}_{i,*u}(\hat{G}_1)-\hat{F}_{i,*v}(\hat{G}_2)\right), & j+k>2.\end{array}\right.
$$
\end{The}
We do not display explicit formula for $R_{j,k}^i$ since they become rather
cumbersome. Nevertheless, they can be computed recursively without difficulty from the above relation. For example, explicit expressions for $R_{3,0}^1$ and $R_{1,2}^1$ are given in Appendix \ref{appc}.
In the scalar case analogous formulae were given, e.g., in the review paper \cite{mnw09}.

\begin{proof} 
The symmetry condition $[\hat{\bff},\hat{\bg}]=0$ is equivalent to
\begin{eqnarray}\label{splitting}
\hat{\pi}_l ([\hat{\bff},\hat{\bg}]))=\sum_{s=-1}^{l+1}\hat{\Pi}_{\hat{\cL}}^{s,l-s}([\hat{\bff},\hat{\bg}])={\bf 0} 
\Longleftrightarrow \hat{\Pi}_{\hat{\cL}}^{s,l-s}([\hat{\bff},\hat{\bg}])={\bf 0}, \quad s=-1, \cdots, l+1, ,\quad l=0, 1, 2, \cdots
\end{eqnarray}
For $l=0$ this condition is satisfied by Proposition \ref{prolinpart}. For $l=1$, applying formula \eqref{gfun}, the vanishing of projections on $\hat{\cL}^{-1,2},\hat{\cL}^{0,1},\hat{\cL}^{1,0}$ and $\hat{\cL}^{2,-1}$ leads to 
\eqref{Aijkg} for $j+k=2$. In particular, the coefficients $A_{1,1}^1$ and $A_{0,2}^2$ are obtained from the case $s=0$ and $l=1$ in \eqref{splitting}, namely,
$$
 \Big{[}\begin{pmatrix}\hu\omega_1(\xi_1)\\ \hv\omega_2(\zeta_1)\end{pmatrix}, \begin{pmatrix} \hu \hv A^1_{1,1}(\xi_1, \zeta_1)\\ \hv^2 A^2_{0,2}(\zeta_1,\zeta_2) \end{pmatrix} \Big{]}
 +\Big{[}\begin{pmatrix} \hu \hv a^1_{1,1}(\xi_1, \zeta_1)\\ \hv^2 a^2_{0,2}(\zeta_1,\zeta_2) \end{pmatrix}, \begin{pmatrix}\hu\Omega_1(\xi_1)\\ \hv\Omega_2(\zeta_1)\end{pmatrix} \Big{]}=0 .
$$
The projection conditions for $l\ge 2$ lead to the relations \eqref{Aijkg} for all $j+k\geq 3$.
\end{proof}
\begin{Cor} Under the assumptions of Theorem \ref{symsys}, the algebra of symmetries of (\ref{sys1}) is commutative.
\end{Cor}
\begin{proof} We first show that any generalised symmetry $\bg$ with vanishing linear part is trivial. Suppose, in the symbolic representation \eqref{sys1sym}, that $\Omega_1(\xi_1)=\Omega_2(\zeta_1)=0$. Then formula \eqref{Aijkg} implies that
$A^i_{j,k}=0$ for $i=1,2$ and $j+k=2$.
By induction on the total degree, the same relation shows that $A^i_{j,k}=0$ for all $j+k\ge3$. Hence $\bg=0$.
Since the commutator of two symmetries is again a symmetry and its linear part vanishes identically, it follows that the commutator itself must vanish. Therefore, the symmetry algebra is commutative.
\end{proof}

For a given system, Theorem~\ref{symsys} shows that any symmetry, if it exists, is uniquely determined by its linear part, that is, by the polynomials $\Omega_1(\xi_1)$ and $\Omega_2(\zeta_1)$. All higher-degree coefficients $A^i_{j,k}$, with $i=1,2$ and $j+k\ge2$, are then determined recursively by~\eqref{Aijkg}.
The expression $\hat{u}^j\hat{v}^kA^i_{j,k}$ represents the symbol of an element of the differential algebra $\cR$ if and only if $A^i_{j,k}$ is a symmetric polynomial in the variables $\xi_1,\ldots,\xi_j$ and $\zeta_1,\ldots,\zeta_k$. In general, however, formula~\eqref{Aijkg} produces a symmetric rational function. Requiring $A^i_{j,k}$ to be polynomial, equivalently, requiring all denominators to cancel against appropriate factors in the numerators, provides the \emph{necessary integrability conditions} for the system. These conditions impose algebraic constraints both on $\Omega_1$ and $\Omega_2$, as well as on the coefficients $a^i_{j,k}$ of the system itself.

To establish the integrability of a given system, one must show that there exist infinitely many pairs of polynomials $\Omega_1(\xi_1)$ and $\Omega_2(\zeta_1)$ such that all coefficients $A^i_{j,k}$, determined by~\eqref{Aijkg}, are symmetric polynomials. In that case, the corresponding symmetries are well defined, and the system admits an infinite hierarchy of commuting symmetries.
Thus, the classification of integrable hierarchies can be reduced to the problem of determining all possible infinite sequences of pairs of polynomials $\Omega_1(\xi_1)$ and $\Omega_2(\zeta_1)$  that give rise to approximate symmetries, and consequently to the determination of the equations comprising the hierarchy. In the next section, we develop this approach further and classify all possible odd-order homogeneous hierarchies satisfying $\Omega_1(\xi)\neq \Omega_2(\xi)$.

\section{Integrability conditions for homogeneous systems}\label{hom}
In what follows, we focus on {\it homogeneous} systems of evolutionary equations of the form
\begin{equation}
\label{sys1hom}
\left\{\begin{array}{l} 
u_t=f_1=\lambda_1 u_n+F_1(u_{n-1},v_{n-1},\ldots,u,v),\\
v_t=f_2=\lambda_2 v_n+F_2(u_{n-1},v_{n-1},\ldots,u,v),
\end{array}\right.\quad n>1,\quad \lambda_1,\lambda_2\in\C, \quad \lambda_1\,\lambda_2\ne 0, \quad \lambda_1\ne\lambda_2,
\end{equation}
which satisfy $\bff=\begin{pmatrix} f_1, f_2   \end{pmatrix}^T\in\cL_n, n\in \mathbb{N}$ with respect to the weights $W(u)=\mu>0,\,W(v)=\nu>0$.

Our goal is to derive the necessary and sufficient integrability conditions and to obtain a complete description of integrable systems of form \eqref{sys1hom}.
We first prove that any integrable system must contain non-vanishing quadratic or cubic terms. Using the factorisation properties of symmetric polynomials, we then determine 
the spectral invariant of the hierarchy when the system is approximate integrable up to degree $2$ or $3$. 
Finally, we prove the statement that, under explicit and verifiable conditions, the existence of an approximate symmetry of degree $2$ or $3$ implies the existence of a genuine symmetry.
The assumption of homogeneity is primarily technical. With suitable modifications, the results presented below can also be extended to non-homogeneous systems.

By Proposition \ref{prolinpart} and the homogeneity of system \eqref{sys1hom}, the symmetry algebra of the system is generated by homogeneous symmetries of the form
\begin{equation}
\label{sym1hom}
\left\{\begin{array}{l} 
u_{\tau}=g_1=\mu_1 u_m+G_1(u_{m-1},v_{m-1},\ldots,u,v),\\
v_{\tau}=g_2=\mu_2 v_m+G_2(u_{m-1},v_{m-1},\ldots,u,v),
\end{array}\right.\quad \mu_1,\mu_2\in\C,
\end{equation}
where $\bg=\begin{pmatrix} g_1, g_2   \end{pmatrix}^T\in\cL_m,\,  m\in \mathbb{N}$, and $\mu_1 \mu_2\neq 0$. 
We represent both system \eqref{sys1hom} and its generalised symmetry \eqref{sym1hom} ($m>1$ and $m\ne n$) in their symbolic representations:
\begin{equation}
\label{sys1sysh}
\left\{\begin{array}{l} \hat{f}_1= \hat{u} \lambda_1 \xi_1^n+\sum_{i+j\ge 2}\hat{u}^i\hat{v}^{j}a^1_{i,j}(\xi_1,\ldots,\xi_i,\zeta_1,\ldots,\zeta_{j}),\\
 \hat{f}_2=\hat{v} \lambda_2 \zeta_1^n+\sum_{i+j\ge 2}\hat{u}^i\hat{v}^{j}a^2_{i,j}(\xi_1,\ldots,\xi_i,\zeta_1,\ldots,\zeta_{j}),
\end{array}\right.\quad  \lambda_1, \lambda_2\in \mathbb{C}^*
\end{equation}
and
\begin{equation}
\label{sys1symh}
\left\{\begin{array}{l} \hat{g}_1=
 \hat{u} \mu_1\xi_1^m+\sum_{i+j\ge 2}\hat{u}^i\hat{v}^{j}A^1_{i,j}(\xi_1,\ldots,\xi_i,\zeta_1,\ldots,\zeta_{j}),\\  
\hat{g}_2= \hat{v} \mu_2 \zeta_1^m+\sum_{i+j\ge 2}\hat{u}^i\hat{v}^{j}A^2_{i,j}(\xi_1,\ldots,\xi_i,\zeta_1,\ldots,\zeta_{j}),
\end{array}\right.\quad \mu_1, \mu_2\in \mathbb{C}^*
\end{equation}
respectively, where $i, j\in \mathbb{Z}_{\ge 0}$ and $i+j\ge 2$.
\subsection{Systems with neither quadratic nor cubic terms}
For scalar evolutionary equations, it has been shown in \cite{mnw09} that an equation of the form
$$
u_t=\sum_{k=0}^n \alpha_k u_k+f(u_{n-1}, \cdots, u), \quad n\ge 2, \quad \alpha_n\neq 0
$$
cannot be integrable if $f$ is a differential polynomial whose every monomial has a total degree
greater than $3$ (i.e., without quadratic and cubic terms in $u$ and its derivatives) and a total number of derivatives strictly less than $n$. For homogeneous equations with positive weight $W(u)=\mu>0$, the latter condition is automatically satisfied.
We investigate whether an analogous result holds for two-component systems.
\begin{Def}\label{triangular} A system of the form (\ref{sys1hom}) is called {\it triangular} if $\frac{\partial F_1}{\partial
v_k}=0$ or $\frac{\partial F_2}{\partial u_k}=0,\,\,\forall k\ge 0$.
\end{Def}
A particularly important class of triangular systems is given by
\begin{equation}\label{Bakirov}
\left\{\begin{array}{l}
  u_t=\lambda_1 u_n+f(v_{n-1}, \cdots, v),\\ v_t=\lambda_2 v_n,   
\end{array}\right. \quad \lambda_1,\lambda_2\in\C ,
\end{equation}
which we refer to as \emph{linear--triangular} systems. 
Such systems are integrable by standard methods for arbitrary polynomial functions $f$. Their symmetry algebras and the spectral invariants of the symmetry hierarchies have been studied extensively in \cite{mr99i:35005, mr1829636, ba91, ph02}.
\begin{The}\label{noqct}
Assume that system~\eqref{sys1hom} contains neither quadratic nor cubic terms, that is,
$
\Pi_{\cL}^i (\bff)
=0, i=1,2.
$
If system~\eqref{sys1hom} admits a genertalised symmetry, then, up to the involution $u\leftrightarrow v$, it is a linear--triangular system of the form~\eqref{Bakirov}.
\end{The}
\begin{proof} According to Proposition \ref{prolinpart}, any generalised symmetry of the system \eqref{sys1hom} is of the form \eqref{sym1hom} with symbolic representation \eqref{sys1symh}. Since the system contains neither quadratic nor cubic terms, we have $a^s_{i,j}=0$ for $2\le i+j\le3$ and $s=1,2$.
Suppose that there exist indices $i, j$ and $k$ such that $\Pi_{\cR}^{l,m}(f_i)=0$ for all $i=1,2$ and $l+m<j+k$, 
but $\Pi_{\cR}^{j,k}(f_i)\ne 0$ for some $i$ and $j+k\ge 4$. 
Then by Theorem \ref{symsys}, the symbolic representation of the coefficient $A^i_{j,k}$ is given by
$$
A^i_{j,k}(\xi_1, \ldots, \xi_{j}, \zeta_1, \cdots, \zeta_k)=\frac{G^{\Omega^i}_{j,k}(\xi_1, \ldots, \xi_{j}, \zeta_1, \cdots, \zeta_k)}{G^{\omega_i}_{j,k}(\xi_1, \ldots, \xi_{j}, \zeta_1, \cdots, \zeta_k)}a^i_{j,k}(\xi_1, \ldots, \xi_{j}, \zeta_1, \cdots, \zeta_k),
$$
where the polynomials $G^{\omega_i}_{j,k}$ and $G^{\Omega_i}_{j,k}$ are given by \eqref{gomega} and \eqref{GOmega}, namely,
\begin{eqnarray}
G^{\omega_i}_{j,k}&=&\lambda_i(\sum_{s=1}^j\xi_s+\sum_{r=1}^k\zeta_r)^n-\lambda_1\sum_{s=1}^j\xi_s^n-\lambda_2\sum_{r=1}^k\zeta_r^n,\label{Go}\\
G^{\Omega_i}_{j,k}&=&\mu_i(\sum_{s=1}^j\xi_s+\sum_{r=1}^k\zeta_r)^m-\mu_1\sum_{s=1}^j\xi_s^m-\mu_2\sum_{r=1}^k\zeta_r^m.\label{GO}
\end{eqnarray}
By Lemma \ref{FB1} in Appendix \ref{FBlemmas}, both polynomials $G^{\omega_i}_{j,k}$ and $G^{\Omega_i}_{j,k}$ are irreducible over $\C$ whenever $j+k\ge 4$ and $\lambda_i,\mu_i\in\C^*$ for $i=1,2$. Consequently, $G^{\omega_i}_{j,k}$ cannot divide $G^{\Omega_i}_{j,k}$.
Since the system \eqref{sys1hom} is homogeneous and $W(u)=\mu>0$, $W(v)=\nu>0$, the total degree of every polynomial $a^i_{j,k}$ with $j+k\ge 4$ is strictly smaller than $n$, except possibly for the coefficients $a^1_{0,k}$ and $a^2_{k,0}$.
 Therefore, for all remaining indices, the coefficients $A^i_{j,k}$ are polynomials if and only if  $a^i_{j,k}=0$. 

Thus the only potentially nonzero terms are $a^1_{0,k}$ and $a^2_{k,0}$. Their degrees satisfy $$
0\le\deg(a^1_{0,k})\le n+\mu-k\nu, \quad 0\le\deg(a^2_{k,0})\le n+\nu-k\mu. $$
Hence $\deg(a^1_{0,k})\ge n \iff \mu\ge k\nu$ and $\deg(a^2_{k,0})\ge n\iff \nu\ge k\mu$. Since $k\ge 4$, these inequalities cannot hold simultaneously. Therefore at least one of the coefficients must vanish. Without loss of generality (we may change $u\leftrightarrow v$), we may set $a^2_{k,0}=0$ and assume $\mu\ge k\nu$.  The symbolic representation of the system \eqref{sys1hom}  then takes the form
\begin{equation*}
\left\{\begin{array}{l} \hat{f}_1= \hat{u} \lambda_1 \xi_1^n+\hv^ka^1_{0,k}(\zeta_1,\ldots,\zeta_k)+\sum_{i+j\ge k+1}\hat{u}^i\hat{v}^{j}a^1_{i,j}(\xi_1,\ldots,\xi_i,\zeta_1,\ldots,\zeta_{j}) ,\\
 \hat{f}_2=\hat{v} \lambda_2 \zeta_1^n+\sum_{i+j\ge k+1}\hat{u}^i\hat{v}^{j}a^2_{i,j}(\xi_1,\ldots,\xi_i,\zeta_1,\ldots,\zeta_{j}) .
\end{array}\right.
\end{equation*}
Applying the same argument to the terms of degree $k+1$, we conclude that all coefficients $a^2_{i,j}$ with $i+j=k+1$ vanish. Repeating the argument inductively shows that every higher-degree term in $\hat f_2$ is zero. Consequently,
$
\hat f_2=\hat v\lambda_2\zeta_1^n,
$
and the system reduces to the linear--triangular form \eqref{Bakirov}.
\end{proof}

The proof reveals that the homogeneity assumption is not necessary. Indeed, if one assumes that the total numbers of derivatives appearing in $F_1$ and $F_2$ are strictly less than $n$, i.e., $F_1, F_2\in  \bigoplus_{0\le p< n} \cR_p$, then 
the same argument produces the scalar analogue: a system of the form \eqref{sys1hom} without quadratic and cubic terms admits a generalised symmetry if and only if it is linear. In contrast to the scalar case, this result does not follow solely from homogeneity with positive weights $W(u)>0$ and $W(v)>0$.

The following statement can be proved by an argument analogous to that of Theorem~\ref{noqct}, and we therefore omit the proof. It will be used in the next section to identify and eliminate triangular systems arising in the classification procedure.

\begin{Cor}\label{triang}
Assume that system~\eqref{sys1hom} admits a generalised symmetry, that the total numbers of derivatives of every monomial appearing in $F_1$ and $F_2$ are strictly less than $n$, and that the system is triangular up to cubic terms, i.e., $ 
\Pi_\cR^{i,j}(F_2)=0,
\  2\le i+j\le 3,\  i>0$, or
 $
\Pi_\cR^{i,j}(F_1)=0,
\  2\le i+j\le 3,\quad j>0.
$ 
Then it is triangular.
\end{Cor}
As an immediate consequence of Theorem~\ref{noqct}, any non-triangular integrable system of the form \eqref{sys1hom} must contain non-vanishing quadratic and/or cubic terms. 

\subsection{Odd-order systems with zero quadratic and nonzero cubic terms}\label{seccub}
In this section, we consider systems of the form \eqref{sys1hom} with odd order $n$ that contain no quadratic terms but possess non-vanishing cubic terms, that is, that is, 
$\Pi_{\cL}^1(\bff)=\bf{0}$ and $\Pi_{\cL}^2(\bff)\ne\bf{0}$.  We show that the complete classification in this case can be reduced to examining such system of order $3$ possessing a symmetry of order $5$. 

By Theorem~\ref{symsys}, if \eqref{sym1hom} is an approximate symmetry of degree $3$ for the system \eqref{sys1hom}, then the coefficients appearing in its symbolic representation \eqref{sys1symh} up to cubic order are uniquely determined, and they must be polynomials in their respective arguments:
\begin{eqnarray}
\label{Aredcub}
A_{j,k}^i(x,y,z)=\frac{G^{\Omega_i}_{j,k}(x,y,z)}{G^{\omega_i}_{j,k}(x,y,z)}a_{j,k}^i(x,y,z),\quad i=1,2, \quad j+k=3,
\quad j,k\in\mathbb{Z}_{\ge0},
\end{eqnarray}
where $\omega_1(\xi_1)=\lambda_i \xi_1^n$, $\omega_2(\zeta_1)=\lambda_2 \zeta_1^n$, $\Omega_1(\xi_1)=\mu_1 \xi_1^m$ and $\Omega_2(\zeta_1)=\mu_2 \zeta_1^m$.

If the polynomials $G^{\Omega_i}_{j,k}$ and $G^{\omega_i}_{j,k}$ are relatively prime, then the divisibility condition in \eqref{Aredcub} requires $G^{\omega_i}_{j,k}$ to divide $a^i_{j,k}$. In this case, no restriction is imposed on  $\Omega_i$, and the system admits approximate symmetries of every order. However, if the system is homogeneous with positive weights $W(u)>0$ and $W(v)>0$, then, as shown in the proof of Theorem~\ref{noqct}, one has $\deg a^i_{j,k}<n$ for all cubic coefficients except possibly $a^1_{0,3}$ and $a^2_{3,0}$, with at least one of these necessarily vanishing.
Since $\deg G^{\omega_i}_{j,k}=n$ while $\deg a^i_{j,k}<n$ for all $i=1,2$ and $j,k>0$ with $j+k=3$, it follows that the coefficients $A^i_{j,k}$ can be polynomial only if $G^{\Omega_i}_{j,k}$ and $G^{\omega_i}_{j,k}$ possess a nontrivial common factor. The relevant factorisation results are collected in Lemmas~\ref{Gcubdiag}, \ref{Gcubirred}, and \ref{Gcubred} of Appendix~\ref{FBlemmas}. Using these factorisation properties, we obtain the following statement:

\begin{The}\label{cubthm}
Let \eqref{sys1hom} be a non-triangular system of odd order $n$ satisfying
$
\Pi_{\cL}^1(\bff)=\mathbf{0},\
\Pi_{\cL}^2(\bff)\neq\mathbf{0},
$
and let \eqref{sys1sysh} be its symbolic representation. If the system admits an infinite hierarchy of approximate symmetries of degree $3$ whose orders belong to a set $S\subset\mathbb{N}$, then
$
S=2\mathbb{N}+1.
$
\end{The}

\begin{proof} Under the assumption for system \eqref{sys1hom}, we have $\deg G^{\omega_i}_{j,k}=n$, where $n$ is odd.
We first prove that the set $S$ contains no even integers. Suppose, to the contrary, that $m\in S$ is even. For this value of $m$, denote the corresponding polynomials $G^{\Omega_i}_{j,k}$ by $E^{\Omega_i}_{j,k}$.
 
By Lemma \ref{Gcubdiag}, the polynomials $E^{\Omega_1}_{3,0}$ and $E^{\Omega_2}_{0,3}$ are irreducible, and
$$\gcd (G^{\omega_1}_{3,0},\ E^{\Omega_1}_{3,0})=\gcd(G^{\omega_2}_{0,3},\ E^{\Omega_2}_{0,3})=1.$$
Moreover, $\deg a^1_{3,0}\le n-2 \mu < n$ and $\deg a^2_{0,3}\le n-2\nu<n$ since $\mu>0$ and $\nu>0$. According to \eqref{Aredcub}, the coefficients
$A^1_{3,0}$ and $A^2_{0,3}$ are polynomials if and only if $a^1_{3,0}=a^2_{0,3}=0$. 

Lemma \ref{Gcubirred} in Appendix \ref{FBlemmas} implies that for odd $n$, the polynomials $G^{\omega_1}_{2,1}$ and $G^{\omega_2}_{1,2}$ are irreducible. 
Assume that $1\neq \gcd (G^{\omega_1}_{2,1},\ E^{\Omega_1}_{2,1})=G^{\omega_1}_{2,1}, (m>n)$ . Note that 
$$G^{\omega_1}_{2,1}(x,y,0)=\lambda_1 (x+y)^n-\lambda_1(x^n+y^n), \qquad E^{\Omega_1}_{2,1}(x,y,0)=\mu_1 (x+y)^m-\mu_1(x^m+y^m).$$
Since $m$ is even,  the factor $x+y$ does not divide  $E^{\Omega_1}_{2,1}(x,y,0)$, which leads to the contradiction with the assumption. Thus, we have $\gcd (G^{\omega_1}_{2,1},\ E^{\Omega_1}_{2,1})=1$. Similarly, we also obtain $\gcd(G^{\omega_2}_{1,2},\ E^{\Omega_2}_{1,2})=1$.
For the degrees we know $\deg a^1_{2,1}\le n-\mu-\nu<n$ and $\deg a^2_{1,2}\le n-\mu-\nu<n$. Hence $a^1_{2,1}=a^2_{1,2}=0$ is required for $A^1_{2,1},A^2_{1,2}$ to be polynomials.
Lemma \ref{Gcubred} implies that 
$\gcd (G^{\omega_1}_{1,2},\ E^{\Omega_1}_{1,2})=\gcd(G^{\omega_2}_{2,1},\ E^{\Omega_2}_{2,1})=1,$
 which leads to $a^1_{1,2}=a^2_{2,1}=0$ by a similar argument of degrees using  \eqref{Aredcub}.

Finally, the polynomials $G^{\omega_1}_{0,3}$ and $G^{\omega_2}_{3,0}$ are irreducible for odd $n$ by Lemma \ref{Gcubirred}.
Consequently, $A^1_{0,3},A^2_{3,0}$ are polynomials in \eqref{Aredcub} if and only if $G^{\omega_1}_{0,3}\,|\,a^1_{0,3}$ and $G^{\omega_2}_{3,0}\,|\,a^2_{3,0}$. 
However, for degrees of $a^1_{0,3},\,a^2_{3,0}$ we have $$\deg a^1_{0,3}\le n+\mu-3\nu, \quad
 \deg a^2_{3,0}\le n+\nu-3\mu.$$ 
This leads to either $\mu\ge 3\nu$ with $a^2_{3,0}=0$ or $\nu\ge 3\mu$ with $a^1_{0,3}=0$. In either case applying the same arguments used in the proof of Theorem \ref{noqct}, we show that the system is triangular, contradicting the assumption.
Therefore $S$ contains no even integers.

Thus, the set $S$ contains only odd integers. Repeating the above proof by taking $m$ as odd, 
it follows that once the system admits one approximate symmetry of degree $3$ of odd order, it possesses infinitely many approximate symmetry of degree $3$
at all odd orders and we have $a^1_{2,1}=a^2_{1,2}=0$. Moreover, there are no further restrictions on the factorisation properties of the polynomials $G^{\omega_i}_{j,k}$
and $G^{\Omega_i}_{j,k}$, so $S=\{m\,|\,m\in2\bbbn+1\}$, and the ratio $\mu_2(m)/\mu_1(m)$ remains arbitrary.
\end{proof}

In the present setting, Lemma~\ref{FB1} of Appendix~\ref{FBlemmas} implies that the linear part of any odd-order symmetry is relatively $3$-prime to the linear part of the equation \eqref{sys1hom}. Combining Theorem~\ref{cubthm} with Theorem~\ref{themapp} of Appendix~\ref{appc}, we conclude that the existence of a single generalised symmetry \eqref{sym1hom} implies the existence of infinitely many and hence integrability. Therefore, every non-triangular integrable hierarchy with vanishing quadratic terms begins at order $3$ and contains symmetries of all odd orders. In particular, to classify the integrable systems in this case, we only need to examine this type of systems of order $3$ that possess a symmetry of order $5$.

\subsection{Odd-order systems with nonzero quadratic terms}
In this section, we focus exclusively on the case in which both the system \eqref{sys1hom} and its symmetry \eqref{sym1hom} are of odd order. This restriction is not essential since we will show that whenever an odd-order system possesses an infinite-dimensional algebra of approximate symmetries of degree $2$, it also admits a subalgebra of degree-$2$ approximate symmetries whose orders are all odd. In this setting, the polynomials $G^{\Omega_i}_{j,k}$ (see Lemma \ref{gsym}) satisfy particular relations that significantly simplify the analysis.

We aim to derive the necessary and sufficient integrability conditions for all odd-order systems of the form \eqref{sys1hom}. 
We first investigate the conditions under which such systems are approximately integrable of degree $2$, i.e., admitting infinitely many approximate symmetries of degree $2$.
A related study including even-order systems was carried out in \cite{Peter09}, where minimal lists of approximately integrable systems were obtained up to maximal degree divisors. 
Here, we systematically enumerate all possible cases,  $24$ cases in total as stated in Theorem \ref{thmmn}, in order to classify integrable systems in the next section. 
Most of the lemmas used in this analysis are discussed in \cite{kamp02a, Peter09}, while some also appear in \cite{mr99i:35005, mr1829636, wang98}. For the sake of readability and to emphasise the overall strategy, we state these results without proof.

We then show that the existence of a generalised symmetry implying integrability remains valid in the present context. More precisely, we prove that the existence of an approximate symmetry of degree $2$ ensures the existence of a symmetry. 
This result plays a crucial role in the classification carried out in the following section. 

Assume that system \eqref{sys1hom} contains quadratic terms, i.e., $\Pi_{\cL}^1(\bff)\neq{\bf 0}$ with the symbolic representation given by \eqref{sys1sysh}. Let \eqref{sym1hom} be an approximate symmetry of degree $2$ with the symbolic representation \eqref{sys1symh}. From Theorem \ref{symsys} it follows that if \eqref{sym1hom} is an approximate symmetry of degree $2$ of system \eqref{sys1hom}, then the following six functions must be polynomial in their arguments:
\begin{eqnarray}
\label{A1}
&&A^1_{2,0}(\xi_1, \xi_2)=\frac{G^{\Omega_1}_{2,0}(\xi_1, \xi_2)}{G^{\omega_1}_{2,0}(\xi_1, \xi_2)}a^1_{2,0}(\xi_1, \xi_2),\quad A^2_{0,2}(\zeta_1, \zeta_2)=\frac{G^{\Omega_2}_{0,2}(\zeta_1, \zeta_2)}{G^{\omega_2}_{0,2}(\zeta_1, \zeta_2)}a^2_{0,2}(\zeta_1, \zeta_2),\\
\label{A2}
&&A^1_{0,2}(\zeta_1, \zeta_2)=\frac{G^{\Omega_1}_{0,2}(\zeta_1, \zeta_2)}{G^{\omega_1}_{0,2}(\zeta_1, \zeta_2)}a^1_{0,2}(\zeta_1, \zeta_2),\quad A^2_{1,1}(\xi_1, \zeta_1)=\frac{G^{\Omega_2}_{1,1}(\xi_1, \zeta_1)}{G^{\omega_2}_{1,1}(\xi_1, \zeta_1)}a^2_{1,1}(\xi_1, \zeta_1),\\
\label{A3}
&&A^1_{1,1}(\xi_1, \zeta_1)=\frac{G^{\Omega_1}_{1,1}(\xi_1, \zeta_1)}{G^{\omega_1}_{1,1}(\xi_1,  \zeta_1)}a^1_{1,1}(\xi_1,  \zeta_1),\quad
A^2_{2,0}(\xi_1, \xi_2)=\frac{G^{\Omega_2}_{2,0}(\xi_1, \xi_2)}{G^{\omega_2}_{2,0}(\xi_1, \xi_2)}a^2_{2,0}(\xi_1, \xi_2),
\end{eqnarray}
where $\omega_1(\xi_1)=\lambda_i \xi_1^n$, $\omega_2(\zeta_1)=\lambda_2 \zeta_1^n$, $\Omega_1(\xi_1)=\mu_1 \xi_1^m$ and $\Omega_2(\zeta_1)=\mu_2 \zeta_1^m$.

For simplicity, when discussing the properties of the above polynomials, we denote their arguments by
 $x$ and $y$ without further clarification.
Following the approach of Section \ref{seccub} where systems with non-vanishing cubic terms are considered, we now investigate the factorisation properties 
of the polynomials $G^{\Omega_i}_{j,k}$ with $i=1, 2$, $k, j\ge 0$ and $j+k=2$.
Suppose that the polynomials $G^{\Omega_i}_{j,k}(x,y)$ and
$G^{\omega_i}_{j,k}(x,y)$ admit the factorisations
\begin{eqnarray*}
G^{\Omega_i}_{j,k}(x,y)=g^{i}_{j,k}(x,y)h^{\Omega_i}_{j,k}(x,y),\quad
G^{\omega_i}_{j,k}(x,y)=g^{i}_{j,k}(x,y)h^{\omega_i}_{j,k}(x,y),\quad
i=1,2,\,\,j,k\ge 0,\,\,j+k=2.
\end{eqnarray*}
Then the polynomials $h^{\omega_i}_{j,k}(x,y)$ must divide $a^i_{j,k}(x,y)$.
Consequently, the expressions in \eqref{A1}--\eqref{A3} define polynomial functions $A^i_{j,k}$.

For a given system (\ref{sys1hom}), knowing $a^i_{j,k}$ and $G^{\omega_i}_{j,k}$, to
construct approximate symmetries of degree $2$, we compute the greatest common divisor:
\begin{equation}\label{commonf}
h^{\omega_i}_{j,k}(x,y)=\gcd(a^i_{j,k}(x,y),G^{\omega_i}_{j,k}(x,y))
\end{equation}
and then define
\begin{equation}\label{commona}
g^i_{j,k}(x,y)=\frac{G^{\omega_i}_{j,k}(x,y)}{h^{\omega_i}_{j,k}(x,y)}.
\end{equation}
The next step is to determine 
an integer $m>1$ and the corresponding $\mu_1(m),\mu_2(m)\in\bbbc$ such that
$g^i_{j,k}(x,y)|G^{\Omega_i}_{j,k}(x,y).$ 
However, in general this procedure is not algorithmic and is not suitable for integrability testing or the systematic classification of integrable systems.

We resolve the classification problem by imposing the condition of existence of
{\it infinitely} many approximate symmetries. Specifically, suppose that the subalgebra of approximate symmetries of degree $2$ of system (\ref{sys1hom}) is infinite dimensional. Then there exists an infinite  sequence of positive integers $S=\{m_k\,|\, m_k\in \bbbn,\,m_{k+1}>m_k\}$, together with an infinite  sequence of parameters $\mu_1(m)$ and $\mu_2(m)$ with $ m\in S$ such that for 
$$\Omega_1(\xi_1)=\mu_1(m)\xi_1^m,\quad\Omega_2(\zeta_1)=\mu_2(m) \zeta_1^m,\,\, m\in S,$$ the functions
$A^i_{j,k},\,i=1,2,\,j,k\ge 0,\,j+k=2$ determined by (\ref{A1})-(\ref{A3}) are
polynomials in their arguments. Note that the polynomials
$h^{\omega_i}_{j,k}$ defined by \eqref{commonf} depend only on the right-hand side of 
the system (\ref{sys1hom}). Therefore, for all $m\in S$ the polynomials $g^i_{j,k}$ given by \eqref{commona} must divide
$G^{\Omega_i}_{j,k}(x,y)$. 
This implies that for each $i=1,2,\,\,j,k\ge 0,\,\,j+k=2$, all polynomials in an infinite sequence
\begin{equation}\label{seq}
\big{\{}G^{\Omega_i}_{j,k}(x,y)\,|\,\Omega_1(x)=\mu_1(m)x^m,\,\Omega_2(x)=\mu_2(m)x^m,\,m\in
S\big{\}}
\end{equation}
share a nontrivial common factor $g^i_{j,k}(x,y)$. 

The parameters involved in these sequences consist of the infinite sequence of orders $S$ together with constants $\mu_1(m)$ and $\mu_2(m)$ with $m\in S$. Once all possible sequences \eqref{seq}
and their common factors $g^i_{j,k}$ are classified, these parameters and hence the spectral invariants of hierarchies are determined. 
Since each common factor $g^i_{j,k}$ divides $G^{\omega_i}_{j,k}$,  it follows that the polynomials  $G^{\omega_i}_{j,k}$ themselves  belong to the sequences \eqref{seq}, that is , $n\in S$. From now on, when no confusion can arise, we shall write simply $\mu_1,\mu_2$, omitting their dependence on the order $m$.

\underline{Factorisation properties of the polynomials
$G^{\Omega_1}_{2,0}(x,y)$ and $G^{\Omega_2}_{0,2}(x,y)$}.

The factorisation properties of polynomials $G^{\Omega_1}_{2,0}(x,y)$ and $G^{\Omega_2}_{0,2}(x,y)$ are independent of the parameters $\mu_1,\mu_2$ since
$$
G^{\Omega_1}_{2,0}(x,y)=\mu_1\left((x+y)^m-x^m-y^m\right)=\mu_1g_m(x,y),\quad G^{\Omega_2}_{0,2}(x,y)=\mu_2\left((x+y)^m-x^m-y^m\right)=\mu_2g_m(x,y),
$$
where
\begin{equation}\label{hm}
g_m(x,y)=(x+y)^m-x^m-y^m.    
\end{equation}
\begin{Lem}\label{gdiag} Let $S=\{m_k\,|\,m_k=1\!\!\mod 2,\,3\le m_1<m_2<\cdots<m_k<\cdots\}$ be an infinite strictly increasing sequence of odd positive integers.
Consider the corresponding sequence of polynomials $\{g_m(x,y)\,|\,m\in S\}$ with $g_m$ defined by \eqref{hm}. Let
\[
h(x,y)=\gcd\left(g_{m_1}(x,y),g_{m_2}(x,y),\ldots,g_{m_k}(x,y),\ldots\right)
\]
denote the greatest common divisor of all polynomials in this sequence. Then one of the following cases occurs:
\begin{itemize}
\item[{\rm i.}] $h(x,y)=xy(x+y)(x^2+xy+y^2)^2$,\ \ $S=\{m_k\,|\,m_k=1\mod 6\}$;
\item[{\rm ii.}] $h(x,y)=xy(x+y)(x^2+xy+y^2)$,\ \ $S=\{m_k\,|\,m_k=1,5\mod 6\}$;
\item[{\rm iii.}] $h(x,y)=xy(x+y)$,\ \ $S=\{m_k\,|\,m_k=1\mod  2\}$.
\end{itemize}
\end{Lem}
In affine coordinates, the polynomials $g_m$ defined by \eqref{hm} reduce to Cauchy-Liouville-Mirimanoff
polynomials \cite{Mirimanoff}.
Using the Diophantine approximation theory, F.~Beukers has proved 
that, apart from the common divisors listed in Lemma~\ref{gdiag}, any two distinct polynomials $g_m$ and $g_n$ have no additional nontrivial common divisors \cite{mr98e:11029}. This result is crucial for the global classification of
scalar evolutionary equations \cite{mr99g:35058, wang98}.

\underline{Factorisation properties of polynomials
$G^{\Omega_1}_{1,1}(x,y),G^{\Omega_1}_{0,2}(x,y),G^{\Omega_2}_{2,0}(x,y)$ and $G^{\Omega_2}_{1,1}(x,y)$}.

We begin by establishing several relations among these polynomials, which follow directly from the definition of polynomials $G^{\Omega_i}_{j,k}$ given by \eqref{GO}.

\begin{Lem} \label{gsym} Let $m\in 2\bbbn+1$ and $\Omega_i(x)=\mu_i x^m$ for $i=1,2$,  where $\mu_i\in\bbbc$. Then the polynomials defined by \eqref{GO} satisfy
$$
G^{\Omega_1}_{1,1}(x,y)=G^{\Omega_2}_{2,0}(x,-x-y),\qquad G^{\Omega_2}_{1,1}(x,y)=G^{\Omega_1}_{0,2}(-x-y,y) .
$$
\end{Lem}

Lemma~\ref{gsym} implies that, when the order $n$ of the system is odd and the orders of all symmetries (that is, the elements of the sequence $S$) are also odd, it suffices to study the factorisation properties of one sequence of polynomials, for example,
$$\{G^{\Omega_1}_{0,2}(x,y)\,|\,\Omega_1(x)=\mu_1(m)x^m,\,\Omega_2(x)=\mu_2(m)x^m,\,m\in S\subset 2\bbbn+1\}.$$ 
For the same values of $\mu_1(m),\mu_2(m)$ and the same sequence $S$, the factorisation properties of the remaining polynomials $G_{j,k}^{\Omega_i}$ can then be recovered either from Lemma \ref{gsym} or after interchanging $\mu_1\leftrightarrow\mu_2$, which corresponds to exchanging $G^{\Omega_1}_{0,2}$ and $G^{\Omega_2}_{2,0}$. 
If either the order $n$ of the system or the orders of its symmetries are even, a more delicate analysis of the factorisation properties is required; see \cite{Peter09}.

The factorisation properties of $G^{\Omega_1}_{0,2}$ were investigated in \cite{mr99i:35005,mr1829636,ph02} in connection with the classification of linear--triangular systems of the form \eqref{Bakirov}. The corresponding result was stated as Proposition~13 in \cite{Peter09}. For completeness, we reproduce the statement below in notation consistent with the present paper, while omitting the proof.

\begin{Lem}\label{gsym2}  Let $S=\{m_k\,|\, m_k=1\!\!\mod2,\,3\le m_1<m_2<\cdots<m_k<\cdots\}$ be an infinite strictly increasing sequence of odd positive integers.
Consider the corresponding sequence of polynomials $\{\hat{g}_m(x,y)\,|\,m\in S\}$, where
$$ \hat{g}_m(x,y)=\mu_1(m)(x+y)^m-\mu_2(m)(x^m+y^m),\quad\mu_1(m),\mu_2(m)\in\bbbc .
$$
Suppose that  $\mu_1(m_1)\ne \mu_2(m_1),\,\,\mu_1(m_1)\mu_2(m_1)\ne 0$.
Let
\[
h(x,y)=\gcd\left(\hat{g}_{m_1}(x,y), \hat{g}_{m_2}(x,y),\ldots,\hat{g}_{m_k}(x,y),\ldots,\right)
\]
denote the greatest common divisor of all polynomials in this sequence. Then one of the following cases occurs:
\begin{itemize}
    \item[{\rm i.}]
\[
h(x,y)=(x+y)(x-qy)(y-qx)(x-{\bar q} y)(y-{\bar q} x),\quad q=\alpha\frac{\beta-1}{\alpha-1},
\]
where $\alpha^l=\beta^l=1,\,\,l\in2\bbbn+1,\,\,l\ge 5$, and $\alpha\notin \{1, \beta,\beta^{-1}\},\,\,\beta\ne 1$. In this case,
$$S=\{m_k\,|\,m_k=l\mod 2l\}.$$ 
\item[{\rm ii.}] $$h(x,y)=(x+y)(x-qy)^2(y-qx)^2,$$ where $q$ is a primitive root of unity of even order $l\ge 4$. In this case, $S=\{m_k\,|\,m_k=1\mod l\}.$
\item[{\rm iii.}] 
\[
h(x,y)=(x+y)(x-qy)(y-qx), \quad q\in\bbbc\setminus\{0,-1\}. \qquad \mbox{In this case,}\ S=\{m_k\,|\,m_k\in 2\bbbn+1\}.
\]
\item[{\rm iv.}]
\[
h(x,y)=x+y,\quad S=\{m_k\,|\,m_k\in 2\bbbn+1\}
\]
\end{itemize}
In all cases except {\rm(iv)}, the parameters $\mu_1(m)$ and $\mu_2(m)$ satisfy 
\begin{equation}\label{symspec}
\frac{\mu_2(m)}{\mu_1(m)}=\frac{(1+q)^m}{1+q^m},\quad \forall\ m\in S.
\end{equation}
\end{Lem}

From Lemmas~\ref{gdiag}--\ref{gsym2}, one can determine the factorisation properties of the polynomials $G_{j,k}^{\Omega_i}$ for fixed $i$, $j$, and $k$ when the orders of all symmetries are odd. These six families of polynomials naturally group into the following three cases:
\begin{itemize}
    \item [{\rm (a)}] $G^{\Omega_1}_{2,0}$ and $G^{\Omega_2}_{0,2}$;
\item [{\rm (b)}] $G^{\Omega_1}_{1,1}$ and $G^{\Omega_2}_{2,0}$;
\item [{\rm (c)}] $G^{\Omega_2}_{1,1}$ and $G^{\Omega_1}_{0,2}$.
\end{itemize}
To obtain a complete classification, we combine the factorisation patterns arising in each of these cases. We first analyse cases {\rm(b)} and {\rm(c)} simultaneously and then incorporate case {\rm(a)}. The resulting classification is summarised in the theorem that follows.
We begin with a proposition describing the result of combining cases {\rm(b)} and {\rm(c)}. Since its proof is identical to that of Proposition~18 in \cite{Peter09} for odd-order systems, it is omitted.
\begin{Pro}\label{GG}  
Let $S=\{m_k\,|\,k\in \bbbn,\ m_k=1\!\!\mod2,\,3\le m_1<m_2<\cdots<m_k<\cdots\}$ be an infinite strictly increasing sequence of odd positive integers.
Consider two sequences of polynomials: 
\begin{eqnarray*}
\begin{array}{l}
 g^1_m(x,y)=\mu_1(m)(x+y)^m-\mu_2(m)\left(x^m+y^m\right),\\[4pt]
 g^2_m(x,y)=\mu_2(m)(x+y)^m-\mu_1(m)\left(x^m+y^m\right),
\end{array}
\end{eqnarray*}
where $\mu_1(m),\mu_2(m)\in\bbbc$ and $m\in S$.
Assume that $\mu_1(m_1)\ne\mu_2(m_1),\,\,\,\mu_1(m_1)\mu_2(m_1)\ne 0$.
Let
\begin{eqnarray*}
&&g^1(x,y)=\gcd\left(g^1_{m_1}(x,y),g^1_{m_2}(x,y),\ldots,g^1_{m_k}(x,y),\ldots   \right),\\
&&g^2(x,y)=\gcd\left(g^2_{m_1}(x,y),g^2_{m_2}(x,y),\ldots,g^2_{m_k}(x,y),\ldots   \right)
\end{eqnarray*}
denote the greatest common divisors of the corresponding polynomial sequences.
Then if both $g^1(x,y)$ and $g^2(x,y)$ are nontrivial, one of the cases in  Table  \ref{table1} occurs. Moreover, in every case
the parameters $\mu_1(m)$ and $\mu_2(m)$ satisfy
\eqref{symspec}.

\begin{table}[ht]
\centering
\begin{tabular}{|c|c|c|c|c|}
\hline
Case& $S=\{m_k\,|\,k\in\bbbn\}$&$q$&$g^1(x,y)$& $g^2(x,y)$\\ \hline\hline
 I.& $m_k=1,3,7,9\mod 10$&$e^{\frac{2\pi \i}{5}}$ &$(x+y)(x-qy)(qx-y)$& $(x+y)(x-q^2y)(q^2x-y)$ \\ \hline
 II. &$m_k=1,5\mod 6$& $e^{\frac{\pi \i}{6}}$ &$(x+y)(x-qy)(qx-y)$& $(x+y)(x-q^5y)(q^5x-y)$ \\
 \hline
 III. & $m_k=1\mod 10$& $e^{\frac{2\pi \i}{5}}$& $(x+y)(x-qy)^2(qx-y)^2$& $(x+y)(x-q^2y)^2(q^2x-y)^2$\\ \hline
IV. & $m_k=1\mod 12$& $e^{\frac{\pi \i}{6}}$ &$(x+y)(x-qy)^2(qx-y)^2$& $(x+y)(x-q^5y)^2(q^5x-y)^2$\\ \hline
\end{tabular}
\caption{Factorisation properties of both $G_{0,2}^{\Omega_1}(x, y)$ and $G_{2,0}^{\Omega_2}(x, y)$: $g^1(x,y)$ and
$g^2(x,y)$ are the greatest\\ common divisors of the corresponding polynomial sequences with degrees belonging to $S$.}
\label{table1}
\end{table}
\end{Pro}

\begin{The}\label{thmmn} Let $S$ be an infinite sequence of odd integers greater than $1$. Suppose that system \eqref{sys1hom} with nonzero quadratic terms admits an infinite dimensional algebra of $2$-approximate symmetries of the form \eqref{sym1hom}, whose orders belong to $S$. Then, up to the involution $u\leftrightarrow v$, one of the cases in Table \ref{table2} occurs. Moreover, the quadratic coefficients in the symbolic representation \eqref{sys1sysh} satisfy the following divisibility conditions:
\begin{eqnarray}
&&G^{\omega_1}_{0,2}(x,y)\mid (x+y)g_1(x,y)a^1_{0,2}(x,y),\qquad
G^{\omega_2}_{2,0}(x,y)\mid (x+y)g_2(x,y)a^2_{2,0}(x,y), \nonumber\\
&&G^{\omega_2}_{1,1}(x,y)\mid xg_1(-x-y,y)a^2_{1,1}(x,y),\qquad
G^{\omega_1}_{1,1}(x,y)\mid yg_2(x,-x-y)a^1_{1,1}(x,y), \label{Gfactor}\\
&& G^{\Omega_1}_{2,0}(x,y)\mid x y (x+y) g(x,y) a^1_{2,0}(x,y), \quad \
G^{\Omega_2}_{0,2}(x,y)\mid x y (x+y) g(x,y) a^2_{0,2}(x,y).\nonumber
\end{eqnarray}
Furthermore, in all cases except cases 10, 11 and 12 in Table \ref{table2},  the spectral invariant of the hierarchy is determined by  relations:
\[
\frac{\lambda_2}{\lambda_1}=\frac{(1+q)^n}{1+q^n},\quad
\frac{\mu_2(m)}{\mu_1(m)}=\frac{(1+q)^m}{1+q^m} .
\]
\begin{table}[ht]
\centering
\begin{tabular}{|c|c|c|c|c|c|}
\hline Case& $S: n\in S, m\in S$& $q$& $g_1(x,y)$ & $g_2(x,y)$ & $g(x,y)$\\ \hline\hline 
1.& $1,3,7,9\mod 10$ & $e^{\frac{2\pi \i}{5}}$ & $(x-qy)(y-qx)$ & $(x-q^2y)(y-q^2x)$& $1$\\ \hline 
2.& $1,7,11,13,17,$ & $e^{\frac{2\pi \i}{5}}$ & $(x-qy)(y-qx)$ & $(x-q^2y)(y-q^2x)$& $x^2+xy+y^2$ \\
& $19,23,29\mod 30$ & & & &\\ \hline 
3.& $1,7,13,19\mod 30$ & $e^{\frac{2\pi \i}{5}}$ & $(x-qy)(y-qx)$ & $(x-q^2y)(y-q^2x)$ & $(x^2+xy+y^2)^2$ \\ \hline 
4.&  $1,5\mod 6$& $e^{\frac{\pi \i}{6}}$ & $(x-qy)(y-qx)$ & $(x-q^5y)(y-q^5x)$ & $x^2+xy+y^2$\\ \hline 
5.&  $1\mod 6$& $e^{\frac{\pi \i}{6}}$ & $(x-qy)(y-qx)$ & $(x-q^5y)(y-q^5x)$ & $(x^2+xy+y^2)^2$\\ \hline 
6.& $1\mod 10$ & $e^{\frac{2\pi \i}{5}}$ & $(x-qy)^2(y-qx)^2$ & $(x-q^2y)^2(y-q^2x)^2$ &$1$ \\ \hline 
7.& $1,11\mod 30$ & $e^{\frac{2\pi \i}{5}}$ & $(x-qy)^2(y-qx)^2$ & $(x-q^2y)^2(y-q^2x)^2$ & $x^2+xy+y^2$\\ \hline 
8.& $1\mod 30$ & $e^{\frac{2\pi \i}{5}}$ & $(x-qy)^2(y-qx)^2$ & $(x-q^2y)^2(y-q^2x)^2$ & $(x^2+xy+y^2)^2$\\ \hline 
9.& $1\mod 12$ & $e^{\frac{\pi \i}{6}}$ & $(x-qy)^2(y-qx)^2$ & $(x-q^5y)^2(y-q^5x)^2$ & $(x^2+xy+y^2)^2$ \\ \hline
10.& $1\mod 2$ &  & $1$ & $1$& $1$\\ \hline
11.& $1,5\mod 6$ &  & $1$ & $1$& $x^2+xy+y^2$\\ \hline
12.& $1\mod 6$ &  & $1$ & $1$& $(x^2+xy+y^2)^2$\\ \hline
13.& $1 \mod 2$ & $q\in\bbbc\setminus\{ 0,-1\}$ &  $(x-qy)(y-qx)$ & $1$ & $1$\\ \hline
14.& $1,5 \mod 6$ & $q\in\bbbc\setminus\{ 0,-1\}$ &  $(x-qy)(y-qx)$ & $1$ & $x^2+xy+y^2$\\ \hline 
15.& $1 \mod 6$ & $q\in\bbbc\setminus\{ 0,-1\}$ &  $(x-qy)(y-qx)$ & $1$ & $(x^2+xy+y^2)^2$\\ \hline 
16.& $1\mod  l$ &$q^l=1,\,q\ne\pm1$ & $(x-qy)^2(y-qx)^2$ & $1$&$1$\\
&$3\nmid l\, , 4\le l\in 2 \bbbn$ & & & &\\\hline 
17.& $1, 1+2l\mod 3l$&  $q^l=1,\,q\ne\pm1$ &$(x-qy)^2(y-qx)^2$ & $1$ & $x^2+xy+y^2$\\ 
& $2\neq l=2 \mod 6$& & & &\\\hline
18.& $1, 1+l\mod 3l$&  $q^l=1,\,q\ne\pm1$ &$(x-qy)^2(y-qx)^2$ & $1$ & $x^2+xy+y^2$\\ 
& $l=4 \mod 6$& & & &\\\hline 
19.& $1\mod  3l$ &$q^l=1,\,q\ne\pm1$ & $(x-qy)^2(y-qx)^2$ & $1$&$(x^2+xy+y^2)^2$\\
&$3\nmid l\, , 4\le l\in 2 \bbbn$ & & & &\\\hline 
20.& $1\mod l,\ \ 6\mid l$&  $q^l=1,\,q\ne\pm1$ &$(x-qy)^2(y-qx)^2$ & $1$ &$(x^2+xy+y^2)^2$\\ 
\hline
21.& $l\mod 6l $ &  $\frac{\alpha(\beta-1)}{\alpha-1}\notin \{0,\alpha,-1\}$& $(x-qy)(y-qx)$ & $1$& $(x^2+xy+y^2)^2$\\ 
  &$l =1 \mod 6$ &$\alpha^l=\beta^l=1$  &$\times(x-{\bar q}y)(y-{\bar q}x)$& & \\ \hline
22. & $l\mod 2l\ $ &  $\frac{\alpha(\beta-1)}{\alpha-1}\notin \{0,\alpha,-1\}$& $(x-qy)(y-qx)$ & $1$&$1$\\ 
  & $5\le l \in2\bbbn+1$ &$\alpha^l=\beta^l=1$ &$\times(x-{\bar q} y)(y-{\bar q} x)$& & \\ \hline
23.& $l,5l\mod 6l,\, 3\nmid l\, $ &  $\frac{\alpha(\beta-1)}{\alpha-1}\notin \{0,\alpha,-1\}$& $(x-qy)(y-qx)$ & $1$& $x^2+xy+y^2$\\ 
  &$5\le l \in2\bbbn+1 $ &$\alpha^l=\beta^l=1$ &$\times(x-{\bar q}y)(y-{\bar q}x)$& & \\ \hline
24. & $5 l\mod 6l $ &  $\frac{\alpha(\beta-1)}{\alpha-1}\notin \{0,\alpha,-1\}$& $(x-qy)(y-qx)$ & $1$& $(x^2+xy+y^2)^2$\\ 
  &$l =5 \mod 6$ &$\alpha^l=\beta^l=1$  &$\times(x-{\bar q}y)(y-{\bar q}x)$& & \\ \hline
\end{tabular}
\caption{Spectral invariants and orders of hierarchies of odd-order $2$-approximately integrable systems}
\label{table2}
\end{table}
\end{The}
\begin{proof}
For any fixed system \eqref{sys1hom}, it is clear that $n\in S$. The statement follows by combining the factorisation cases from Lemma~\ref{gdiag} with those listed in Table~\ref{table1}, including the cases where at least one of the common factors $g^1(x,y)$ and $g^2(x,y)$ from Proposition~\ref{GG} is trivial. In the latter situation, we set
 $g^1(x,y)=1$ or  $g^2(x,y)=1$ as appropriate.
For example, combining Case~I of Table~\ref{table1} with Cases~(iii), (ii) and (i) of Lemma~\ref{gdiag} yields Cases~1, 2, and 3 of Table~\ref{table2}, respectively. Continuing this procedure through all admissible combinations gives all remaining cases listed in Table~\ref{table2}.
\end{proof}

\subsection{Necessary and sufficient integrability conditions and global classification}
We define integrability through the existence of an infinite hierarchy of symmetries (see Definition~\ref{defint}). A natural question is whether $2$-approximate integrability already implies full integrability. The following theorem provides an affirmative answer.

\begin{The}\label{bbt}
Let $n, m$ and $k$ be odd integers greater than $1$.
Assume that system (\ref{sys1hom}) of order $n$ has a
generalised symmetry \eqref{sym1hom} of order $m$. Then for any $2$-approximate symmetry of order $k$ of the form $\bh^0+\bh^1$, $\bh^i\in \cL^i$
there exists a unique extension 
$$ \bh=\bh^0+\bh^1+ \bh^{(2)},\quad \bh^{(2)} \in F^2\cL$$ 
such that $\bh$ is a symmetry commuting with \eqref{sym1hom}.
\end{The}
\begin{proof} 
The statement follows directly from Theorem~\ref{themapp} in Appendix~\ref{appc} after identifying the Lie module with the Lie algebra itself, namely by setting
$\cV=\cF=\cL,$
where $\cL$ is defined by \eqref{filterlie}. Since $\cL$ is bi-graded, we decompose the system (\ref{sys1hom}) and its symmetry (\ref{sym1hom}) as
$\bff=\sum_{i,j} \bff^{i,j}$ and $\bg=\sum_{i,j} \bg^{i,j}$. Under the assumptions of the theorem, it therefore remains only to verify that 
$\bg^{0,0}$ is relatively $l$-prime with respect to $\bff^{0,0}$ for $l=2$.

The case $l=2$ corresponds to cubic terms. It suffices to analyse the first component, since the second is obtained by interchanging $u$ and $v$. By Lemmas~\ref{Gcubdiag}, \ref{Gcubirred}, and \ref{Gcubred} in Appendix~\ref{FBlemmas}, the polynomials $G_{j,k}^{\bullet_1}$ with $j+k=3$ are irreducible, except when $j=3$ or $j=1$, where $\bullet$ stands for $\omega$ or $\Omega$.

We first consider the case $j=3$. Since $\bg$ is a symmetry, condition \eqref{quadrac1} from Appendix~\ref{appc} holds. Applying the Jacobi identity to $\bff$, $\bg$, and $\bh$, we obtain
\begin{eqnarray}
&&h_m\, \phi\left(\Pi_{\cR}^{3,0}\left\{ [\bff^{0,1}, \bh^{2,-1}]+[\bff^{1,0}, \bh^{1,0}]+[\bff^{2,-1}, \bh^{0,1}]+[\bff^{2,0}, \bh^{0,0}] \right\}_1\right)=\nonumber\\
&&\qquad=h_n\,\phi\left( \Pi_{\cR}^{3,0} \left\{ [\bg^{0,1}, \bh^{2,-1}]+[\bg^{1,0}, \bh^{1,0}]+[\bg^{2,-1}, \bh^{0,1}]+[\bg^{2,0}, \bh^{0,0}]  \right\}_1\right) , \label{r30} 
\end{eqnarray}
where $\{\ \cdot \}_1$ denotes the first component and $h_m$ is defined in Lemma \ref{Gcubdiag}. Since condition \eqref{quadrac1} is satisfied, both sides of \eqref{r30} are divisible by the factor $(\xi_1+\xi_2)(\xi_1+\xi_3)(\xi_2+\xi_3)$.
Lemma~\ref{Gcubdiag} implies that $h_m$ is irreducible. Thus $G_{3,0}^{\omega_1}$ divides the left-hand side of \eqref{r30}, and therefore there exists a polynomial $h_1^{3,0}$ such that
$$
h_1^{3,0}=\frac{\phi\left(\Pi_{\cR}^{3,0}\left\{ [\bff^{0,1}, \bh^{2,-1}]+[\bff^{1,0}, \bh^{1,0}]+[\bff^{2,-1}, \bh^{0,1}]+[\bff^{2,0}, \bh^{0,0}] \right\}_1\right)}{G_{3,0}^{\omega_1}} .
$$
The case $j=1$ is treated analogously using the condition \eqref{quadrac2}. Finally, since $G_{j,k}^{\omega_i}$ are irreducible for all $j+k>3$ and $i=1,2$, 
it follows that \(\bg^{0,0}\) is relatively \(2\)-prime with respect to \(\bff^{0,0}\). This completes the proof.
\end{proof}

Theorem~\ref{bbt} shows that the existence of a single generalised symmetry is sufficient to ensure that every $2$-approximate symmetry extends uniquely to a genuine symmetry. Hence $2$-approximate integrability is equivalent to integrability.
Since relative $2$-primality of $\bg^{0,0}$ with respect to $\bff^{0,0}$ implies relative $3$-primality, the theorem also covers odd-order systems with vanishing quadratic terms and nonzero cubic terms discussed in Section~\ref{seccub}.

Theorems~\ref{thmmn} and~\ref{bbt} provide the \emph{necessary and sufficient conditions} for integrability. Consider a non-triangular system \eqref{sys1hom} of odd order,  represented in its symbolic form \eqref{sys1sysh}. Its quadratic and cubic terms cannot both vanish if the system is integrable. 
\begin{itemize}
    \item If the quadratic terms are zero, the system is integrable \emph{if and only if} it possesses a symmetry of order $3$, or of order $5$ when $n=3$.
\item If the quadratic terms are nonzero, integrability can be determined through the following procedure:
\begin{enumerate}
    \item Check for $2$-approximate integrability: First, use $\frac{\lambda_2}{\lambda_1}=\frac{(1+q)^n}{1+q^n}$ to obtain the admissible values of $q$. Next, for each $i=1,2$ and $j+k=2$, compute the ratios of the quadratic coefficients $A^i_{jk}$ and the corresponding polynomials $G^{\omega_i}_{jk}$. Compare these with the factorisation formula~\eqref{Gfactor} to identify the factors $g_1(x,y)$, $g_2(x,y)$, and $g(x,y)$. By Theorem \ref{thmmn} , the system is $2$-approximately integrable \emph{if and only if} it corresponds to one of the $24$ cases listed in Table~\ref{table2}.
    \item Verify the existence of a symmetry: For a $2$-approximately integrable system, check whether it possesses a symmetry whose order is the smallest element of $S$, or the second smallest element of $S$ when $n$ is the smallest element. This can be verified using Theorem \ref{symsys}. 
    By Theorem \ref{bbt}, a $2$-approximately integrable system is integrable \emph{if and only if} such a symmetry exists.
\end{enumerate}
\end{itemize}

The preceding results provide a complete description of two-component integrable systems of odd order of the form \eqref{sys1hom}. In particular, the following conclusions hold:
\begin{itemize}
    \item[I.] A non-triangular system of the form \eqref{sys1hom} with neither quadratic nor cubic terms cannot be integrable. Indeed, in the absence of these lower-degree nonlinearities, the system fails to satisfy the necessary conditions for the existence of an infinite hierarchy of symmetries.
\item[II.] 
Consider a non-triangular system \eqref{sys1hom} with $\lambda_1\neq \lambda_2$,  $\lambda_1 \lambda_2\neq 0$, and vanishing quadratic terms but nonzero cubic terms. If such a system is integrable, then it necessarily belongs to a hierarchy of symmetries generated by a third-order system of the same type. The presence of cubic terms in the absence of quadratic ones forces the existence of symmetries at every odd order. Consequently, the classification problem in this setting reduces to identifying third-order systems that admit a symmetry of order $5$; this suffices to recover all integrable systems of this type at arbitrary odd orders.

\item[III.] 
The classification of integrable systems of the form \eqref{sys1hom} therefore reduces to the case of non-vanishing quadratic terms, which is governed by Theorem~\ref{thmmn}. In total, $24$ cases are listed in Table~\ref{table2}, each of which must be analysed individually. By Theorem~\ref{bbt}, the existence of a single symmetry for a $2$-approximately integrable guarantees its integrability. Consequently, for each case, one takes the first element of the corresponding sequence $S$ as the order of the system and the second element as the order of a symmetry with the prescribed linear part, and then applies Theorem~\ref{symsys} to determine all such systems.
For instance, in Case~1 of Table~\ref{table2}, one must determine all third-order systems admitting a seventh-order symmetry with spectral invariant
$$ \frac{\lambda_2}{\lambda_1}=\frac{(1+q)^3}{1+q^3}, \qquad \frac{\mu_2}{\mu_1}=\frac{(1+q)^7}{1+q^7}, \quad q=e^{\frac{2\pi \i}{5}}.$$

In principle, a complete classification can be obtained for Cases $1-15$ in Table~\ref{table2}, since the first two elements of the corresponding sequence $S$ are fixed. However, in several instances, these orders become prohibitively large. For example, in Case~$8$ the hierarchy starts with order $31$, whose first generalised symmetry has order $61$. The computations required in such cases are beyond the scope of this paper.

\item[IV.] In Cases $16-24$, the first two elements of the associated sequence $S$ depend on an integer parameter $l$, allowing the possibility of new integrable systems of arbitrarily high order. Nevertheless, the classification results for fifth-order systems presented in Section~\ref{5thord} suggest that all integrable systems that arise in these cases can be reduced, through Miura transformations, to triangular systems.
\end{itemize}

We illustrate the classification procedure for homogeneous third order systems for a particular choice of  positive weights $W(u)$ and $W(v)$. As a starting point, we determine the most general homogeneous system with undetermined constant coefficients. Using Table~\ref{table2} from Theorem~\ref{thmmn}, we identify the possible spectral invariants of the system and its lowest-order generalised symmetry. For each case, we apply the necessary integrability conditions of Theorem~\ref{symsys} to derive constraints on the parameters and construct the corresponding generalised symmetry. Triangular systems are then discarded. Finally, Theorem~\ref{bbt} implies that all remaining systems obtained in this way are integrable.

\begin{Ex}\label{example22}
Consider the class of third-order homogeneous two-component systems with
$W(u)=W(v)=2.$
The most general system in this class is
\begin{equation}
\label{sysex}
\left\{
\begin{array}{l}
u_t=\lambda_1 u_3+c_1uu_1+c_2vu_1+c_3uv_1+c_4vv_1,\\
v_t=\lambda_2 v_3+d_1uu_1+d_2vu_1+d_3uv_1+d_4vv_1,
\end{array}\right.
\qquad c_i,d_i\in\C,\quad i=1,\ldots,4.
\end{equation}
Our goal is to determine all values of the parameters $\lambda_i$, $c_i$, and $d_i$ for which system \eqref{sysex} is non-triangular and integrable; that is, it admits infinitely many odd-order generalised symmetries of the form \eqref{sym1hom}.
\end{Ex}
The symbolic representation of system \eqref{sysex} is of the form \eqref{sys1sysh}, with nonzero coefficients
\begin{eqnarray*}
&&a^1_{20}(\xi_1, \xi_2)=\frac{c_1}{2}(\xi_1+\xi_2),\quad a^1_{11}(\xi_1, \zeta_1)=c_2 \xi_1+c_3 \zeta_1,\quad a^1_{02}(\zeta_1, \zeta_2)=\frac{c_4}{2}(\zeta_1 +\zeta_2),\\
&&a^2_{20}(\xi_1, \xi_2)=\frac{d_1}{2}(\xi_1+\xi_2),\quad a^2_{11}(\xi_1, \zeta_1)=d_2\xi_1+d_3 \zeta_1,\quad a^2_{02}(\zeta_1, \zeta_2)=\frac{d_4}{2}(\zeta_1 +\zeta_2),
\end{eqnarray*}
while all remaining coefficients $a^i_{jk}$ vanish for $i=1,2$ and $j+k>2$.

Assume that system \eqref{sysex} is integrable. Then, by Theorem~\ref{thmmn}, and up to the interchange $u\leftrightarrow v$, only the following three cases in Table \ref{table2} are possible:
\begin{itemize}
    \item Case 1: The ratio $\frac{\lambda_2}{\lambda_1}=\frac{(1+q)^3}{1+q^3}$, where $q=e^{\frac{2\pi{\i}}{5}}$, and the system admits a hierarchy of symmetries of orders $m=1,3,7,9\mod 10$ with $\frac{\mu_2}{\mu_1}=\frac{(1+q)^m}{1+q^m}$. To classify all integrable systems in this case it is sufficient to impose the existence of a symmetry of order $7$.
    \item Case 13: The ratio $\frac{\lambda_2}{\lambda_1}=\frac{(1+q)^3}{1+q^3}$, where $q\in\C\setminus\{-1\}$, and the system admits a hierarchy of symmetries of every odd order, $m=1\mod 2$ with $\frac{\mu_2}{\mu_1}=\frac{(1+q)^m}{1+q^m}$. It is sufficient to require the existence of a symmetry of order $5$. In the special case $q=e^{\frac{\pi{\i}}{3}}$ one has $\lambda_1=0$.
    \item Case 10: The ratios $\frac{\lambda_2}{\lambda_1}$ and $\frac{\mu_2}{\mu_1}$ are arbitrary, and the system possesses a hierarchy of symmetries of all odd orders $m=1\mod 2$. Again, it suffices to impose the existence of a symmetry of order $5$.
\end{itemize}
We analyse the three possible cases separately. Since Case~13 also contains the case $\lambda_1=0$, we treat this case in details. The existence of a fifth-order symmetry imposes strong polynomiality conditions on the coefficients of the symmetry. These conditions allow us to determine the admissible values of the parameters and reduce the family \eqref{sysex} to a finite number of integrable systems. The remaining cases are then treated in a similar way.

\underline{Case 13}: Let $\frac{\lambda_2}{\lambda_1}=\frac{(1+q)^3}{1+q^3}$ and $\frac{\mu_2}{\mu_1}=\frac{(1+q)^5}{1+q^5}$. 
By Theorem~\ref{symsys}, the quadratic coefficients of the symmetry are given by
\begin{eqnarray*}
&&A^1_{20}=\frac{5\mu_1}{3\lambda_1}(\xi_1^2+\xi_1 \xi_2+\xi_2^2)a^1_{20},
\qquad \qquad
A^2_{02}=\frac{5\mu_2}{3\lambda_2} (\zeta_1^2+\zeta_1 \zeta_2+\zeta_2^2)a^2_{02},\\
&&
A^1_{11}=\frac{\mu_1(\xi_1+\zeta_1)^5-\mu_1\xi_1^5-\mu_2\zeta_1^5}
{\lambda_1(\xi_1+\zeta_1)^3-\lambda_1\xi_1^3-\lambda_2\zeta_1^3}a^1_{11},
\quad
A^2_{20}=\frac{\mu_2(\xi_1+\xi_2)^5-\mu_1(\xi_1^5+\xi_2^5)}{\lambda_2(\xi_1+\xi_2)^3-\lambda_1(\xi_1^3+\xi_2^3)}a^2_{20},\\
&&
A^1_{02}=\frac{\mu_1(\zeta_1+\zeta_2)^5-\mu_2(\zeta_1^5+\zeta_2^5)}{\lambda_1(\zeta_1+\zeta_2)^3-\lambda_2(\zeta_1^3+\zeta_2^3)}a^1_{02},
\quad
A^2_{11}=\frac{\mu_2(\xi_1+\zeta_1)^5-\mu_1\xi_1^5-\mu_2\zeta_1^5}{\lambda_2(\xi_1+\zeta_1)^3-\lambda_1\xi_1^3-\lambda_2\zeta_1^3}a^2_{11} .
\end{eqnarray*}

Assume first that $q\ne e^{\frac{\pi{\i}}{3}}$. We distinguish two cases depending on whether $\mu_1$ is zero:
\begin{itemize}
    \item $q^5=-1$, i.e., $\mu_1=0$: from the expressions of $A^1_{jk},\,i=1,2,\,j+k=2$, it follows that they are polynomials if and only if $c_2=c_3=c_4=0$, i.e., the system is triangular and is discarded.
    \item $q^5\ne -1$. Then coefficients $A^i_{jk},\,i=1,2,\,j+k=2$ are polynomials in their arguments if and only if $c_2=c_3=d_1=0$, and
 \begin{eqnarray*}
&&A^1_{11}=A^2_{20}=0;\\
&&A^1_{02}=\frac{5}{3}((1+q+q^2)(\zeta_1^2+\zeta_1 \zeta_2+\zeta_2^2)-q \zeta_1 \zeta_2)a^1_{02};\\
&&A^2_{11}=\frac{5}{3}((1+q+q^2)(\zeta_1^2+\zeta_1 \xi_1+\xi_1^2)+q \zeta_1(\xi_1+\zeta_1))a^2_{11} .
\end{eqnarray*}
We may then assume that $c_4\ne 0$, as otherwise the system will be triangular. 

Using the expressions of quadratic terms, we first compute the coefficients of nonzero cubic terms in the $u$-component, namely $A^1_{30}, A^1_{03}$ and $A^1_{12}$:
\[
A^1_{30}=\frac{2}{\lambda_1}\frac{\bigl\langle A^1_{20}(\xi_1,\xi_2+\xi_3) a^1_{20}(\xi_2,\xi_3)-a^1_{20}(\xi_1,\xi_2+\xi_3) A^1_{20}(\xi_2,\xi_3)\bigr\rangle_{S_3^{\xi}}}{(\xi_1+\xi_2+\xi_3)^3-\xi_1^3-\xi_2^3-\xi_3^3} =\frac{5 \mu_1}{18 \lambda_1^2}  c_1^2 (\xi_1+\xi_2+\xi_3),
\]
which is a polynomial in its arguments.  For $A^1_{03}$, we obtain
\[
A^1_{03}=\frac{2 \bigl\langle A^1_{02}(\zeta_1,\zeta_2+\zeta_3) a^2_{02}(\zeta_2,\zeta_3)-a^1_{02}(\zeta_1,\zeta_2+\zeta_3) A^2_{02}(\zeta_2,\zeta_3)\bigr\rangle_{S_3^{\zeta}}}{\lambda_1(\zeta_1+\zeta_2+\zeta_3)^3-\lambda_2(\zeta_1^3+\zeta_2^3+\zeta_3^3)} .
\]
The requirement that \(A^1_{03}\) be a polynomial leads to \(d_4=0\) (assuming \(c_4\neq 0\)), implying 
\[
a^2_{02}=A^2_{02}=A^1_{03}=0.
\]
Considering $A^1_{12}$, we have
\begin{eqnarray*}
&&A^1_{12}=\frac{2 A^1_{20}(\xi_1,\zeta_1+\zeta_2) a^1_{02}(\zeta_1, \zeta_2) + 2 \bigl\langle A^1_{02}(\zeta_1, \xi_1+\zeta_2) a^2_{11}(\xi_1, \zeta_2)\bigr\rangle_{S_2^{\zeta}}}
{\lambda_1(\xi_1+\zeta_1+\zeta_2)^3-\lambda_1 \xi_1^3-\lambda_2 (\zeta_1^3+\zeta_2^3)}\\
&&\qquad -\frac{2 a^1_{20}(\xi_1,\zeta_1+\zeta_2) A^1_{02}(\zeta_1, \zeta_2) + 2 \bigl\langle a^1_{02}(\zeta_1, \xi_1+\zeta_2) A^2_{11}(\xi_1, \zeta_2)\bigr\rangle_{S_2^{\zeta}}}
{\lambda_1(\xi_1+\zeta_1+\zeta_2)^3-\lambda_1 \xi_1^3-\lambda_2 (\zeta_1^3+\zeta_2^3)} .
\end{eqnarray*}
The requirement that \(A^1_{12}\) be a polynomial, together with the system being non‑triangular and the assumption $q\notin\{0, e^{\frac{\pi\i}{3}}\}$ leads to $c_1=-d_3$ and $q=\pm \i$ and
\begin{eqnarray*}
A^1_{12}=\frac{5(1-\i)}{36}c_4 d_3 (\xi_1+\zeta_1+\zeta_2).
\end{eqnarray*}

We now compute the coefficients of the nonzero cubic terms in the \(v\)-component, namely \(A^2_{03}\) and \(A^2_{21}\):
\[
A^2_{03}=\frac{\bigl\langle A^2_{11}(\zeta_1+\zeta_2,\zeta_3) a^1_{02}(\zeta_1, \zeta_2)-a^2_{11}(\zeta_1+\zeta_2,\zeta_3) A^1_{02}(\zeta_1, \zeta_2)\bigr\rangle_{S_3^{\zeta}}}{\lambda_2 ((\zeta_1+\zeta_2+\zeta_3)^3-\zeta_1^3-\zeta_2^3-\zeta_3^3)}
=\frac{5(1-\i)}{108}c_4 (2 d_2+ d_3) (\zeta_1+\zeta_2+\zeta_3),
\]
which is a polynomial in its arguments. For $A^2_{21}$, we get
\begin{eqnarray*}
&&A^2_{21}=\frac{A^2_{11}(\xi_1+\xi_2,\zeta_1) a^1_{20}(\xi_1, \xi_2) + \bigl\langle A^2_{11}(\xi_1, \xi_2+\zeta_1) a^2_{11}(\xi_2, \zeta_1)\bigr\rangle_{S_2^{\zeta}}}{\lambda_2(\xi_1+\xi_2+\zeta_1)^3-\lambda_1 (\xi_1^3+\xi_2^3)-\lambda_2 \zeta_1^3}\\
&&\qquad -\frac{a^2_{11}(\xi_1+\xi_2,\zeta_1) A^1_{20}(\xi_1, \xi_2) + \bigl\langle a^2_{11}(\xi_1, \xi_2+\zeta_1) A^2_{11}(\xi_2, \zeta_1)\bigr\rangle_{S_2^{\zeta}}}{\lambda_2(\xi_1+\xi_2+\zeta_1)^3-\lambda_1 (\xi_1^3+\xi_2^3)-\lambda_2 \zeta_1^3} .
\end{eqnarray*}

The requirement that \(A^2_{21}\) be a polynomial leads to either \(d_2=0\) or \(d_2=d_3\), and
\[
A^2_{21}=
\begin{cases}
\dfrac{5(1-\i)}{36}d_3^2 (\xi_1+\xi_2+\zeta_1), & d_2=d_3,\\[4pt]
\dfrac{5(1-\i)}{36}d_3^2 \zeta_1, & d_2=0.
\end{cases}
\]
In addition to computing the Lie bracket between the quadratic terms of the system and the cubic terms of the symmetry, we obtain zero when \(d_2=0\). This leads, up to rescaling, to the Hirota‑Satsuma system \eqref{B_1} in the next section. In the case \(d_2=d_3\), the requirement that \(A^i_{jk}\) for \(i=1,2\) and \(j+k=4\) be polynomials leads to \(d_3=0\), implying that the system is triangular and we discard it.

\end{itemize}
Consider now the special case $q=e^{\frac{\pi{\i}}{3}}$. In this case $\lambda_1=0,\,\lambda_2=1$ and $\frac{\mu_2}{\mu_1}=-9$.
Take $\mu_1=1,\,\mu_2=-9$. Then
\begin{eqnarray*}
&&A^1_{20}=a^1_{20}=0,
\quad
A^2_{02}=-15 (\zeta_1^2+\zeta_1 \zeta_2+\zeta_2^2)a^2_{02},\\
&&A^1_{11}=a^1_{11}=0=A^2_{20}=a^2_{20},\\
&&A^1_{02}=5(2 \zeta_1^2+\zeta_1 \zeta_2 +2 \zeta_2^2) a^1_{02},
\quad
A^2_{11}=-5 (2 \xi_1^2+3 \xi_1 \zeta_1 +3 \zeta_1^2) a^2_{11} .
\end{eqnarray*}
Analysing further higher-degree terms as we did above, we find that if the system is not triangular then, up to rescaling, it is system \eqref{A_1p} in the next section.

\underline{Case 1}: Without loss of generality, we may choose $\lambda_1=5+3\sqrt{5},\,\,\lambda_2=5-3\sqrt{5}$. We impose the existence of a symmetry of order $7$ with $\frac{\mu_2}{\mu_1}=-\frac{1}{2}\left(47+21\sqrt{5}\right)$. 
In this case all quadratic coefficients $A^i_{jk}$ are automatically polynomial, so the system possesses an infinite-dimensional algebra of $2$-approximate symmetries. Requiring polynomiality of all cubic coefficients $A^i_{jk}$ leads, up to rescaling, to the system \eqref{A_1} in the next section.

\underline{Case 10}: In this case it is sufficient to impose the existence of a symmetry of order $5$ with arbitrary ratio $\frac{\mu_2}{\mu_1}$. There is no relation between ratios $\frac{\lambda_2}{\lambda_1}$ and $\frac{\mu_2}{\mu_1}$. Polynomiality of the quadratic coefficients immediately gives $c_2=c_3=c_4=d_1=d_2=d_3=0$, which implies that the system is linear. Hence Case~10 yields no non-triangular integrable systems.

This classification problem reduces to Cases~1, 10, and~13 of Theorem~\ref{thmmn}. Case~10 produces only triangular systems. Case~1 yields, up to rescaling, the system~\eqref{A_1}. Case~13 gives rise to two non-triangular systems: the exceptional system~\eqref{A_1p}, corresponding to $q=e^{\frac{\pi\i}{3}}$, and the Hirota--Satsuma system~\eqref{B_1}, corresponding to $q=\pm\i$. Thus, up to rescaling and the interchange $u\leftrightarrow v$, the family~\eqref{sysex} contains precisely these three non-triangular systems. Their integrability is guaranteed by Theorem~\ref{bbt}.

\section{Classification results of integrable homogeneous systems of orders 3 and 5}\label{secclass}
In this section, we present a complete list of homogeneous polynomial two-component systems of order  $n=3$ and $n=5$ of the form
\begin{equation}
\label{sysg}
\left\{
\begin{array}{l}
u_t=\lambda_1 u_n+F_1(u_{n-1},v_{n-1},\ldots,u,v),\\
v_t=\lambda_2 v_n+F_2(u_{n-1},v_{n-1},\ldots,u,v),
\end{array}\right. 
\end{equation}
which admit an infinite hierarchy of symmetries 
\begin{equation}
\label{symg}
\left\{
\begin{array}{l}
u_{t_m}=\mu_{1}(m) u_m+G_1(u_{m-1},v_{m-1},\ldots,u,v),\\
v_{t_m}=\mu_{2}(m) v_m+G_2(u_{m-1},v_{m-1},\ldots,u,v),
\end{array}\right. 
\end{equation} 
of odd order $m$, subject to the following assumptions:
\begin{itemize}
\item[(I).] The system (\ref{sysg}) is homogeneous with \emph{positive} weights $W(u),\,W(v)>0$;
\item[(II).] Polynomials $F_1, F_2\in \bigoplus_{k\ge 2}\cR^k$, i.e., they do not contain linear terms; 
\item[(III).] The system is non-triangular (see Definition \ref{triangular}); 
\item[(IV).] The system does not admit non-point symmetries of the form (\ref{symg}) of order $m<n$;
\item[(V).] The parameters $\lambda_1\lambda_2\ne 0$ and $\lambda_1 \ne \lambda_2$.
\end{itemize}

Condition~(IV) excludes systems that are members of a commuting hierarchy of symmetries generated by a system of lower order. In Section~\ref{eigzero}, we consider the degenerate case in which one of parameters vanishes by replacing Condition~(V) with the following assumption:

\medskip
\noindent
(${\rm V'}$). $\lambda_1=0,\ \lambda_2\ne 0$, and there exists an infinite subsequence of symmetries (\ref{symg}) with $\mu_1(m)\mu_2(m)\ne 0.$

The group of invertible transformations that preserves the form of equations~(\ref{sysg}) and Conditions~(I)-(II) is generated by scalings of the dependent and independent variables together with the involution
\begin{equation}\label{scal}
u\mapsto\alpha_1 u,\quad v\mapsto\alpha_2 v,\quad x\mapsto \alpha_3 x,\quad t\mapsto\alpha_4 t,\qquad u\leftrightarrow v,
\end{equation}
where $\alpha_i$ are arbitrary nonzero complex constants. Systems related by these transformations are considered equivalent. The classification results presented in Sections~\ref{3order} and~\ref{5thord} are given modulo this equivalence relation.

Conditions~(I) and~(II) imply that $F_1$ and $F_2$ are homogeneous differential polynomials whose admissible forms are determined by the choice of weights $W(u)$ and $W(v)$. For instance, in Example~\ref{example22}, where $W(u)=W(v)=2$, the most general homogeneous third-order system
\[
u_t=\lambda_1 u_3+F_1,\qquad v_t=\lambda_2 v_3+F_2,
\]
contains eight free parameters once the values of $\lambda_1$ and $\lambda_2$ are fixed according to Table~\ref{table2}.

The possible choices of weights that lead to non-triangular homogeneous polynomial systems are finite. They are obtained by balancing the linear terms with the quadratic and, when necessary, cubic terms in the nonlinearities. For each admissible set of weights, we construct the most general homogeneous polynomials $F_1$ and $F_2$, while excluding triangular systems as required by Condition~(III).
For every remaining family, the integrability conditions are analysed following the procedure illustrated in Example~\ref{example22}. During this process, cases corresponding to triangular systems are discarded. Finally, from the resulting list of integrable equations, we remove those admitting lower-order symmetries, as required by Condition~(IV). The complete classification results for integrable homogeneous systems of orders $3$ and $5$ are presented in the following sections.

\subsection{Integrable systems of form \eqref{sysg} with $n=3$}\label{3order}
 By Theorem \ref{thmmn}, the classification of all third-order integrable systems of the form \eqref{sysg} is limited to Cases 1, 10, and 13 as shown in Table \ref{table2}. In each of these cases, it follows from Theorem \ref{bbt} that it is sufficient to verify the existence of a symmetry of order $m_2$, the second member of the corresponding sequence $S$.

{{\bf Case 1}}:
We have symmetries of order $m=1,3,7,9\mod 10$ with spectral invariant 
$$
\frac{\lambda_2}{\lambda_1}=\frac{(1+q)^3}{1+q^3},\quad\frac{\mu_2(m)}{\mu_1(m)}=\frac{(1+q)^m}{1+q^m},\quad q=e^{\frac{2\pi \i}{5}}.
$$

Without loss of generality, we choose $\lambda_1=5+3\sqrt{5}$ and $\lambda_2=\overline{\lambda}=5-3\sqrt{5}$. To classify integrable systems of this type, it is sufficient to impose the existence of a symmetry of order $7$ as derived from Theorem \ref{bbt}. These integrable systems are labelled by the letter ${\rm A}$, as they possess Lax representations in terms of the affine Lie algebra $A_4^{(2)}$.

{\bf Case 13}:
The symmetries are of order $m=1\mod 2$ with spectral invariant
$$
\frac{\lambda_2}{\lambda_1}=\frac{(1+q)^3}{1+q^3},\quad\frac{\mu_2(m)}{\mu_1(m)}=\frac{(1+q)^m}{1+q^m}.
$$
To classify systems of this type, it suffices to impose the existence of a symmetry of order $5$. The value of $q$ is determined by cubic integrability conditions, leading to three cases: $q=\i$, $q=1$, or an arbitrary $q\in\C\setminus\{0,-1\}$. The integrable
systems in the first case are labelled by the letter ${\rm B}$
as they possess Lax representations in terms of Lie algebra $B_2^{(1)}$. The integrable systems in the second case are labelled by the letter $K$ and are related to a triangular extension of the KdV equation through Miura type transformations. In the third case, the integrable systems are labelled by the letter ${\rm L}$.

{\bf Case 10}: We have symmetries of order $m=1\mod 2$ with arbitrary ratios $\frac{\lambda_2}{\lambda_1}$ and $\frac{\mu_2}{\mu_1}$.
To classify systems of this type, it is sufficient to impose the existence of a symmetry of order $5$. These integrable systems are labelled by the letter ${\rm M}$.
 
\begin{The}\label{thmA}
Let the system~\eqref{sysg} with $n=3$ satisfy Conditions~(I)--(V) and assume that it is integrable. Then, up to the transformations~\eqref{scal}, the system is equivalent to one of the following integrable systems:

\medskip

\noindent
{\bf i. Case 1},
$\qquad
\lambda_1=\lambda=5+3\sqrt{5},
\qquad
\lambda_2=\bar\lambda=5-3\sqrt{5},\qquad q=e^{\frac{2\pi \i}{5}},\qquad
$  type {\rm A}.

\medskip

\underline{\(W(u)=2,\ W(v)=2\)}:
\begin{align}
\left\{
\begin{array}{l}
u_t=\lambda u_3+6uu_1-(\lambda+4)vu_1-6uv_1+(\lambda-8)vv_1,\\[2mm]
v_t=\bar\lambda v_3+6vv_1-(\bar\lambda+4)uv_1-6vu_1+(\bar\lambda-8)uu_1,
\end{array}
\right.
\tag{$\rm A_1$}\label{A_1}
\end{align}

\underline{\(W(u)=1,\ W(v)=1\)}:
\begin{align}
&\left\{
\begin{array}{l}
u_t=\lambda u_3
+\left[
5(\lambda-2)(u+v)v_1
+10(\lambda-5)u^2v
-5(\lambda+10)uv^2
-5\lambda v^3
\right]_x,\\[2mm]
v_t=\bar\lambda v_3
+\left[
5(\bar\lambda-2)(u+v)u_1
+10(\bar\lambda-5)v^2u
-5(\bar\lambda+10)vu^2
-5\bar\lambda u^3
\right]_x,
\end{array}
\right.
\tag{$\rm A_2$}\label{A_2}
\\[4mm]
&\left\{
\begin{array}{l}
u_t=\lambda u_3+(\lambda+1)u_1^2
-6uv_2-6u_1v_1+6vv_2+(17-\lambda)v_1^2\\
\qquad
+\dfrac{2}{5}(\lambda-5)(u-v)^2u_1
-\dfrac{2}{5}(\lambda-20)(u-v)^2v_1
+\dfrac{3}{5}(u-v)^4,\\[2mm]
v_t=\bar\lambda v_3+(\bar\lambda+1)v_1^2
-6u_2v-6u_1v_1+6uu_2+(17-\bar\lambda)u_1^2\\
\qquad
+\dfrac{2}{5}(\bar\lambda-5)(u-v)^2v_1
-\dfrac{2}{5}(\bar\lambda-20)(u-v)^2u_1
+\dfrac{3}{5}(u-v)^4,
\end{array}
\right.
\tag{$\rm A_3$}\label{A_3}
\\[4mm]
&\left\{
\begin{array}{l}
u_t=\lambda u_3
-(\lambda-2)\left(u_1^2-uv_2-u_1v_1+vv_2\right)
-6v_1^2\\
\qquad
-\dfrac{2}{5}(\lambda-5)(u-v)^2u_1
+\dfrac{2}{5}(\lambda+10)(u-v)^2v_1
-\dfrac{6}{5}(u-v)^4,\\[2mm]
v_t=\bar\lambda v_3
-(\bar\lambda-2)\left(v_1^2-u_2v-u_1v_1+uu_2\right)
-6u_1^2\\
\qquad
-\dfrac{2}{5}(\bar\lambda-5)(u-v)^2v_1
+\dfrac{2}{5}(\bar\lambda+10)(u-v)^2u_1
-\dfrac{6}{5}(u-v)^4,
\end{array}
\right.
\tag{$\rm A_4$}\label{A_4}
\end{align}
\begin{align}
&\left\{
\begin{array}{l}
u_t=\lambda u_3+(\lambda-8)u_1^2
+2(\lambda+4)\left(uv_2+u_1v_1-vv_2\right)
-3(\lambda+2)v_1^2\\
\qquad
+\dfrac{8}{5}(\lambda-5)(u-v)^2u_1
-\dfrac{4}{5}(2\lambda+5)(u-v)^2v_1
-\dfrac{6}{5}(u-v)^4,\\[2mm]
v_t=\bar\lambda v_3+(\bar\lambda-8)v_1^2
+2(\bar\lambda+4)\left(u_2v+u_1v_1-uu_2\right)
-3(\bar\lambda+2)u_1^2\\
\qquad
+\dfrac{8}{5}(\bar\lambda-5)(u-v)^2v_1
-\dfrac{4}{5}(2\bar\lambda+5)(u-v)^2u_1
-\dfrac{6}{5}(u-v)^4,
\end{array}
\right.
\tag{$\rm A_5$}\label{A_5}
\end{align}

\medskip

\noindent
{\bf ii. Case 13}, $\qquad 
\lambda_2=-2\lambda_1,
\qquad
q={\i},\qquad
$ type {\rm B}.

\medskip

\underline{\(W(u)=2,\ W(v)=2\)}:
\begin{align}
\left\{
\begin{array}{l}
u_t=u_3+uu_1-vv_1,\\[2mm]
v_t=-2v_3-uv_1,
\end{array}
\right.
\tag{$\rm B_1$}\label{B_1}
\end{align}

\underline{\(W(u)=2,\ W(v)=1\)}:
\begin{align}
\left\{
\begin{array}{l}
u_t=u_3+uu_1+6v_1v_2-v^2u_1,\\[2mm]
v_t=-2v_3-uv_1-vu_1+v^2v_1,
\end{array}
\right.
\tag{$\rm B_2$}\label{B_2}
\end{align}

\underline{\(W(u)=1,\ W(v)=2\)}:
\begin{align}
\left\{
\begin{array}{l}
u_t=u_3+u_1^2-v^2,\\[2mm]
v_t=-2v_3-2u_1v_1,
\end{array}
\right.
\tag{$\rm B_3$}\label{B_3}
\end{align}

\underline{\(W(u)=1,\ W(v)=1\)}:
\begin{align}
&\left\{
\begin{array}{l}
u_t=u_3+u_1^2+3v_1^2+6vv_2+4v^2u_1+2v^4,\\[2mm]
v_t=-2v_3-2u_2v-2u_1v_1-4v^2v_1,
\end{array}
\right.
\tag{$\rm B_4$}\label{B_4}
\\[4mm]
&\left\{
\begin{array}{l}
u_t=u_3-\left[6vv_1+2u(u^2-3v^2)\right]_x,\\[2mm]
v_t=-2v_3+\left[6vu_1+2v(3u^2-v^2)\right]_x,
\end{array}
\right.
\tag{$\rm B_5$}\label{B_5}
\\[4mm]
&\left\{
\begin{array}{l}
u_t=u_3-3u_1^2+6v_1^2-12(u-v)^2v_1+6(u-v)^4,\\[2mm]
v_t=-2v_3+6uu_2+3u_1^2-6u_2v-6u_1v_1
+6v_1^2-12(u-v)^2u_1+6(u-v)^4.
\end{array}
\right.
\tag{$\rm B_6$}\label{B_6}
\end{align}
\medskip

\noindent
{\bf iii. Case 13},
$\qquad
 \lambda_2=4\lambda_1,\qquad  q=1,\qquad
$ type {\rm K}.
 \medskip

\underline{\(W(u)=1,\ W(v)=1\)}:
\begin{align}
&\left\{ \begin{array}{l}
u_t=u_3-3u_1^2-6vv_2+6v_1^2+3v^4,\\
v_t=4v_3-3vu_2-6u_1v_1-12v^2v_1;
\end{array}\right. \tag{$\rm K_1$}\label{K_1}\\[4mm]
&\left\{ \begin{array}{l} u_t=u_3- 6 v_1^2 + 6vv_2-6 u^2 u_1   + 18 v^2u_1 + 12 u v v_1 + 24 v^2 v_1 ,\\
v_t=4 v_3+ 6 vu_2+ 18 u_1 v_1+ 12 u v_2+ 
 12 v v_2+ 24 v_1^2 +12 uv u_1  + 6 u^2 v_1  + 
 24 u v v_1 + 6 v^2 v_1.
\end{array}\right. \tag{$\rm K_2$} \label{K_2}
\end{align}
\medskip

\noindent
{\bf iv. Case 13},
$\qquad
\lambda_1=1,\qquad \lambda_2=\frac{(1+q)^3}{1+q^3}=:\kappa,\qquad q\in\C\setminus\{0,-1\},\qquad
$ type {\rm L}
\medskip

\underline{\(W(u)=1,\ W(v)=\frac{3}{2}\)}:
\begin{align}
&\left\{ \begin{array}{l}
u_t=u_3+(6uu_1+\kappa v^2+4u^3)_x,\\
v_t=\kappa v_3+(\kappa-1) u_2 v+3\kappa uv_2+3 \kappa u_1 v_1+3(\kappa-2)uvu_1+3\kappa u^2v_1-\kappa v^3+(\kappa-4)u^3v;
\end{array}\right.  \tag{$\rm L_1$} \label{L_1}
\end{align}

\underline{\(W(u)=1,\ W(v)=\frac{1}{2}\)}{\rm:}
\begin{align}
&\left\{ \begin{array}{l}
u_t=u_3 + 6 u u_2 + 6 u_1^2 + 6 \kappa  v_1 v_2 + 
 12 u^2 u_1 + 6 \delta v^2 u_2 + 
 2 (3 \kappa  + \delta (\kappa  + 2)) u v_1^2 - 
 2 (2 \delta (\kappa  - 1) - 3 \kappa ) u v v_2 \\ \qquad
 +6 (2 \delta + \kappa ) v u_1  v_1 - 
 2 (2 \delta (\kappa  - 4) - 3 \kappa ) u^2 v v_1 - 
 2 (2 \delta (\kappa  - 7) - 3 \kappa ) v^2 u u_1 + 
 6 \delta \kappa  v^3 v_2 + 
 24 \delta \kappa  v^2 v_1^2\\ \qquad
 -12 \delta (\delta (\kappa  - 2) - 3 \kappa ) u v^3 v_1 + 6 \delta (2 \delta + \kappa ) v^4 u_1 - 
 2 \delta (\kappa  - 4) u^3 v^2 + 
 30 \delta^2 \kappa  v^5 v_1\\ \qquad - 
 2 \delta (-3 \kappa  + 2 \delta (\kappa  - 4)) u^2 v^4 - 2 \delta^2 (-6 \kappa  + \delta (\kappa  - 4)) u v^6 + 6 \delta^3 \kappa  v^8,
\\
v_t=\kappa v_3 + (\kappa - 1) v u_2 + 3 \kappa u v_2 + 
 3 \kappa u_1 v_1 + 3 (\kappa - 2) u v u_1 + 
 3 \kappa u^2 v_1 + 
 3 (3 \delta - 1) \kappa v v_1^2 + 
 3 \delta \kappa v^2 v_2\\ \qquad + 
 2 (-2 \delta - 3 \kappa + 5 \delta \kappa) u v^2 v_1 + (\kappa - 4) u^3 v + \delta (\kappa - 4) v^3 u_1 + (-3 \kappa + 2 \delta (\kappa - 4)) u^2 v^3 \\ \qquad+ 
 3 \delta (\delta - 2) \kappa v^4 v_1 + \delta (-6 \kappa + \delta (\kappa - 4)) u v^5 - 3 \delta^2 \kappa v^7,
\end{array}\right.  \tag{$\rm L_2$} \label{L_2}
\end{align}
where   $\delta\in\C\setminus\{1\}$ is an arbitrary constant.
\medskip

\noindent
\newpage
{\bf  v. Case 10},
$\qquad
\lambda_1=1,\qquad \lambda_2=\kappa\in\C\setminus\{0,1\},\qquad
$ type {\rm M}.
\medskip

\underline{\(W(u)=1,\ W(v)=1\)}:
\begin{align}
&\left\{ \begin{array}{l}
u_t=u_3+(\kappa-1) \, [2(u+v)v_2+6u^2v_1-12uvv_1-10v^2v_1  +(u-v)^2(u+3v)(v+3u)] \, +6uu_2\\
\qquad +9u_1^2 -6 u_2 v -6 u_1 v_1 +(\kappa+2)\left[-v_1^2+4(u^2+v^2)u_1\right]+8(\kappa-4)uvu_1,\\
v_t=\kappa v_3+(\kappa-1) \, [2(u+v)u_2+10u^2u_1+12uvu_1-6v^2u_1  +(u-v)^2(u+3v)(v+3u)] \, \\
\qquad+(2\kappa+1)\left[u_1^2+4(u^2+v^2)v_1\right]  + 3\kappa\left(2uv_2+2 u_1 v_1-2vv_2-3v_1^2\right)-8(4\kappa-1)uvv_1;
 \end{array}\right.  \tag{$\rm M_1$} \label{M_1}
\end{align}

\underline{\(W(u)=\frac{1}{2},\ W(v)=\frac{1}{2}\)}:
\begin{align}
&\left\{ \begin{array}{l}
u_t=u_3 + 3 u^2 u_2 + 9 u u_1^2 + (\kappa+ 2) u v_1^2 - 
 2 (\kappa- 1) u v v_2 + 3 v^2 u_2 + 
 6 u_1 v v_1 + 3 u^4 u_1 - 
 2 (\kappa- 1) u^3 v v_1\\ \qquad - 
 2 (2 \kappa- 5) u^2 v^2 u_1 + 3 v^4 u_1 - 
 6 (\kappa- 1) u v^3 v_1 - (\kappa- 1) u v^2 (u^2 + v^2)^2,\\
v_t=\kappa v_3 + 
 3 \kappa u^2 v_2 + (2 \kappa+ 1) v u_1^2 + 
 2 (\kappa- 1) v u u_2 + 6 \kappa v_1 u u_1 + 
 3 \kappa v^2 v_2 + 9 \kappa v v_1^2 + 
 3 \kappa u^4 v_1 + 3 \kappa v^4 v_1 \\ \qquad+ 6 (\kappa- 1) u^3 v u_1 + 
 2 (\kappa- 1) u v^3 u_1 + 
 2 (5 \kappa- 2) u^2 v^2 v_1+ (\kappa- 1) u^2 v (u^2 + v^2)^2.
 \end{array}\right.  \tag{$\rm M_2$} \label{M_2}
\end{align}
\end{The}
\begin{Rem}
In the integrable family \eqref{L_2}, the value $\delta=1$ is excluded. Indeed, for this value the system admits a second-order generalised symmetry:
\begin{equation}\label{2ndsym}
 \left\{ \begin{array}{l}  u_{\tau}=(2+\kappa)u_2+4(\kappa+2) \left(u u_1+u_1 v^2+u v v_1 + u v^4\right)+6 \kappa \left(v_1^2+2 v^3 v_1+v^6 \right)+2 (4-\kappa)u^2v^2 ,
 \\
 v_{\tau}=3\kappa v_2+2(\kappa-1) u_1 v +6 \kappa u v_1+6 \kappa v^2 v_1+(\kappa -4) u^2v-2 (\kappa+2) u v^3
 -3 \kappa v^5 .
  \end{array}\right. 
\end{equation}
Therefore, Condition~(IV) is violated since it belongs to a hierarchy generated by a  second-order symmetry.
\end{Rem}

We compare our results with those available in the literature.
Foursov studied two-component homogeneous evolutionary systems with $W(u)=W(v)=2$ of orders $2$, $3$ and $5$ \cite{Foursov03, Foursov} assuming the existence of symmetries of ad-hoc prescribed  orders. In the third-order case, two of the systems obtained satisfy the assumptions imposed in the present paper, and they are equivalent, under linear transformations of the dependent variables, to systems~\eqref{B_1} and~\eqref{A_1}.

Two-component integrable systems of the form
\begin{equation*}\label{a3eq}
 u_t=\lambda u_3+F(u, v, u_1, v_1, u_2, v_2), \qquad v_t=v_3+G(u, v, u_1, v_1, u_2, v_2),\qquad \lambda\in\C,   
\end{equation*}
have been studied by Meshkov and collaborators using the canonical densities approach for specific values of the parameter $\lambda$ \cite{meshkov94}.
The results obtained in \cite{Mesh08, MeshB08, Balak23} are closely related to the A-type and B-type systems appearing in Theorem~\ref{thmA}.
In particular, \cite{MeshB08} provides a list of integrable systems in the case $\lambda=-\tfrac12$, which includes all B-type systems in Theorem~\ref{thmA}. The case $\lambda=\tfrac{\pm 3\sqrt{5}-7}{2}$ was studied in \cite{Balak23}. However, the integrable systems \eqref{A_3}--\eqref{A_5} do not appear in these classifications and were not identified therein.

\subsection{Integrable systems of form \eqref{sysg} with $n=5$}\label{5thord}
By Theorem \ref{thmmn}, the classification of all integrable systems of order $5$ of the form \eqref{sysg} reduces to the analysis of Cases 4, 11, 14, 16, 18, 22, and 23 in Table \ref{table2}. As in the case $n=3$, for each case we verify the existence of a symmetry of order $m_2$, namely the second member of the corresponding sequence $S$.

{\bf Case 4}: In this case, the system admits symmetries of order $m=1,5,7,11\mod 12$, or equivalently, $m=1,5\mod 6$, with spectral invariant
$$
\frac{\lambda_2}{\lambda_1}=\frac{(1+q)^5}{1+q^5},\quad\frac{\mu_2(m)}{\mu_1(m)}=\frac{(1+q)^m}{1+q^m},\quad q=e^{\frac{\pi \i}{6}}.
$$
Without loss of generality, we choose $\lambda_1=9+5\sqrt{3}$, $\lambda_2=\overline{\lambda}=9-5\sqrt{3}$. To classify systems in this case, it is sufficient to impose the existence of a symmetry of order $7$. Integrable systems of this type possess Lax representations in terms of the Lie algebra $D_4^{(3)}$, and we denote them by the letter $D$.

{\bf Case 11}: The ratios $\frac{\lambda_2}{\lambda_1}$ and $\frac{\mu_2(m)}{\mu_1(m)}$ are arbitrary. An integrable system in this case admits symmetries of order $1,5\mod 6$. The classification procedure shows that all integrable systems arising are triangular. Hence, they are excluded from the classification under Condition~(III).

{\bf Case 14}: In this case, the system admits symmetries of order $m=1,5\mod 6$ with spectral invariant
$$
\frac{\lambda_2}{\lambda_1}=\frac{(1+q)^5}{1+q^5},\quad\frac{\mu_2(m)}{\mu_1(m)}=\frac{(1+q)^m}{1+q^m}.
$$
Again, it is sufficient to impose the existence of a symmetry of order $7$. The cubic integrability conditions determine the value of $q$, leading to $q=e^{\frac{\pi \i}{3}}$, and the minimal order of the recursion operator is $6$. Integrable systems of this type possess Lax representations in terms of the Lie algebra $G_2^{(1)}$, and we label them with the letter $G$.

{\bf Case 16}:  In this case, the system admits symmetries of order $m=1\mod 4$ with spectral invariant 
 $$
\frac{\lambda_2}{\lambda_1}=\frac{(1+q)^5}{1+q^5},\quad\frac{\mu_2(m)}{\mu_1(m)}=\frac{(1+q)^m}{1+q^m},\quad q=\i.
$$
We choose  $\lambda_1=1,\,\,\lambda_2=-4$. To classify systems of this type, it is sufficient to impose the existence of a symmetry of order $9$.  Integrable systems of this type can be transformed into triangular systems by differential substitutions.

{\bf Case 18}: Here the system admits symmetries of order $m=1,5\mod 12$ with spectral invariant
 $$
\frac{\lambda_2}{\lambda_1}=\frac{(1+q)^5}{1+q^5},\quad\frac{\mu_2(m)}{\mu_1(m)}=\frac{(1+q)^m}{1+q^m}, \quad q=\i.
$$
Without loss of generality, we choose 
$\lambda_1=1,\,\,\lambda_2=-4$. It is sufficient to impose the existence of a symmetry of order $13$. However, these systems also admit symmetries of order $9$ and therefore are a subcase of Case 16. 

{\bf Case 22}: In this case, the system admits symmetries of order $m=5\mod 10$ with spectral invariant
$$
\frac{\lambda_2}{\lambda_1}=\frac{(1+q)^5}{1+q^5},\quad\frac{\mu_2(m)}{\mu_1(m)}=\frac{(1+q)^m}{1+q^m},
$$
where $q=\alpha\frac{\beta-1}{\alpha-1}$, $\,\,\alpha^5=\beta^5=1,\, \beta\ne1$ and $\,\,\alpha\notin \{1, \beta,\beta^{-1}\}$.
We choose $\lambda_1=1,\,\,\lambda_2=-\frac{13\pm 5\sqrt{5}}{22}$.
The existence of a symmetry of order $15$ is sufficient for the classification. These systems can be reduced to triangular ones by differential substitutions. 

{\bf Case 23}: In this final case, the system admits symmetries of order $m=5,25\mod 30$, with spectral invariant
 $$
\frac{\lambda_2}{\lambda_1}=\frac{(1+q)^5}{1+q^5},\quad\frac{\mu_2(m)}{\mu_1(m)}=\frac{(1+q)^m}{1+q^m},
$$
where the value of $q$ is the same as Case 22. Thus we also choose $\lambda_1=1,\,\,\lambda_2=-\frac{13\pm 5\sqrt{5}}{22}$. 
It is sufficient to impose the existence of a symmetry of order $25$. Since these systems also possess symmetries of order $15$, they are contained in Case~22.

The integrable systems corresponding to Cases~$16, 18, 22$ and $23$ are not included in the classification theorem below. They will be discussed separately in Section~\ref{sec64}.

\begin{The}\label{thmD}
Let system~\eqref{sysg} with $n=5$ be integrable, satisfy conditions~(I)--(V), and not belong to Cases $16, 18, 22$ and $23$. Then, under the transformations~\eqref{scal}, it can be reduced to one of the following equations:
\begin{enumerate}
    \item[{\rm\bf i.}]  ${\rm\bf Case\ 4},\qquad  
    \lambda_1=\lambda=9+5\sqrt{3},\qquad \lambda_2=\bar\lambda=9-5\sqrt{3} ,\qquad q=e^{\frac{\i\pi}{6}},
    \qquad$  type ${\rm D}$.

\underline{$W(u)=2,\ W(v)=2$}:
\begin{align}
&\left\{ \begin{array}{l}
u_t=F(\left[u\right],[v],\lambda):=\lambda u_5 + 3 (\lambda - 4) u u_3 - 15 u_1 u_2 - 
 3 (\lambda + 6) u v_3 - 3 (4 \lambda - 1) u_1 v_2 \\ \qquad- 
 15 (\lambda - 1) u_2 v_1- 3 (3 \lambda - 2) v u_3 + 
 3 (\lambda - 4) v v_3 + 3 (3 \lambda - 7) v_1 v_2 - 
 3 (\lambda - 9) u^2 u_1  \\ \qquad
 - 6 (2 \lambda - 3) (u^2 + v^2) v_1 - 
 6 (2 \lambda - 3) u v u_1 + 
 15 (\lambda + 3) v^2 u_1 + 24 (\lambda + 6) u v v_1,\\
v_t=F([v], [u], \lb);
\end{array}\right. \tag{$\rm D_1$} \label{D_1}
\end{align}

\underline{$W(u)=1,\ W(v)=1$}:
\begin{align}
&\left\{ \begin{array}{l} u_t =F([u],[v],\lambda):=\lambda  u_5 + 
 2 (4 \lambda  + 9) u_1 u_3 + (7 \lambda  + 12) u_2^2 + 
 2 (\lambda  + 6) (u - v) v_4 + 6 (\lambda  + 1) u_1 v_3\\ \qquad + 
 2 (2 \lambda  - 3) (u_3 v_1 + 2 u_2 v_2) - 
 2 (7 \lambda  - 3) v_1 v_3 - (11 \lambda  + 6) v_2^2 + 
 4 (\lambda  - 9) (u - v)^2 u_3\\ \qquad - 
 12 (\lambda  + 1) (u - v)^2 v_3 + 
 20 (\lambda  - 9) (u - v) u_1 u_2 - 
 20 (\lambda  - 3) (u - v) v_1 u_2\\ \qquad - 
 12 (3 \lambda  - 7) (u - v) u_1 v_2 + 
 4 (23 \lambda  + 33) (u - v) v_1 v_2 + 
 4 (\lambda  + 11) u_1^3 - 8 (4 \lambda  - 1) v_1^3 \\ \qquad- 
 8 (\lambda  + 6) u_1^2 v_1 + 36 (\lambda  + 1) u_1 v_1^2 - 
 120 (u - v)^3 u_2 + 
 8 (4 \lambda  + 9) (u - v)^3 v_2 \\ \qquad+ 
 4 (4 \lambda  - 141) (u - v)^2 u_1^2 + 
 24 (4 \lambda  - 1) (u - v)^2 u_1 v_1 - 
 4 (28 \lambda  + 33) (u - v)^2 v_1^2\\ \qquad + 
 48 (\lambda  - 9) (u - v)^4 u_1 - 
 48 (\lambda  + 1) (u - v)^4 v_1 - 80 (u - v)^6, \\
v_t=F([v], [u], \lb);
 \end{array}\right. \tag{$\rm D_2$} \label{D_2}\\[4mm]
&\left\{ \begin{array}{l} u_t =F([u], [v], \l):=\lambda u_5 + (\lambda + 1) (3 u_1 - v_1) u_3 - 
 2 (\lambda + 1) u_2 v_2 + 
 \frac{5 \lambda}{2} u_2^2 - (\lambda - 4) (u - v) v_4\\ \qquad - (2 \lambda - 3) u_1 v_3 + (2 \lambda + 7) v_1 v_3 + 
 \frac{1}{2} (3 \lambda - 2) v_2^2 - 
 \frac{1}{3} (\lambda - 9) (u - v)^2 u_3 + (\lambda + 1) (u - v)^2 v_3\\ \qquad - 
 \frac{1}{3} (4 \lambda - 21) (u - v) u_1 u_2 + 
 \frac{1}{3} (4 \lambda + 9) (u - v) v_1 u_2 - 
 \frac{7}{3} (2 \lambda - 3) (u - v) v_1 v_2\\ \qquad + 
 \frac{1}{3} (8 \lambda - 27) (u - v) u_1 v_2 + 
 5 u_1^3 + (\lambda + 6) v_1^3 + (\lambda - 4) u_1^2 v_1 - (2 \lambda - 3) u_1 v_1^2 + 5 (u - v)^3 u_2\\ \qquad - 
 \frac{1}{3} (2 \lambda + 7) (u - v)^3 v_2 + 
 \frac{1}{6} (2 \lambda + 27) (u - v)^2 u_1^2 + 
 \frac{5}{6} (2 \lambda + 3) (u - v)^2 v_1^2\\ \qquad  - (2 \lambda - 
    3) (u - v)^2 u_1 v_1 - 
 \frac{1}{3} (\lambda - 9) (u - v)^4 u_1 + 
 \frac{1}{3} (\lambda + 1) (u - v)^4 v_1 + \frac{5}{18} (u - v)^6 , \\
v_t=F([v], [u], \lb);
 \end{array}\right. \tag{$\rm D_3$} \label{D_3}\\[4mm]
&\left\{ \begin{array}{l} u_t =F([u], [v], \l):=\lambda  u_5 + (7 \lambda  - 3) u_1 u_3 - 
 \frac{5\lambda}{2}  u_2^2 + (\lambda  - 24) (u - v) v_4 - 
 15 u_1 v_3 \\ \qquad+ (309 - 16 \lambda ) v_1 v_3 - (\lambda  - 9) (2 u_2 v_2 + v_1 u_3) + 
 \frac{1}{2} (234 - 11 \lambda ) v_2^2 + (\lambda  - 9) (u - v)^2 u_3\\ \qquad - 3 (\lambda  - 19) (u - v)^2 v_3  +(2 \lambda  - 33) (u - v) (u_1 u_2 - 
    3 u_1 v_2) - ((2 \lambda  - 123) (u - v)) v_1 u_2\\ \qquad - (4 \lambda  - 51) (u - v) v_1 v_2  - (2 \lambda  + 7) u_1^3+ (7 \lambda  - 118) v_1^3 - 
 5 (\lambda  - 24) u_1^2 v_1 + 15 u_1 v_1^2\\ \qquad + 
 15 (u - v)^3 u_2 - (2 \lambda  - 33) (u - v)^3 v_2 - 
 \frac{1}{2} (14 \lambda  - 51) (u - v)^2 u_1^2 + 
 \frac{1}{2} (26 \lambda  - 489) (u - v)^2 v_1^2\\ \qquad  - 
 3 (2 \lambda  - 43) (u - v)^2 u_1 v_1 - 
 3 (\lambda  - 9) (u - v)^4 u_1 + 
 3 (\lambda  - 19) (u - v)^4 v_1 - \frac{5}{2} (u - v)^6, \\
v_t =F([v], [u], \lb);
 \end{array}\right. \tag{$\rm D_4$} \label{D_4}\\[4mm]
&\left\{ \begin{array}{l} u_t =F([u], [v], \l):=\l u_5 + \bigg(15 \l u_1 u_2 + 30 (u - v) v_3 + 
3 (\l - 4) (u_1 v_2 + u_2 v_1) - 3 (\l + 16) v_1 v_2\\
\qquad - 12 (\l + 6) (u^2 + v^2) u_2 - 
12 (\l - 9) u v u_2 - 18 (\l - 4) u^2 v_2 + 18 (\l - 14) v^2 v_2\\ \qquad - 12 (\l + 6) u u_1^2
 - 6 (\l - 9) v u_1^2 - 12 (\l - 9) u u_1 v_1 - 24 (\l + 6) v u_1 v_1 - 12 (\l - 24) u v_1^2
\\ \qquad + 30 (\l - 15) v v_1^2
 + 36 (u - v) \left((2 \l - 13) u^2 + 20 u v + (23 - 2 \l) v^2\right) v_1 + 270 u^5\\ \qquad - 
   1080 u^4 v + 540 u^3 v^2 + 2160 u^2 v^3 + 270 u v^4 - 216 v^5\bigg)_x, \\
v_t=F([v], [u], \lb).
 \end{array}\right. \tag{$\rm D_5$} \label{D_5}
\end{align}

\item[{\rm\bf ii.}]  ${\rm\bf Case\ 14},\qquad  
    \lambda_1=1,\qquad \lambda_2=-9,\qquad q=e^{\frac{\i\pi}{6}},
    \qquad$  type ${\rm G}$.
\medskip

\underline{$W(u)=2,\ W(v)=2$}:
\begin{align}
&\left\{ \begin{array}{l} u_t=u_5+10 u u_3 +25 u_1 u_2 -10 v v_3-15 v_1 v_2 +20 u^2 u_1-10 u_1 v^2\\
\qquad -20 u v v_1+10 v^2 v_1, \\
v_t=-9v_5-10u_3v-35u_2v_1-45u_1v_2-30uv_3+30v_1v_2-20u u_1 v\\
\qquad -20u^2v_1-10 u_1 v^2+10v^2v_1 ; \end{array}\right.\tag{$\rm G_1$} \label{G_1}
\end{align}

\underline{$W(u)=2,\ W(v)=1$}:
\begin{align}
&\left\{ \begin{array}{l} 
u_t=u_5 + 5 u u_3 + 5 u_1 u_2 + 5 u^2 u_1 + 10 v_1 v_4 + 
 15 v_2 v_3 - \frac{5}{2} v^2 u_3 + 5 v u_1 v_2 + 10 u_1 v_1^2 - 
 \frac{10}{3} v^2 v_4\\
 \qquad + 10 u v v_3 + 25 u v_1 v_2 - 
 \frac{50}{3} v v_1 v_3 - 10 v v_2^2 - 40 v_1^2 v_2 - 
 \frac{10}{3} v^3 u_2 + 5 u v^2 u_1 + \frac{5}{9} v^3 v_3\\ 
 \qquad - 
 \frac{40}{3}  v^2 u_1 v_1 + 10 u^2 v v_1 - 
 \frac{80}{3} u v v_1^2 + \frac{15}{2} v^2 v_1 v_2 + \frac{10}{3} v v_1^3 - 
 \frac{40}{3} u v^2 v_2 + \frac{10}{3} v^4 v_2 + 10 v^3 v_1^2\\
 \qquad - 
 \frac{55}{36} v^4 u_1 - \frac{40}{9} u v^3 v_1 + \frac{25}{54} v^5 v_1,\\
v_t=-9 v_5 + 
 \bigg(-10 v u_2 - 15 u_1 v_1 - 15 u v_2 + 30 v_1 v_2 - 
   10 v^2 u_1 - 5 u^2 v + \frac{15}{2} v^2 v_2 + 
   10 v v_1^2\\
   \qquad - \frac{5}{3} u v^3 + \frac{10}{3} v^3 v_1 + \frac{11}{36} v^5\bigg)_x;
 \end{array}\right. \tag{$\rm G_2$}\label{G_2}
\end{align}

\underline{$W(u)=1,\ W(v)=2$}:
\begin{align}
&\left\{ \begin{array}{l} 
u_t=u_5 + 5 u_1 u_3 + \frac{15}{4} u_2^2 - 20 v v_2 - 5 v_1^2 + 
 \frac{5}{3} u_1^3 - 10 v^2 u_1 + \frac{20}{3} v^3,\\
v_t=-9 v_5 - 5 v u_4 - \frac{35}{2} v_1 u_3 - 15 u_1 v_3 - 
 \frac{45}{2} u_2 v_2 + 30 v_1 v_2 - 5 v^2 u_2 - 5 v u_1 u_2 - 
 5 v_1 u_1^2 + 10 v^2 v_1;
\end{array}\right. \tag{$\rm G_3$} \label{G_3}
\end{align}

\underline{$W(u)=1,\ W(v)=1$}:
\begin{align}
&\left\{ \begin{array}{l} u_t =u_5 + \left(5 u_1 u_2 -5 u^2 u_2-5 u u_1^2+u^5-30 v v_3-45 v_1 v_2+15 u_2 v^2  +60 uvv_2 + 90 v^2 v_2\right)_x\\
\qquad  +  \left(30 v u_1 v_1+15 u v_1^2+180 v v_1^2+30 u^2 v v_1 -180 u v^2 v_1 -90 v^3 v_1-30 u^3 v^2+45 u v^4\right)_x, \\
 v_t = -9 v_5 + \left(10 u_3 v+ 15 u_2 v_1+15 u_1 v_2 +90 v_1 v_2+20 uu_2 v + 30 u_2 v^2 + 15 u_1^2 v + 30 uu_1 v_1 \right)_{x} \\
\qquad  +  \left(15 u^2 v_2+15 v^2 v_2+15 v v_1^2+60 u u_1 v^2 +30 u_1 v^3-10 u^2 u_1 v+30 u^2 v^3 -5 u^4 v-9 v^5\right)_x;  \end{array}\right. \tag{$\rm G_4$} \label{G_4}\\[4mm]
&\left\{ \begin{array}{l} u_t=u_5 + 5 u_1 u_3 + \frac{5}{3} u_1^3 + 10 v v_4 + 5 v_1 v_3 + 
 10 v_2^2 + 10 v^2 u_3 + 10 v v_1 u_2 + 
 30  v u_1 v_2 + 5 u_1 v_1^2\\
 \qquad - \frac{10}{3} v^2 v_3 - 
 40 v v_1 v_2 + \frac{20}{3} v^3 u_2 + 10 v^2 u_1^2 + 
 \frac{140}{9} v^3 v_2 + \frac{50}{3} v^2 v_1^2 - \frac{40}{3} v^2 u_1  v_1\\
 \qquad + 
 \frac{80}{9}  v^4 u_1 + \frac{160}{81} v^6,\\
v_t=-9 v_5 - 
 \bigg(10 v u_3 + 15 v_1 u_2 + 15 u_1 v_2 - 30 v_1 v_2 + 
   10 v^2 u_2 + 5 v u_1^2 + 10 v^2 v_2 + 15 v v_1^2\\
   \qquad + 
   \frac{20}{3} v^3 u_1 + \frac{20}{3} v^3 v_1 + \frac{16}{9} v^5\bigg)_x;
 \end{array}\right. \tag{$\rm G_5$} \label{G_5}\\[4mm]
&\left\{ \begin{array}{l} u_t=u_5 - 10 u_1 u_3 - \frac{15}{2} u_2^2 + 10 v_1 v_3 + \frac{5}{2} v_2^2 - 
 5 (u - v)^2 v_3 - 10 (u - v) v_1 u_2 - 
 5 (u - v) u_1 v_2\\
 \qquad + 15 (u - v) v_1 v_2 + 
 \frac{20}{3} u_1^3 - 10 v_1 u_1^2 + 10 v_1^2 u_1 - \frac{40}{3} v_1^3 + 
 5 (u - v)^3 (u_2 - v_2) + \frac{15}{2} (u - v)^2 u_1^2\\
 \qquad - 
 5 (u - v)^2 u_1 v_1 + \frac{25}{2} (u - v)^2 v_1^2 - 
 \frac{5}{2} (u - v)^4 u_1 - \frac{5}{2} (u - v)^4 v_1 + 
 \frac{5}{12} (u - v)^6,\\
v_t =-9 v_5 + 10 (u - v) u_4 + 35 u_1 u_3 - 25 v_1 u_3 - 
 15 u_1 v_3 + 45 v_1 v_3 + \frac{55}{2} u_2^2 + \frac{75}{2} v_2^2 - 
 30 u_2 v_2\\
 \qquad - 15 (u - v)^2 u_3 - 
 80 (u - v) u_1 u_2 + 30 (u - v) v_1 u_2 + 
 15 (u - v) u_1 v_2 - 5 (u - v) v_1 v_2 - 
 \frac{55}{3} u_1^3\\
 \qquad + 25 u_1^2 v_1 - 5 u_1 v_1^2 - \frac{25}{3} v_1^3 + 
 15 (u - v)^3 u_2 - 5 (u - v)^3 v_2 + 
 \frac{75}{2} (u - v)^2 u_1^2 - 35 (u - v)^2 u_1 v_1\\
 \qquad + 
 \frac{25}{2} (u - v)^2 v_1^2 - \frac{15}{2} (u - v)^4 u_1 + 
 \frac{5}{2} (u - v)^4 v_1 + \frac{5}{12} (u - v)^6;
 \end{array}\right. \tag{$\rm G_6$} \label{G_6}\\[4mm]
&\left\{ \begin{array}{l} u_t =u_5 - 10 u_1 u_3 + 90 v_1 v_3 - \frac{15}{2} u_2^2 + \frac{45}{2} v_2^2 - 
 45 (u - v)^2 v_3 - 90 (u - v) v_1 u_2\\
 \qquad - 
 45 (u - v) u_1 v_2 + 135 (u - v) v_1 v_2 + 
 \frac{20}{3} u_1^3 - 90 u_1^2 v_1 + 90 u_1 v_1^2 + 
 45 (u - v)^3 u_2\\
 \qquad - 45 (u - v)^3 v_2 + 
 \frac{135}{2} (u - v)^2 u_1^2 - 45 (u - v)^2 u_1 v_1 - 
 \frac{135}{2} (u - v)^2 v_1^2 - \frac{45}{2} (u - v)^4 u_1\\
 \qquad + 
 \frac{135}{2} (u - v)^4 v_1 - \frac{45}{4} (u - v)^6, \\
v_t=-9 v_5 + 10 (u - v) u_4 + 35 u_1 u_3 - 25 v_1 u_3 - 
 15 u_1 v_3 + 45 v_1 v_3 + \frac{55}{2} u_2^2 - 30 u_2 v_2 - 
 \frac{45}{2} v_2^2\\
 \qquad - 15 (u - v)^2 u_3 - 80 (u - v) u_1 u_2 + 
 30 (u - v) v_1 u_2 + 135 (u - v) u_1 v_2 - 
 45 (u - v) v_1 v_2\\
 \qquad - \frac{55}{3} u_1^3 + 25 u_1^2 v_1 - 
 45 u_1 v_1^2 + 45 v_1^3 + 15 (u - v)^3 u_2 - 
 45 (u - v)^3 v_2 - \frac{45}{2} (u - v)^2 u_1^2\\ 
 \qquad + 
 45 (u - v)^2 u_1 v_1 - \frac{135}{2} (u - v)^2 v_1^2 + 
 \frac{45}{2} (u - v)^4 u_1 + \frac{45}{2} (u - v)^4 v_1 - 
 \frac{45}{4} (u - v)^6.
 \end{array}\right. \tag{$\rm G_7$} \label{G_7}
\end{align}
\end{enumerate}
\end{The}

For systems~\eqref{D_1}--\eqref{D_5}, we display only the $u_t$-equation,
$ 
u_t=F([u],[v],\lambda).
$ 
The corresponding $v_t$-equation is obtained by interchanging $u$ and $v$ and replacing $\lambda$ with $\bar{\lambda}$.

The integrable systems of the form~\eqref{sysg} corresponding to Cases~$16, 18, 22$ and~$23$ have rather lengthy expressions. Nevertheless, they can be transformed into simple and compact linear triangular systems through almost invertible transformations, consisting of combinations of point transformations and the introduction of potentials. The classification results for these systems are presented in Section~\ref{LaxTrans}.

For fifth-order systems~\eqref{sysg}, no classification results are currently available in the literature, and only scattered examples are known.
In \cite{TW17}, the authors studied fifth-order integrable coupled systems of weights $W(u)=W(v)=0$ and $W(u)=W(v)=1$ that admit a seventh-order symmetry. They identified four integrable systems. Among these, system (9) is, up to a linear transformation, the potential form of \eqref{D_5}. The remaining three systems have linear parts with one zero eigenvalue and therefore fall outside the class of systems considered in our classification.
In \cite{Gerdjikov20}, the authors derived the hierarchies of the mKdV type equations related to $D_4^{(1)} , D_4^{(2)}$ and $D_4^{(3)}$ Kac-Moody algebras based on the construction of the three non-equivalent Coxeter gradings in the Lie algebra $D_4$. The one  related to $D_4^{(3)}$ is the system  (\ref{D_5}) under the scaling transformations of time $t$ and dependent variables.

\subsection{Integrable hierarchies starting from systems with $\lambda_1=0$}\label{eigzero}

In this section, we consider integrable hierarchies starting with a system of the form \eqref{sysg} with $\lambda_1=0$. Without loss of generality, we can set $\lambda_2=1$. Hierarchies of this type can occur only in the following cases:

\begin{description}
\item[Case 13.] Here $q=e^{\frac{\pi \i}{3}}$, and the hierarchy starts with a system of order $3$. To classify these systems it is sufficient to impose the existence of  symmetry of order $5$ with $\mu_2=-9,\,\mu_1=1$. 
\item[Case 14.] Here $q=e^{\frac{\pi \i}{5}}$ or $q=e^{\frac{3\pi \i}{5}}$, and the hierarchy starts with a system of order $5$. In this case, it is sufficient to impose the  existence of a symmetry of order $7$ with $\frac{\mu_2}{\mu_1}=\frac{(1+q)^7}{1+q^7}$.
\item[Case 15.] Here $q=e^{\frac{\pi \i}{7}},\,e^{\frac{3\pi \i}{7}}$ or $e^{\frac{5\pi \i}{7}}$, and the hierarchy starts with a system of order $7$. In this case, it is sufficient to impose the existence of a symmetry of order $13$ with $\frac{\mu_2}{\mu_1}=\frac{(1+q)^{13}}{1+q^{13}}$.
\item[Case 10.] The hierarchy starts with a system of order $3$. To classify these systems, it is sufficient to impose the existence of a symmetry of order $5$ with arbitrary ratio $\frac{\mu_2}{\mu_1}$. 
\item[Case 11.] The hierarchy starts with a system of order $5$. To classify these systems, it is sufficient to impose the existence of a symmetry of order $7$ with arbitrary ratio $\frac{\mu_2}{\mu_1}$. 
\item[Case 12.] The hierarchy starts with a system of order $7$. To classify these systems, it is sufficient to impose the existence of a symmetry of order $13$ with arbitrary ratio $\frac{\mu_2}{\mu_1}$. 
\end{description}

In this section, we present the classification of all such systems of order $3$, that is, we consider only Cases~13 and~10. We conjecture that, in all remaining cases, systems of this type are either triangular or related to triangular systems by differential substitutions.
Integrable systems corresponding to Case~13 are denoted by the labels $_0A$ and $_0L$. The former admit Lax representations in terms of Lie algebra $A_3^{(2)}$, while the latter are related to linear triangular systems by differential substitutions. Integrable systems corresponding to Case~10 are denoted by $_0M$.

\begin{The}\label{thmApA}
Let the system~(\ref{sysg}) with $n=3$ and $\lambda_1=0,\,\lambda_2=1$ satisfy conditions~(I)--(IV)  and be integrable. Then, under the transformations~(\ref{scal}), it can be reduced to one of the following equations:

\begin{enumerate}
    \item[{\rm\bf i.}]  ${\rm\bf Case\ 13},\qquad  q=e^{\frac{\i \pi }{3}}$,  type ${\rm _0A}$.
\medskip

\underline{$W(u)=2,\ W(v)=2$}:
\begin{align}
&\left\{ \begin{array}{l} u_t=v v_1, \\
v_t=v_3 + 2 u v_1 + v u_1; 
\end{array}\right.\tag{$\rm _0A_1$} \label{A_1p}
\end{align}

\underline{$W(u)=2,\ W(v)=1$}:
\begin{align}
&\left\{ \begin{array}{l} u_t=-2 v_1 v_2 - 6 u v v_1 + 2 v^3 v_1, \\
v_t=v_3 + (3 u v -  v^3)_x; 
\end{array}\right.\tag{$\rm _0A_2$} \label{A_2p}
\end{align}

\underline{$W(u)=1,\ W(v)=2$}:
\begin{align}
&\left\{ \begin{array}{l} u_t=v^2, \\
v_t=v_3 + 2 u_1 v_1 + v u_2; 
\end{array}\right.\tag{$\rm _0A_3$} \label{A_3p}
\end{align}

\underline{$W(u)=1,\ W(v)=1$}:
\begin{align}
&\left\{ \begin{array}{l}
u_t=\left(vv_1-uv^2\right)_x,\\
v_t=v_3-\left(vu_1+u^2v\right)_x;
\end{array}\right.\tag{$\rm _0A_4$} \label{A_4p}\\[4mm]
&\left\{ \begin{array}{l} u_t=3 v_1^2 + 6 (u - v)^2 v_1 + 3 (u - v)^4, \\
v_t=v_3 + 2 u u_2 + 2 u_1^2 - 2 v u_2 - 2 u_1 v_1 + 
 3 v_1^2 + 6 (u - v)^2 u_1 + 3 (u - v)^4; 
\end{array}\right.\tag{$\rm _0A_5$} \label{A_5p}\\[4mm]
&\left\{ \begin{array}{l}
u_t=-3 v v_2 -3 v^2 u_1 - 3 v^4,\\
v_t=v_3 + (v u_1 +   v^3)_x;
\end{array}\right.\tag{$\rm _0A_6$} \label{A_7p}
\end{align}
\item[{\rm\bf ii.}]  
 ${\rm\bf Case\ 13},\qquad  q=e^{\frac{\i \pi }{3}}$, type ${\rm _0L}$.
 \medskip

\underline{$W(u)=1,\ W(v)=\frac{3}{2}$}:
\begin{align}
&\left\{ \begin{array}{l} u_t=2 v v_1, \\
v_t=v_3 + v u_2 + 3 u_1 v_1 + 3 u v_2 + 3 u^2 v_1 + 
 3 u v u_1 - v^3 + u^3 v; 
\end{array}\right.\tag{$\rm _0L_1$} \label{L_1pp}
\end{align}

\underline{$W(u)=1,\ W(v)=\frac{1}{2}$}:
\begin{align}
&\left\{ \begin{array}{l}
u_t=6 v_1 v_2 + 2 (3 - 2 \delta) u v v_2 + 6 v u_1 v_1 + 
 2 (\delta+3) u v_1^2  + 6 \delta v^3 v_2 + 
 24 \delta v^2 v_1^2- 2 (2 \delta - 3) u v^2 u_1\\
\qquad + 
 6 \delta v^4 u_1 - 2 (2 \delta -3) u^2 v v_1 - 
 12 (\delta - 3) \delta u v^3 v_1 + 
 30 \delta^2 v^5 v_1 - 
 2 \delta v^2 (u - 3v^2) (u + \delta v^2)^2,\\
v_t=v_3 + v u_2 + 3 u_1 v_1 + 
 3 (u + \delta v^2) v_2 + (3 u + \delta v^2) v u_1 + 3(3 \delta - 1) v v_1^2 + 3 u^2 v_1
 \\ \qquad  + 
 2 (5 \delta - 3) u v^2 v_1 + 3\delta (  \delta - 2) v^4 v_1 + v (u - 3v^2) (u + \delta v^2)^2,
\end{array}\right.\tag{$\rm _0L_2$} \label{L_3pp}
\end{align}
where $\delta\in\C\setminus\{1\}$;

\underline{$W(u)=1,\ W(v)=1$}:
\begin{align}
&\left\{ \begin{array}{l}
u_t=-2 v v_2 + v_1^2 - 2  v^2 u_1 - 2 v^4,\\
v_t=v_3 + v u_2 + 2 u_1 v_1 + 4 v^2 v_1;
\end{array}\right.\tag{$\rm  _0K_1$} \label{A_8p}
\\[4mm]
&\left\{ \begin{array}{l}
u_t=-2 v v_2 + v_1^2 + 4  v^2 u_1 + 2 u v v_1 - 
 6 v^2 v_1,\\
v_t=v_3 - 2 v u_2 - 3 (u - v) v_2 + 6 v_1^2 - 5 u_1 v_1 + 
 4 u v u_1 + 2 u (u - 3 v) v_1;
\end{array}\right.\tag{$\rm _0K_2$} \label{A_6p}
\end{align}
\item[{\rm\bf iii.}] ${\rm\bf Case\ 10}$,   type ${\rm _0M}$.
\medskip

\underline{$W(u)=1,\ W(v)=1$}:
\begin{align}
&\left\{ \begin{array}{l}
u_t=2 (u + v) v_2 - v_1^2 + 4 (u + v)^2 u_1 + 
 2 (3 u^2 - 6 u v - 5 v^2) v_1 + (u - v)^2 (3 u + v) (u + 3 v),\\
v_t=v_3 + 2 (u + v) u_2 + 6 (u - v) v_2 + 2 u_1^2 - 
 9 v_1^2 + 6 u_1 v_1 + 
 2 (5 u^2 + 6 u v - 3 v^2) u_1 \\ \qquad + 
 8 (u^2 - 4 u v + v^2) v_1 + (u - v)^2 (3 u + v) (u + 3 v);
\end{array}\right.\tag{$\rm _0M_1$} \label{0M1}
\end{align}

\underline{$W(u)=\frac{1}{2},\ W(v)=\frac{1}{2}$}:
\begin{align}
&\left\{ \begin{array}{l}
u_t=u v_1^2 - 2 u v v_2 - 4 u^2 v^2 u_1 - 
 2 u v (u^2 + 3 v^2) v_1 - 
 u v^2 (u^2 + v^2)^2,\\
v_t=v_3 + 2 u v u_2 + 3 (u^2 + v^2) v_2 + 
 6 u u_1 v_1 + 9 v v_1^2 + 2 v u_1^2 + 
 2 u v (3 u^2 + v^2) u_1\\ \qquad + (3 u^2 + v^2) (u^2 + 3 v^2) v_1 + 
 u^2 v (u^2 + v^2)^2.
\end{array}\right.\tag{$\rm _0M_2$} \label{L_4pp}
\end{align}
\end{enumerate}
\end{The}

\begin{Rem} System \eqref{L_3pp} with $\delta=1$ is excluded, since in this case the system admits a symmetry of order $2$:
\begin{equation*}
 \left\{ \begin{array}{l}  u_{\tau}=u_2+6 v_1^2+4 u u_1+4 u_1 v^2+4 u v v_1+12 v^3 v_1-2 u^2v^2 +4 u v^4+6v^6 ,
 \\
 v_{\tau}=3 v_2+2 u_1 v +6 u v_1+6 v^2 v_1+ u^2v-2 u v^3 -3 v^5 .
  \end{array}\right. 
\end{equation*}
Therefore, condition~(IV) is not satisfied.
\end{Rem}
In \cite{KuleminMeshkov1997}, the authors attempted to classify completely integrable systems of the form
\begin{eqnarray*}
 \left\{ \begin{array}{l}
u_t=u_3+f(u, v, u_1, v_1, u_2, v_2),\\
v_t=g(u, v, u_1,v_1).
\end{array}\right.
\end{eqnarray*}
Their classification contains systems~\eqref{A_1p} and~\eqref{A_3p}. However, system~\eqref{A_5p} is absent from their list.

\section{Lax representations and transformations}\label{LaxTrans}

The extensive list of equations presented in Theorems~\ref{thmA}--\ref{thmApA} can be divided into two classes, corresponding to the S-- and C--integrable systems. 
Systems labelled by the letters $A,B,D,G,K$ and $_0A$ in Theorems ~\ref{thmA}--\ref{thmApA} are S--integrable, while all remaining systems belong to the class of C--integrable systems.

Drinfeld and Sokolov established a fundamental connection between affine Lie algebras and integrable systems of KdV type \cite{DS}. They constructed canonical Lax operators associated with affine Lie algebras of rank~$n$, thereby generating integrable hierarchies of $n$ coupled KdV-type equations. In the rank-two case, they presented four examples of $L$-operators, corresponding to the affine Lie algebras $B_2^{(1)}$ and $A_3^{(2)}$.


In Section~\ref{CanLax}, we present the complete list of canonical Lax pairs and the corresponding canonical KdV-type systems associated with all affine Lie algebras of rank two. Section~\ref{Miura} is devoted to the relationship between these canonical systems and the systems appearing in our classification list. For the C--integrable systems, the differential substitutions that reduce them to linear--triangular form are presented in Section~\ref{diffsub}. Finally, in Section~\ref{sec64} we provide classification results and linearising transformations for the integrable systems of order $5$ corresponding to cases $16, 18, 22$ and $23$ in Table \ref{table2}.

\subsection{Canonical Lax representations corresponding to affine Lie algebras of rank $2$}\label{CanLax}

In this section, we present canonical Lax operators $\cL^{\mathrm{can}}$ associated with each affine Lie algebra of rank two and with each choice of a vertex of its Dynkin diagram, constructed according to the Drinfeld--Sokolov procedure \cite{DS}. For each operator $\cL^{\mathrm{can}}$, we also present the corresponding operator $\cA^{\mathrm{can}}$ forming the Lax pair for the KdV-type system of the lowest-order in the associated hierarchy. These systems are determined by the compatibility condition
$ 
[\cL^{\mathrm{can}},\cA^{\mathrm{can}}]=0.
$ 

The operators $\cL^{\mathrm{can}}$ and $ \cA^{\mathrm{can}}$ are expressed in terms of the canonical generators
$ 
\{e_i,f_i,h_i\}_{i=0,1,2}
$ 
of the affine Lie algebras
$ 
A_3^{(2)},\,  A_4^{(2)},\,  B_2^{(1)},\,  G_2^{(1)},\,  D_4^{(3)},
$ 
whose matrix representations are given in Appendix~\ref{RepAlg}. Throughout this section, we omit the superscript ``$\mathrm{can}$'' and write $\cL$ and $\cA$ in place of $\cL^{\mathrm{can}}$ and $\cA^{\mathrm{can}}$, respectively.
\medskip

\underline{Lie algebra $A_3^{(2)}$, vertex $c_2$}.

The Lax pair
\begin{eqnarray*}
\cL&=&D_x+p\,f_0+q\,f_1+e_0+e_1+e_2,\\
\cA&=&D_t+\left(\frac{1}{2} p_2 - \frac{1}{2} q_2 - p^2 + p q\right) f_0 + \left(-\frac{1}{2} p_2 + \frac{1}{2} q_2 + p q - 
    q^2\right) f_1 + (-\frac{1}{2} p_1 + \frac{1}{2} q_1) h_0 + \left(\frac{1}{2} p_1 - \frac{1}{2} q_1\right) h_1\\
    &&+ (q - 
    p) e_0 + (p - q) e_1 - 
 \frac{1}{2} \left([e_0, [e_0, e_2]] - [e_0, [e_1, e_2]] + [e_1, [e_1, e_2]]\right)
\end{eqnarray*}
yields the canonical  system 
\begin{equation}\label{DSA3_2_c2}\tag{$A_3^{(2)};c_2$}
    \left\{\begin{array}{l}
p_t=\frac{1}{2} p_3 - \frac{1}{2} q_3 - 3 p p_1 + q p_1 + 2 p q_1,\\ [4pt]
q_t=-\frac{1}{2} p_3 + \frac{1}{2} q_3 + 2 q p_1 + p q_1 - 3 q q_1.
\end{array}\right.
\end{equation}

\underline{Lie algebra $A_3^{(2)}$, vertex $c_0$}.

The Lax pair
\begin{eqnarray*}
\cL&=&D_x+p\,f_1+q\,[f_1,[f_1,f_2]]+e_0+e_1+e_2,\\
\cA&=&D_t+\left(\frac{1}{4} p_4 - \frac{1}{2} q_2 - \frac{1}{2} p p_2\right) [f_1, [f_1, f_2]] + \left(-\frac{1}{2} p_3 + q_1\right) [
   f_1, f_2] + (p_2 - 2 q - p^2) f_1 + (p_2 - 2 q) f_2 \\
   & +& \frac{1}{2} p_1 h_0- 
 \frac{1}{2} p_1 h_1 + p e_0 - p e_1 - 
 \frac{1}{2} \left([e_0, [e_0, e_2]] - [e_0, [e_1, e_2]] + [e_1, [e_1, e_2]]\right)
\end{eqnarray*}
yields the canonical  system 
\begin{equation}\label{DSA3_2_c0}\tag{$A_3^{(2)};c_0$}
    \left\{\begin{array}{l}
p_t=\frac{3}{2} p_3 - 3 q_1 - 3 p p_1,\\ [4pt]
q_t=\frac{1}{4} p_5 - \frac{1}{2} q_3 - p p_3 - \frac{1}{2} p_1 p_2 - 2 q p_1 + p q_1.
    \end{array}\right.
\end{equation}
The choice of vertex $c_1$ results in the same system \eqref{DSA3_2_c0}.

\underline{Lie algebra $A_4^{(2)}$, vertex $c_1$}. 

The Lax pair
\begin{eqnarray*}
\cL&=&D_x+p\,f_0+q\,f_2+e_0+e_1+e_2,\\
\cA&=&D_t+(2 p_2 - 12 q_2 - 4 p^2 + 24 p q) f_0 + (-3 p_2 + 8 q_2 + 6 p q - 
    16 q^2) f_2 + (-8 p_1 + 28 q_1) h_0 + (-6 p_1 + 16 q_1) h_1\\
    &&+ 4( 
    6q-p) e_0 + 4(   q-p) e_1 + (6 p - 16 q) e_2 - 
 5 \left( [e_1, [e_0, e_1]] - 2[e_2, [e_0, e_1]] - 2[e_2, [e_1, e_2]]\right)
\end{eqnarray*}
yields the canonical system 
\begin{align}\label{DSA4_2_c1}\tag{$A_4^{(2)};c_1$}
&\left\{ \begin{array}{l}
p_t=2 p_3 - 12 q_3 - 12 p p_1 + 24 q p_1 + 48 p q_1,\\ [4pt]
q_t=-3 p_3 + 8 q_3 + 12 q p_1 + 6 p q_1 - 48 q q_1
\end{array}\right.   
\end{align}

\underline{Lie algebra $A_4^{(2)}$, vertex $c_2$}. 

The Lax pair
\begin{eqnarray*}
\cL&=&D_x+p\,f_0+q\,[f_1,[f_0,f_1]]+e_0+e_1+e_2,\\
\cA&=&D_t+(3 p_4 - 10 q_2 - 3 p p_2 + 6 p q)  [f_1, 
    [f_0, f_1]] + (20 q_1 - 6 p_3)  [f_0, 
   f_1] + (8 p_2 - 20 q - 4 p^2) f_0 + (6 p_2 - 20 q) f_1\\
   &&- 8 p_1 h_0 - 
 6 p_1 h_1 - 4 p e_0 - 4 p e_1 + 6 p e_2 - 
5 \left( [e_1, [e_0, e_1]] - 2[e_2, [e_0, e_1]] - 2[e_2, [e_1, e_2]]\right)
\end{eqnarray*}
yields the canonical  system 
\begin{equation}\label{DSA4_2_c2}\tag{$A_4^{(2)};c_2$}
    \left\{\begin{array}{l}
p_t=20 p_3 - 60 q_1 - 12 p p_1,\\ [4pt]
q_t=3 p_5 - 10 q_3 - 3 p p_3 - 3 p_1 p_2 - 6 p_1 q + 
 6 p q_1.
    \end{array}\right.
\end{equation}

\underline{Lie algebra $A_4^{(2)}$, vertex $c_0$}. 

The Lax pair
\begin{eqnarray*}
\cL&=&D_x+p\,f_1+q\,[f_2,[f_1,f_2]]+e_0+e_1+e_2,\\
\cA&=&D_t+\left(-\frac{3}{2} p_4 + 20 q_2 + \frac{3}{2} p p_2 - 18 p q\right) [f_2, 
   [f_1, f_2]] + (3 p_3 - 40 q_1) [f_1, 
   f_2] + (-4 p_2 + 40 q + 2 p^2) f_1\\
   &&+ (-3 p_2 + 40 q) f_2 - 6 p_1 h_0 - 
 2 p_1 h_1 - 8 p e_0 + 2 p e_1 + 2 p e_2 - 
5 \left( [e_1, [e_0, e_1]] - 2[e_2, [e_0, e_1]] - 2[e_2, [e_1, e_2]]\right)
\end{eqnarray*}
yields the canonical  system 
\begin{equation}\label{DSA4_2_c0}\tag{$A_4^{(2)};c_0$}
    \left\{\begin{array}{l}
p_t=-10 p_3 + 120 q_1 + 6 p p_1,\\ [4pt]
q_t=-\frac{3}{2} p_5 + 20 q_3 + \frac{3}{2} p p_3 + \frac{3}{2} p_1 p_2 - 
 12 p_1 q - 18 p q_1.
    \end{array}\right.
\end{equation}

\underline{Lie algebra $B_2^{(1)}$, vertex $c_2$}. 

The Lax pair
\begin{eqnarray*}
\cL&=&D_x+p\,f_0+q\,f_1+e_0+e_1+e_2,\\
\cA&=&D_t+\left(-\frac{1}{2} p_2 + \frac{3}{2} q_2 - 3 p q + p^2\right) f_0 + \left(\frac{3}{2} p_2 - \frac{1}{2} q_2 - 3 p q + 
    q^2\right) f_1 + \left(\frac{1}{2} p_1 - \frac{3}{2} q_1\right) h_0 + \left(-\frac{3}{2} p_1 + \frac{1}{2} q_1\right) h_1\\&& + (p - 
    3 q) e_0 + (q - 3 p) e_1 + (p + q) e_2 -2 \left(2 [e_0, [e_1, e_2]] -
     [e_2, [e_0, e_2]] - [e_2, [e_1, e_2]]\right)
\end{eqnarray*}
yields the canonical  system 
\begin{equation}\label{DSB2_1_c2}\tag{$B_2^{(1)};c_2$}
    \left\{\begin{array}{l}
p_t=-\frac{1}{2} p_3 + \frac{3}{2} q_3 + 3 p p_1 - 3 q p_1 - 6 p q_1,\\ [4pt]
q_t=\frac{3}{2} p_3 - \frac{1}{2} q_3 - 6 q p_1 - 3 p q_1 + 3 q q_1.
    \end{array}\right.
\end{equation}

\underline{Lie algebra $B_2^{(1)}$, vertex $c_0$}.

The Lax pair
\begin{eqnarray*}
\cL&=&D_x+p\,f_1+q\,[f_2,[f_1,f_2]]+e_0+e_1+e_2,\\
\cA&=&D_t+\left(-\frac{3}{4} p_4 + 4 q_2 + \frac{3}{4} p p_2 - 3 p q\right)  [f_2, 
     [f_1, f_2]] + \left(\frac{3}{2} p_3 - 8 q_1\right)  [f_1, 
    f_2] + (-2 p_2 + 8 q + p^2) f_1\\
    &&+ \left(-\frac{3}{2} p_2 + 8 q\right) f_2 - \frac{3}{2} p_1 h_0 + 
  \frac{1}{2} p_1 h_1 - 3 p e_0 + p e_1 + p e_2 -2 \left(2 [e_0, [e_1, e_2]] -
     [e_2, [e_0, e_2]] - [e_2, [e_1, e_2]]\right)
\end{eqnarray*}
yields the canonical $(B_2^{(1)},c_0)$ system 
\begin{equation}\label{DSB2_1_c0}\tag{$B_2^{(1)};c_0$}
    \left\{\begin{array}{l}
p_t=-5 p_3 + 24 q_1 + 3 p p_1,\\ [4pt]
q_t=-\frac{3}{4} p_5 + 4 q_3 + \frac{3}{4} p p_3 + \frac{3}{4} p_1 p_2 - 
 3 p q_1.
    \end{array}\right.
\end{equation}
The choice of vertex $c_1$ results in the same system \eqref{DSB2_1_c0}.

\underline{Lie algebra $D_4^{(3)}$, vertex $c_1$}. 

The Lax pair
\begin{eqnarray*}
\cL&=& D_x+p\,f_0+q\,f_2+e_0+e_1+e_2,\\
\cA&=&D_t+a_1f_0+a_2f_2+a_3h_0+a_4h_1+a_5e_0+a_6e_1+a_7e_2+a_8[e_0,e_1]+a_9[e_1,e_2]+a_{10}[e_2,[e_0,e_1]]\\
&&+a_{11}[e_1,[e_1,e_2]]+6  [e_0,  [e_0,  [e_1,  [e_1, e_2]]]] - 
 2  [e_0,  [e_1,  [e_1,  [e_1, e_2]]]] + 
 2  [e_2,  [e_1,  [e_1, [e_1, e_2]]]],
\end{eqnarray*}
where
\begin{eqnarray*}
a_1&=&\frac{33}{2} p_4 - \frac{5}{2} q_4 - \frac{69}{2} p p_2 - \frac{25}{2} q p_2 - \frac{15}{2} p q_2 + \frac{5}{2} q q_2 - 
 \frac{3}{2} p_1^2 - 25 p_1 q_1 + \frac{5}{2} q_1^2 + \frac{3}{2} p^3 + 25 p^2 q - \frac{5}{2} p q^2,\\
a_2&=&-\frac{15}{2} p_4 + \frac{3}{2} q_4 - \frac{75}{2} p p_2 + \frac{45}{2} q p_2 + \frac{15}{2} p q_2 - \frac{9}{2} q q_2 - 
 \frac{75}{2} p_1^2 + 15 p_1 q_1 - \frac{3}{2} q_1^2 + \frac{75}{2} p^2 q - 15 p q^2 + \frac{3}{2} q^3,\\
a_3&=&-19 p_3 + 3 q_3 - 11 p p_1 + 15 q p_1 + 15 p q_1 - 3 q q_1,\\
a_4&=&-5 p_3 + q_3 - 25 p p_1 + 5 q p_1 + 5 p q_1 - q q_1,\\
a_5&=&-33 p_2 + 5 q_2 + \frac{3}{2} p^2 + 25 p q - \frac{5}{2} q^2,\\
a_6&=&9 p_2 - q_2 + \frac{3}{2} p^2 - 11 p q + \frac{3}{2} q^2,\\
a_7&=&15 p_2 - 3 q_2 + \frac{75}{2} p^2 - 15 p q + \frac{3}{2} q^2,\\
a_8&=&-42 p_1 + 6 q_1,\\
a_9&=&-6 p_1 + 2 q_1,\\
a_{10}&=&36 p - 4 q,\\
a_{11}&=&6 p - 2 q
\end{eqnarray*}
yields the canonical  system 
\begin{equation}\label{DSD4_3_c1}\tag{$D_4^{(3)};c_1$}
    \left\{\begin{array}{l}
p_t=\frac{33}{2} p_5 - \frac{5}{2} q_5 - \frac{135}{2} p p_3 - \frac{25}{2} q p_3 - \frac{5}{2} p q_3 + \frac{5}{2} q q_3 - 
 \frac{75}{2} p_1 p_2 - \frac{75}{2} q_1 p_2 - \frac{65}{2} p_1 q_2 + \frac{15}{2} q_1 q_2\\ [3pt] \qquad + \frac{15}{2} p^2 p_1+ 
 75 p q p_1 - \frac{5}{2} q^2 p_1 + 50 p^2 q_1 - 10 p q q_1,\\ [4pt]
q_t=-\frac{15}{2} p_5 + \frac{3}{2} q_5 - \frac{75}{2} p p_3 + \frac{75}{2} q p_3 + \frac{15}{2} p q_3 - \frac{15}{2} q q_3 - 
 \frac{225}{2} p_1 p_2 + \frac{75}{2} q_1 p_2 + \frac{45}{2} p_1 q_2 - \frac{15}{2} q_1 q_2 \\ [3pt] \qquad+ 150 p q p_1 - 
 30 q^2 p_1 + \frac{75}{2} p^2 q_1 - 45 p q q_1 + \frac{15}{2} q^2 q_1.
    \end{array}\right.
\end{equation}

\underline{Lie algebra $D_4^{(3)}$, vertex $c_2$}.

The Lax pair
\begin{eqnarray*}
\cL&=&D_x+p\,f_0+q\,[f_0,f_1]+e_0+e_1+e_2,\\
\cA&=&D_t+a_1[f_0,f_1]+a_2f_0+a_3f_1+a_4h_0+a_5h_1+a_6e_0+a_7e_1+a_8e_2+a_9[e_0,e_1]+a_{10}[e_1,e_2]\\ &&+a_{11}[e_0,[e_1,e_2]]+a_{12}[e_1,[e_1,e_2]]+6  [e_0,  [e_0,  [e_1,  [e_1, e_2]]]] - 
 2  [e_0,  [e_1,  [e_1,  [e_1, e_2]]]] + 
 2  [e_2,  [e_1,  [e_1, [e_1, e_2]]]],
\end{eqnarray*}
where
\begin{eqnarray*}
a_1&=&-5 p_5 - 6 q_4 - 25 p p_3 - 36 p q_2 - 75 p_1 p_2 - 111 q p_2 - 114 p_1 q_1 - 
 150 q q_1 + \frac{3}{2} p^2 q,\\
a_2&=&19 p_4 + 24 q_3 - 22 p p_2 + 11 p_1^2 - 24 p q_1 + 60 q p_1 + 54 q^2 + 
 \frac{3}{2} p^3,\\
a_3&=&5 p_4 + 6 q_3 + 25 p p_2 + 25 p_1^2 + 36 p q_1 + 78 q p_1 + 54 q^2,\\
a_4&=&-19 p_3 - 24 q_2 - 11 p p_1 - 18 p q,\\
a_5&=&-5 p_3 - 6 q_2 - 25 p p_1 - 36 p q,\\
a_6&=&-33 p_2 - 42 q_1 + \frac{3}{2} p^2,\\
a_7&=&9 p_2 + 12 q_1 + \frac{3}{2} p^2,\\
a_8&=&15 p_2 + 18 q_1 + \frac{75}{2} p^2,\\
a_9&=&-42 p_1 - 54 q,\qquad
a_{10}=-6 p_1 - 6 q,\qquad
a_{11}=-36 p,\qquad
a_{12}=6 p
\end{eqnarray*}
yields the canonical  system 
\begin{equation}\label{DSD4_3_c2}\tag{$D_4^{(3)};c_2$}
    \left\{\begin{array}{l}
    p_t=24 p_5 + 30 q_4 - 30 p p_3 + 75 p_1 p_2 - 30 p q_2 + 270 q q_1 + 180 q p_2 + 
 150 p_1 q_1 + \frac{15}{2} p^2 p_1,\\ [4pt]
 q_t=-5 p_6 - 6 q_5 - 20 p p_4 - 100 p_1 p_3 - 75 p_2^2 - 30 p q_3 - 135 q p_3 -150 p_1 q_2 - 225 q_1 p_2 - 180 q q_2\\ [3pt]
 \quad - 150 q_1^2 + 25 p^2 p_2+ 
 25 p p_1^2 + 45 p q p_1 + \frac{75}{2} p^2 q_1.
    \end{array}
    \right.
\end{equation}
\newpage
\underline{Lie algebra $D_4^{(3)}$, vertex $c_0$}.

The Lax pair
\begin{eqnarray*}
\cL&=&L=D_x+p\,f_1+q\, [f_2,[f_1,[f_1,[f_1,f_2]]]]+e_0+e_1+e_2,\\ 
\cA&=&D_t+a_1[f_2, [f_1, [f_1, [f_1, f_2]]]]+a_2[f_1, [f_1, [f_1, f_2]]]+a_3[f_1,[f_1,f_2]]+a_4[f_1,f_2]+a_5f_1+a_6f_2\\ 
&&+a_7h_0+a_8h_1+a_9e_0+a_{10}e_1+a_{11}e_2+a_{12}[e_0,e_1]+a_{13}[e_0,[e_1,e_2]]\\
&&+6  [e_0,  [e_0,  [e_1,  [e_1, e_2]]]] - 
 2  [e_0,  [e_1,  [e_1,  [e_1, e_2]]]] + 
 2  [e_2,  [e_1,  [e_1, [e_1, e_2]]]],
\end{eqnarray*}
where
\begin{eqnarray*}
a_1&=&\frac{5}{12} p_8 + 4 q_4 - \frac{55}{12} p p_6 - \frac{20}{3} p_1 p_5 - \frac{15}{2} p_2 p_4 - \frac{25}{6} p_3^2 - 
 40 p q_2 - 40 p_1 q_1 - 15 q p_2 + \frac{95}{12} p^2 p_4 + \frac{125}{6} p p_1 p_3\\
 &&+ 
 \frac{55}{4} p p_2^2 + \frac{55}{4} p_1^2 p_2 + \frac{75}{2} p^2 q - \frac{15}{4} p^2 p_1^2 - \frac{15}{4} p^3 p_2,\\
a_2&=&-\frac{5}{12} p_7 - 4 q_3 + \frac{10}{3} p p_5 + \frac{10}{3} p_1 p_4 + \frac{25}{6} p_2 p_3 + 28 p q_1 + 
 12 q p_1 - \frac{5}{4} p_1^3 - 10 p p_1 p_2 - \frac{35}{12} p^2 p_3,\\
a_3&=&\frac{5}{4} p_6 + 12 q_2 - 5 p p_4 - 5 p_1 p_3 - \frac{15}{4} p_2^2 - 36 p q + 
 \frac{15}{4} p^2 p_2 + \frac{15}{4} p p_1^2,\\
a_4&=&-5 p_5 - 48 q_1 + 5 p p_3 + 15 p_1 p_2,\\
a_5&=&9 p_4 + 72 q - 12 p p_2 - 9 p_1^2 + \frac{3}{2} p^3,\\
a_6&=&15 p_4 + 144 q - 15 p p_2 - 15 p_1^2,\\
a_7&=&5 p_3 - 5 p p_1,\\
a_8&=&p_3 - p p_1,\\
a_9&=&9 p_2 + \frac{3}{2} p^2,\\
a_{10}&=&-3 p_2 + \frac{3}{2} p^2,\\
a_{11}&=&-3 p_2 + \frac{3}{2} p^2,\\
a_{12}&=&12 p_1,\\
a_{13}&=&12 p
\end{eqnarray*}
yields the canonical  system 
\begin{equation}\label{DSD4_3_c0}\tag{$D_4^{(3)};c_0$}
    \left\{\begin{array}{l}
        p_t=14 p_5 + 120 q_1 - 20 p p_3 - 45 p_1 p_2 + \frac{15}{2} p^2 p_1 , \\ [4pt]
        q_t= \frac{5}{12} p_9 + 4 q_5 - 30 p_3 q - 40 p q_3 - 
 80 p_1 q_2 - 55 p_2 q_1 - \frac{45}{4} p_1 p_6 - 
 \frac{85}{6} p_2 p_5 - \frac{55}{12} p p_7 - \frac{95}{6} p_3 p_4\\ [3pt] \quad + 
 90 p p_1 q  + \frac{75}{2} p^2 q_1 + \frac{95}{12} p^2 p_5 + 
 \frac{110}{3} p p_1 p_4 + \frac{145}{3} p p_2 p_3 + 
 \frac{415}{12} p_1^2 p_3 + \frac{165}{4} p_1 p_2^2 - \frac{15}{4} p^3 p_3\\ [3pt] \quad - 
 \frac{75}{4} p^2 p_1 p_2 - \frac{15}{2} p p_1^3.
    \end{array}  \right.
\end{equation}

\underline{Lie algebra $G_2^{(1)}$, vertex $c_1$}.

The Lax pair
\begin{eqnarray*}
\cL&=&D_x+p\,f_0+q\,f_2+e_0+e_1+e_2,\\
\cA&=&D_t+a_1f_0+a_2f_2+a_3h_0+a_4h_1+a_5e_0+a_6e_1+a_7e_2+a_8[e_0,e_1]+a_9[e_1,e_2]+a_{10}[e_2,[e_0,e_1]]\\&&+a_{11}[e_2,[e_1,e_2]]-9 [e_1, [e_0, [e_2, [e_1, e_2]]]] + 
  6 [e_0, [e_2, [e_2, [e_1, e_2]]]] - 
  3 [e_1, [e_2, [e_2, [e_1, e_2]]]],
\end{eqnarray*}
where
\begin{eqnarray*}
a_1&=&-\frac{13}{2} p_4 + \frac{15}{2} q_4 + 12 p p_2 + 15 q p_2 - 45 q q_2 - p_1^2 - 45 q_1^2 + 
 30 p_1 q_1 + p^3 - 30 p^2 q + 45 p q^2,\\
a_2&=&\frac{5}{2} p_4 - \frac{3}{2} q_4 + 5 p p_2 - 20 q p_2 - 15 p q_2 + 12 q q_2 + 5 p_1^2 - 
 30 p_1 q_1 + 9 q_1^2 - 5 p^2 q + 30 p q^2 - 9 q^3,\\
a_3&=&9 p_3 - 9 q_3 + 6 p p_1 - 30 q p_1 - 30 p q_1 + 54 q q_1,\\
a_4&=&5 p_3 - 3 q_3 + 10 p p_1 - 30 q p_1 - 30 p q_1 + 18 q q_1,\\
a_5&=&13 p_2 - 15 q_2 + p^2 - 30 p q + 45 q^2,
\end{eqnarray*}
\begin{eqnarray*}
a_6&=&p_2 + 3 q_2 + p^2 - 12 p q - 9 q^2,\qquad
a_7=-5 p_2 + 3 q_2 - 5 p^2 + 30 p q - 9 q^2,\qquad
a_8=12 p_1 - 18 q_1,\\
a_9&=&6 p_1,\qquad\qquad\qquad\qquad\qquad\qquad
a_{10}=-6 p + 18 q,\qquad\qquad\qquad\qquad\qquad\,\,\,
a_{11}=-6 p
\end{eqnarray*}
 yields the canonical  system 
\begin{equation}\label{DSG2_1_c1}\tag{$G_2^{(1)};c_1$}
    \left\{\begin{array}{l}
        p_t=-\frac{13}{2} p_5 + \frac{15}{2} q_5 + 25 p p_3 + 15 q p_3 - 15 p q_3 - 45 q q_3 + 
 10 p_1 p_2 \\ [3pt] \quad + 45 q_1 p_2 + 30 p_1 q_2 - 135 q_1 q_2 + 5 p^2 p_1- 90 p q p_1 + 
 45 q^2 p_1 - 60 p^2 q_1 + 180 p q q_1, \\ [4pt]
        q_t=\frac{5}{2} p_5 - \frac{3}{2} q_5 + 5 p p_3 - 25 q p_3 - 15 p q_3 + 15 q q_3 + 15 p_1 p_2 - 
 50 q_1 p_2 \\ [3pt] \quad - 45 p_1 q_2 + 30 q_1 q_2 - 20 p q p_1 + 60 q^2 p_1 - 5 p^2 q_1 + 
 90 p q q_1 - 45 q^2 q_1.
    \end{array}  \right.
\end{equation}

\underline{Lie algebra $G_2^{(1)}$, vertex $c_0$}.

The Lax pair
\begin{eqnarray*}
\cL&=&D_x+p\,f_1+q\,[f_1,[f_2,[f_2,[f_1,f_2]]]]+e_0+e_1+e_2,\\
\cA&=&D_t+a_1[f_1,[f_2,[f_2,[f_1,f_2]]]]+a_2[f_2,[f_2,[f_1,f_2]]]+a_3[f_2,[f_1,f_2]]+a_5[f_1,f_2]+a_4f_1+a_6f_2+a_7h_1\\
&&+a_8e_0+a_9e_1+a_{10}e_2+a_{11}[e_0,e_1]+a_{12}[e_1,e_2]+a_{13}[e_2,[e_0,e_1]]+a_{14}[e_2,[e_1,e_2]]\\
&&-9[e_1, [e_0, [e_2, [e_1, e_2]]]] + 
  6 [e_0, [e_2, [e_2, [e_1, e_2]]]] -
  3 [e_1, [e_2, [e_2, [e_1, e_2]]]],
\end{eqnarray*}
where
\[
\begin{array}{llll}
     a_1=-9 q_4 + q p_2 + 6 p q_2 + 9 p_1 q_1 + p^2 q,\quad& a_5=108 q_1,\quad &
     a_9\,\,=-2 p_2 + p^2,\quad&a_{13}=-6 p\\
     a_2=9 q_3 - 3 q p_1 - 6 p q_1,& a_6=-108 q,&a_{10}=p_2 + p^2,&a_{14}=3p,\\
     a_3=-27 q_2 + 18 p q,&a_7=-p_3 + p p_1,\quad&a_{11}=3p_1,\\
     a_4=p_4 - 216 q - p_1^2 - 3 p p_2 + p^3,&a_8=p_2 + p^2,&a_{12}=3p_1
\end{array}
\]
 yields the canonical  system 
\begin{equation}\label{DSG2_1_c0}\tag{$G_2^{(1)};c_0$}
    \left\{\begin{array}{l}
p_t=p_5 - 540 q_1 - 5 p p_3 - 5 p_1 p_2 + 5 p^2 p_1,\\ [4pt]
q_t=-9 q_5 + 15 p q_3 + 15 p_1 q_2 + 10 q_1 p_2  - 
 5 p^2 q_1.
    \end{array}\right.
\end{equation}

\underline{Lie algebra $G_2^{(1)}$, vertex $c_2$}.

The Lax pair
\begin{eqnarray*}
\cL&=&D_x+pf_0+q\,[f_0,f_1]+e_0+e_1+e_2,\\
\cA&=&D_t+a_1[f_0,f_1]+a_2f_0+a_3f_1+a_4h_0+a_5h_1+a_6e_0+a_7e_1+a_8e_2\\
&&+a_{9}[e_0,e_1]+a_{10}[e_1,e_2]+a_{11}[e_2,[e_0,e_1]]+a_{12}[e_2,[e_1,e_2]]\\
&&-9 [e_1, [e_0, [e_2, [e_1, e_2]]]] + 
  6 [e_0, [e_2, [e_2, [e_1, e_2]]]] -
  3 [e_1, [e_2, [e_2, [e_1, e_2]]]] ,
\end{eqnarray*}
where
\begin{eqnarray*}
a_1&=&5 p_5 + 6 q_4 + 10 p p_3 + 30 p_1 p_2 + 31 q p_2 + 6 p q_2 + 24 p_1 q_1 + 
 30 q q_1 + q p^2,\\
a_2&=&-9 p_4 - 9 q_3 + 7 p p_2 - 6 p_1^2 - 15 q p_1 + 9 p q_1 + p^3 - 9 q^2,\\
a_3&=&-5 p_4 - 6 q_3 - 10 p p_2 - 10 p_1^2 - 18 q p_1 - 6 p q_1 - 9 q^2,\\
a_4&=&9 p_3 + 9 q_2 + 6 p p_1 + 3 p q,\\
a_5&=&5 p_3 + 6 q_2 + 10 p p_1 + 6 p q,\\
a_6&=&13 p_2 + 12 q_1 + p^2,\\
a_7&=&p_2 + 3 q_1 + p^2,\\
a_8&=&-5 p_2 - 6 q_1 - 5 p^2,\\
a_9&=&12 p_1 + 9 q,\\
a_{10}&=&6 p_1 + 9 q,\\
a_{11}&=&-6 p,\\
a_{12}&=&-6 p
\end{eqnarray*}
yields the canonical  system 
\begin{equation}\label{DSG2_1_c2}\tag{$G_2^{(1)};c_2$}
    \left\{\begin{array}{l}
p_t=-14 p_5 - 15 q_4 + 10 p p_3 - 35 p_1 p_2 - 45 q p_2 + 15 p q_2 - 30 p_1 q_1 - 
 45 q q_1 + 5 p^2 p_1,\\ [4pt]
q_t=5 p_6 + 6 q_5 + 5 p p_4 + 40 p_1 p_3 + 30 p_2^2 + 45 q p_3 + 55 q_1 p_2 + 
 30 p_1 q_2 + 45 q q_2 + 30 q_1^2\\ [3pt] \quad - 10 p^2 p_2 - 10 p p_1^2 - 5 p^2 q_1.
    \end{array}\right.
\end{equation}

\subsection{Connections between canonical systems and classification results}\label{Miura}

The integrable systems labelled by the letters $A$, $B$, $D$, $G$, and $_0A$ in Theorems~\ref{thmA}--\ref{thmApA} naturally split into families associated with the affine Lie algebras $A_4^{(2)}$, $B_2^{(1)}$, $D_4^{(3)}$, $G_2^{(1)}$, and $A_3^{(2)}$, respectively, together with a choice of vertex in the corresponding Dynkin diagram. Systems belonging to the same family are related either by linear changes of variables or by Miura-type transformations.

\begin{description}
\item[\underline{System (\ref{A_1}):}] 

System \eqref{A_1}  is related to \eqref{DSA4_2_c1}  by a linear transformation 
$$p=\frac{1}{10}\left((5-\sqrt{5})u+(5+\sqrt{5})v\right),\,\,q=\frac{\sqrt{5}}{10}(v-u).$$

\item[\underline{System (\ref{A_2}):}] 
 System~(\ref{A_1}), in variables $\hat{u},\hat{v}$, is related to system~(\ref{A_2}) in variables $u,v$ by the  transformation:
$$\hat{u}= 5u_1+\frac{5}{4}(5-3\sqrt{5})u^2+\frac{5}{2}(5+\sqrt{5})uv+\frac{5}{4}(5+\sqrt{5})v^2,\quad \hat{v}= 5v_1+\frac{5}{4}(5+3\sqrt{5})v^2+\frac{5}{2}(5-\sqrt{5})uv+\frac{5}{4}(5-\sqrt{5})u^2.$$
Throughout this section, the same convention will be used: hatted variables denote the dependent variables of the source system, and unhatted variables those of the target system.

After the rescaling $t\mapsto 2t$ and the linear change of variables
$$P= -\sqrt{5} (u-v),\quad Q=\frac{\sqrt{5}-5}{2} u-\frac{5+\sqrt{5}}{2} v$$ 
system~\eqref{A_2} takes the non-diagonal form 

\begin{equation}\label{A_21}
\left\{
\begin{array}{l}
P_t=4P_3+3Q_3-\left(3PQ_1+6QQ_1-3PQ^2+2Q^3\right)_x,\\[2mm]
Q_t=3P_3+Q_3+\left(3PP_1+6P_1Q+3P^2Q-2Q^3\right)_x.
\end{array}
\right.
\end{equation}
This system was studied in \cite{Talati3rd}, where its Hamiltonian and symplectic structures were obtained. System (\ref{A_2}), in its potential form, and the transformation  (\ref{A_1}) to (\ref{A_2})  were found in \cite{Mesh08, Balak23}.

\item[\underline{System (\ref{A_3}):}] System (\ref{A_1}) is related to (\ref{A_3}) by the   transformation: 
$$\hat{u}=\frac{1}{2}(1+\sqrt{5})u_1+\frac{1}{20}(5+\sqrt{5})(u-v)^2,\quad \hat{v}= \frac{1}{2}(1-\sqrt{5})v_1+\frac{1}{20}(5-\sqrt{5})(u-v)^2.$$

\item[\underline{System (\ref{A_4}):}] System (\ref{A_1}) is related to (\ref{A_4}) by the transformation
$$
\hat{u}= u_1-\frac{1}{20}(5+\sqrt{5})(u-v)^2,\quad \hat{v}= v_1-\frac{1}{20}(5-\sqrt{5})(u-v)^2.
$$

\item[\underline{System (\ref{A_5}):}] System \eqref{A_5}  is related to system \eqref{DSA4_2_c2} by the transformation
\begin{eqnarray*}
p&=&-2u_1-2v_1,\\
q&=&-\frac{1}{10} (5 - \sqrt{5}) u_3 - \frac{1}{10} (5 + \sqrt{5}) v_3 + 
 \frac{1}{5} (-3 + \sqrt{5}) (u - v) u_2 + 
 \frac{1}{5} (3 + \sqrt{5}) (u - v) v_2 + \frac{2}{5} (-3 + \sqrt{5}) u_1^2\\
 &&+ 
 \frac{2}{5} u_1 v_1 - \frac{2}{5} (3 + \sqrt{5}) v_1^2 + 
 \frac{2}{25} (-5 + 4 \sqrt{5}) (u - v)^2 u_1 - 
 \frac{2}{25} (5 + 4 \sqrt{5}) (u - v)^2 v_1 - \frac{2}{25} (u - v)^4.
\end{eqnarray*}

System \eqref{A_5}  is related to system \eqref{DSA4_2_c0}  by the transformation
\begin{eqnarray*}
p&=&(-1 + \sqrt{5}) u_1 - (1 + \sqrt{5}) v_1,\\
q&=&\frac{1}{2 \sqrt{5}} u_3 - \frac{1}{2 \sqrt{5}} v_3 + 
 \frac{1}{10} (-1 + \sqrt{5}) (u - v) u_2 + 
 \frac{1}{10} (1 + \sqrt{5}) (u - v) v_2 + \frac{1}{10} (-2 + \sqrt{5}) u_1^2\\
 && + 
 \frac{1}{10} (-2 - \sqrt{5}) v_1^2+ \frac{2}{5} u_1 v_1 + 
 \frac{1}{50} (5 + \sqrt{5}) (u - v)^2 u_1 - 
 \frac{1}{50} (-5 + \sqrt{5}) (u - v)^2 v_1 + \frac{1}{50} (u - v)^4.
\end{eqnarray*}

The rescaling $t\mapsto \frac{1}{10}t$ and the  transformation
\begin{eqnarray*}
P&=&\frac{6}{5}(u_1+v_1),\\
Q&=&\frac{3}{25} (5 + 3 \sqrt{5}) u_3 + \frac{3}{25} (5 - 3 \sqrt{5}) v_3 - 
 \frac{18}{25} (3 - \sqrt{5}) (u - v) u_2 + 
 \frac{18}{25} (3 + \sqrt{5}) (u - v) v_2  -\frac{36}{25} (2 - \sqrt{5}) u_1^2\\
 &&- \frac{36}{25} (2 + \sqrt{5}) v_1^2 + 
 \frac{108}{25} u_1 v_1 - \frac{36}{125} (5 - 4 \sqrt{5}) (u - v)^2 u_1 - 
 \frac{36}{125} (5 + 4 \sqrt{5}) (u - v)^2 v_1 - \frac{36}{125} (u - v)^4
\end{eqnarray*}
link system~\eqref{A_5} with system~(80) from \cite{mnw07}:
\begin{eqnarray*}
  \left\{
\begin{array}{l}
P_t=Q_1,\\[2mm]
Q_t=\frac{1}{5}P_5+Q_3+D_x\left(PP_2+PQ+\frac{1}{3}P^3\right).
\end{array}
\right.
\end{eqnarray*}

While, the rescaling $t\mapsto \frac{1}{10}t$ and the transformation
\begin{eqnarray*}
P&=&\frac{3}{5} (1 + \sqrt{5}) v_1 + \frac{3}{5} (1 - \sqrt{5}) u_1,\\
Q&=&-\frac{3}{25} (5 + \sqrt{5}) u_3 - \frac{3}{25} (5 - \sqrt{5}) v_3 + 
 \frac{18}{25} (1 - \sqrt{5}) (u - v) u_2 - 
 \frac{18}{25} (1 + \sqrt{5}) (u - v) v_2 + \frac{9}{25} (1 - \sqrt{5}) u_1^2\\
 &&+ 
 \frac{9}{25} (1 + \sqrt{5}) v_1^2 - 36/25 u_1 v_1 - 
 \frac{18}{125} (5 + \sqrt{5}) (u - v)^2 u_1 - 
 \frac{18}{125} (5 - \sqrt{5}) (u - v)^2 v_1 - \frac{18}{125} (u - v)^4
\end{eqnarray*}
link system~\eqref{A_5} with system~(81) from \cite{mnw07}:
\begin{eqnarray*}
\left\{
\begin{array}{l}
P_t=Q_1,\\[2mm]
Q_t=\frac{1}{5}P_5+Q_3+\left(2PP_2+\frac{3}{2}P_1^2+2QP+\frac{4}{3}P^3\right)_x.
\end{array}
\right.
\end{eqnarray*} 

\item[\underline{System (\ref{B_1}):}] System \eqref{B_1}, known as the Hirota-Satsuma equation \cite{his81}, is related to system \eqref{DSB2_1_c2} by the linear change of variables
$$p=\frac{1}{12}\left(-2u-\sqrt{2}v\right),\,\,q=\frac{1}{12}(-2u+\sqrt{2}v).$$




\item[\underline{System (\ref{B_2}):}] 

System \eqref{B_2}  is related to system \eqref{DSB2_1_c0} by the transformation
\begin{eqnarray*}
p=-\frac{2}{\sqrt{3}} v_1 - \frac{1}{3} u + \frac{1}{3} v^2,\quad
q=-\frac{1}{4\sqrt{3}} v_3 - \frac{1}{12} u_2 + \frac{1}{12} v v_2 + 
 \frac{1}{12\sqrt{3}} v u_1  + \frac{1}{24} v_1^2 - \frac{1}{72} u^2.
\end{eqnarray*}

The transformations $t\to -t$ and
\begin{eqnarray*}
P&=&3\sqrt{3}v_1+\frac{3}{2}(u-v^2),\\
Q&=&6 \sqrt{3} v_3 - \frac{3}{2} u_2 - 6 v v_2 - \frac{3}{2} v_1^2 + 
 3 \sqrt{3} (v u_1 + u v_1) - 3 \sqrt{3} v^2 v_1 - 
 \frac{3}{4} u^2 - \frac{3}{2} u v^2 + \frac{3}{4} v^4
\end{eqnarray*}
link system~\eqref{B_2} with equation~(79) from \cite{mnw07}:
\begin{eqnarray*}
 \left\{
\begin{array}{l}
P_t=Q_1,\\[2mm]
Q_t=2P_5+Q_3+\left(2PP_2+P_1^2+\frac{4}{27}P^3\right)_x.
\end{array}
\right.
\end{eqnarray*}

\item[\underline{System (\ref{B_3}):}] System (\ref{B_1}) is related to   system (\ref{B_3}) by the   transformation 
$$\hat{u}= 2u_1,\quad \hat{v}= 2v.$$

\item[\underline{System (\ref{B_4}):}] System (\ref{B_2}) is related to system (\ref{B_4}) by the transformation $$\hat{u}= 2u_1+2v^2,\quad  \hat{v}= \sqrt{2}v.$$

\item[\underline{System (\ref{B_5}):}] System (\ref{B_1}) is related to  system (\ref{B_5}) by the transformation $$\hat{u}= 6u_1-6u^2-6v^2,\quad  \hat{v}= -6\sqrt{2}v_1+12\sqrt{2}uv.$$

\item[\underline{System (\ref{B_6}):}] System (\ref{B_1}) is related to   system (\ref{B_6}) by the transformation $$\hat{u}= -6u_1,\quad  \hat{v}= 6\sqrt{2}\left(v_1-(u-v)^2\right).$$

\item[\underline{System (\ref{D_1}):}] System \eqref{D_1} is related to system \eqref{DSD4_3_c1}  via linear transformation $p=\frac{\sqrt{3}}{2}(v-u),\,\,q=\frac{3}{2}\left((2+\sqrt{3})v+(2-\sqrt{3})u)\right)$.

System (\ref{D_1}) was first introduced by the authors in \cite{mnw09}. 
The symplectic-Hamiltonian formulation for the system and its recursion operator were given in \cite{Talati15}.

\item[\underline{System (\ref{D_2}):}] System \eqref{D_2} is related to system \eqref{DSD4_3_c2} by the transformation
\begin{eqnarray*}
p&=&2 (-1 + \sqrt{3}) u_1 - 2 (1 + \sqrt{3}) v_1,\\
q&=&\frac{2}{3} (3 - 2 \sqrt{3}) u_2 + \frac{2}{3} (3 + 2 \sqrt{3}) v_2 + 
 \frac{4 \sqrt{3}}{3} (u - v) (u_1 + v_1) + \frac{8 \sqrt{3}}{
  9} (u - v)^3.
\end{eqnarray*}

System \eqref{D_2} is related to system \eqref{DSD4_3_c0} via the transformation
\begin{eqnarray*}
p&=&-\frac{2}{3} (3 + \sqrt{3}) u_1 - \frac{2}{3} (3 - \sqrt{3}) v_1,\\
q&=&-\frac{\sqrt{3}}{18}  u_5 + \frac{\sqrt{3}}{18} v_5 - 
 \frac{1}{3} (-2 + \sqrt{3}) (u - v) u_4 - 
 \frac{1}{3} (2 + \sqrt{3}) (u - v) v_4 + 
 \frac{1}{9} (17 - 23 \sqrt{3}) u_1 u_3\\&& + \frac{1}{9} (-17 + 3 \sqrt{3}) v_1 u_3 - 
 \frac{2}{3} (-4 + 3 \sqrt{3}) (u - v)^2 u_3 + 
 \frac{1}{9} (-17 - 3 \sqrt{3}) u_1 v_3 + \frac{1}{9} (17 + 23 \sqrt{3}) v_1 v_3\\ && + 
 \frac{2}{3} (4 + 3 \sqrt{3}) (u - v)^2 v_3 - 
 \frac{5}{18} (-5 + 8 \sqrt{3}) u_2^2 + \frac{5}{18} (5 + 8 \sqrt{3}) v_2^2 - 
 \frac{25}{9} u_2 v_2 - \frac{22}{3} (-3 + 2 \sqrt{3}) (u - v) u_1 u_2\\ && + 
 \frac{10}{3} (-1 + 2 \sqrt{3}) (u - v) v_1 u_2 - 
 \frac{4}{3} (-8 + 3 \sqrt{3}) (u - v)^3 u_2 + 
 \frac{10}{3} (1 + 2 \sqrt{3}) (u - v) u_1 v_2\\ && - 
 \frac{22}{3} (3 + 2 \sqrt{3}) (u - v) v_1 v_2 - 
 \frac{4}{3} (8 + 3 \sqrt{3}) (u - v)^3 v_2 - 
 \frac{2}{27} (-81 + 65 \sqrt{3}) u_1^3 + \frac{2}{27} (81 + 65 \sqrt{3}) v_1^3\\ && + 
 \frac{2}{9} (-27 + 29 \sqrt{3}) u_1^2 v_1 - 
 \frac{2}{9} (27 + 29 \sqrt{3}) u_1 v_1^2 - 
 \frac{2}{3} (-51 + 22 \sqrt{3}) (u - v)^2 u_1^2 + 
 \frac{2}{3} (51 + 22 \sqrt{3}) (u - v)^2 v_1^2\\ && - 
 44 (u - v)^2 u_1 v_1 - 
 \frac{8}{3} (-3 + 4 \sqrt{3}) (u - v)^4 u_1 + 
 \frac{8}{3} (3 + 4 \sqrt{3}) (u - v)^4 v_1 + \frac{8}{3} (u - v)^6.
\end{eqnarray*}

\item[\underline{System (\ref{D_3}):}] System (\ref{D_1}) relates to (\ref{D_3}) by the  transformtation 
$$\hat{u}= -\frac{1}{3}u_1 - \frac{1}{36} (3 + \sqrt{3}) (u - v)^2,\,\,\hat{v}= -\frac{1}{3}v_1-\frac{1}{36} (3 - \sqrt{3})(u - v)^2.$$

\item[\underline{System (\ref{D_4}):}] System (\ref{D_1}) relates to (\ref{D_4}) by the  transformtation 
$$\hat{u}= \frac{1}{3}(3 + 2 \sqrt{3}) u_1+\frac{1}{12}(3 + \sqrt{3}) (u - v)^2,\,\,\hat{v}= \frac{1}{3}(3 - 2\sqrt{3})v_1+\frac{1}{12} (3 - \sqrt{3})(u - v)^2.$$

\item[\underline{System (\ref{D_5}):}] System (\ref{D_1}) relates to (\ref{D_5}) by the  transformtation 
\begin{eqnarray*}
\hat{u}&=& -(1+\sqrt{3})u_1-\frac{1}{2}(3+5\sqrt{3})u^2-(3+\sqrt{3})uv+\frac{1}{2}(3+\sqrt{3})v^2,\\
\hat{v}&=& -(1-\sqrt{3})v_1+\frac{1}{2}(3-\sqrt{3})u^2-(3-\sqrt{3})uv-\frac{1}{2}(3-5\sqrt{3})v^2.
\end{eqnarray*}
The potential version of this system under a linear transformation is equation (9) in \cite{TW17}, where the authors only provided one Hamiltonian operator.
Under the scaling transformation, this system is equations (70) and (71)  in \cite{Gerdjikov20}.

\item[\underline{System (\ref{G_1}):}] System \eqref{G_1} is related to system \eqref{DSG2_1_c1} by linear change of variables
$$p=-\frac{1}{2}(u+v),\,\,q=-\frac{1}{2}u+\frac{1}{6}v.$$

This system was  introduced by the authors in \cite{mnw09}. 
Its recursion operator and compatible Hamiltonian and symplectic structures were constructed in \cite{Voj11}. System \eqref{G_1} admits  reduction $v=0$ to the Kaup-Kuperschmidt equation.

\item[\underline{System (\ref{G_2}):}] 

System \eqref{G_2} is related to system \eqref{DSG2_1_c2} by the transformation
\begin{eqnarray*}
p&=&-u - v_1 + \frac{1}{2} v^2,\quad
q=u_1 + \frac{1}{3} v_2 - \frac{1}{3} v v_1 - \frac{2}{3} u v + \frac{1}{27} v^3.
\end{eqnarray*}

System \eqref{G_2} is related to system \eqref{DSG2_1_c0} by the transformation
\begin{eqnarray*}
p&=&3 v_1 - u + \frac{1}{2} v^2,\\
q&=&\frac{1}{18} v_5 + \frac{1}{54} v v_4 + \frac{1}{9} u v_3 + \frac{1}{6} u_1 v_2 + 
 \frac{1}{6} v_1 u_2 - \frac{13}{54} v_1 v_3 + \frac{1}{18} v u_3 - \frac{4}{27} v_2^2 - 
 \frac{5}{81} v^2 v_3 - \frac{8}{27} v v_1 v_2\\
 &&+ \frac{2}{27} v^2 u_2 + 
 \frac{1}{18} u v u_1 + \frac{1}{18} u^2 v_1 + \frac{7}{54} v u_1 v_1 + 
 \frac{1}{18} u v v_2 - \frac{2}{27} u v_1^2 - \frac{35}{972} v^3 v_2 - 
 \frac{2}{81} v^2 v_1^2\\
 &&+ \frac{7}{324} v^3 u_1 - \frac{2}{81} u v^2 v_1 - 
 \frac{1}{648} v^4 v_1 + \frac{1}{17496} v^2 (18 u - v^2)^2.
\end{eqnarray*}

The reduction $v=0$ in system \eqref{G_2} is the Sawada-Kotera equation. 

\item[\underline{System (\ref{G_3}):}] System (\ref{G_1}) relates to (\ref{G_3}) by the  transformation 
$$\hat{u}= \frac{1}{2}u_1,\quad \hat{v}= v.$$

\item[\underline{System (\ref{G_4}):}] System (\ref{G_1}) relates to (\ref{G_4}) by the  transformation 
$$\hat{u}= u_1-\frac{1}{2}u^2-\frac{3}{2}v^2,\,\,\hat{v}= 3v_1-3v^2-3uv.$$
By setting $v=0$, system \eqref{G_4} reduces to the Kuperschmidt equation 
$$ u_t=u_5+ (5 u_1 u_2 -5 u^2 u_2-5 u u_1^2+u^5)_x.$$
System \eqref{G_4} was first introduced in \cite{Talati13}, where the author constructed its bi-Hamiltonian structure.

\item[\underline{System (\ref{G_5}):}] System (\ref{G_2}) relates to (\ref{G_5}) by the transformation $$\hat{u}= u_1+\frac{1}{2}v^2,\quad \hat{v}= v.$$

\item[\underline{System (\ref{G_6}):}] System (\ref{G_1}) relates to (\ref{G_6}) by the transformation 
$$\hat{u}= -u_1,\,\,\hat{v}= v_1-\frac{1}{2}(u-v)^2.$$

\item[\underline{System (\ref{G_7}):}] System (\ref{G_1}) relates to (\ref{G_7}) by the transformation 
$$\hat{u}= -u_1,\,\,\hat{v}= -3v_1+\frac{3}{2}(u-v)^2.$$

\item[\underline{System (\ref{A_1p}):}] System \eqref{A_1p} relates to system \eqref{DSA3_2_c2} via invertible transformation
\begin{eqnarray*}
p&=&-\frac{1}{6}(3u+\sqrt{3}v),\quad q=\frac{1}{6}(\sqrt{3}v-3u).
\end{eqnarray*}

\item[\underline{System (\ref{A_2p}):}] System \eqref{A_2p} relates to system \eqref{DSA3_2_c0} by the  transformation
\begin{eqnarray*}
p&=&\frac{3   }{2}  (\sqrt{2} v_1 -  u +  v^2),\\
q&=&\frac{\sqrt{2} }{4}   v_3 - \frac{3}{4} u_2 - \frac{9}{8} u^2 + \frac{1}{2} v v_2 - 
 \frac{3}{4} v_1^2 - \frac{3 \sqrt{2}   }{2} v u_1 + \frac{3 \sqrt{2}   }{4} u v_1 - 
 \frac{3}{4} u v^2 - \frac{3 \sqrt{2} }{4} v^2 v_1 - \frac{1}{8} v^4.
\end{eqnarray*}

\item[\underline{System (\ref{A_3p}):}] System \eqref{A_1p} relates to system \eqref{A_3p} by the  transformation
\begin{eqnarray*}
\hat{u}= u_1,\quad \hat{v}=\sqrt{2}v.
\end{eqnarray*}

\item[\underline{System (\ref{A_4p}):}] System \eqref{A_1p} relates to system \eqref{A_4p} by the  transformation
\begin{eqnarray*}
\hat{u}= u_1 - \frac{1}{2}( u^2 + v^2),\quad \hat{v}= \sqrt{3} (v_1 -  uv).
\end{eqnarray*}

\item[\underline{System (\ref{A_5p}):}] System \eqref{A_1p} relates to system \eqref{A_5p} by the  transformation
\begin{eqnarray*}
\hat{u}= 2 u_1,\quad \hat{v}= 2 \sqrt{3} v_1 + 2 \sqrt{3} (u - v)^2.
\end{eqnarray*}

\item[\underline{System (\ref{A_7p}):}] System \eqref{A_2p} relates to system \eqref{A_7p} by the  transformation
\begin{eqnarray*}
\hat{u}= \frac{1}{3}u_1+\frac{1}{2}v^2,\quad \hat{v}= \frac{1}{\sqrt{2}}v.
\end{eqnarray*}

\end{description}

\begin{description}

\item[\underline{System (\ref{K_1}):}] Transformation
\[
p = u_1 + v^2, \quad q = v
\]
relates system (\ref{K_1}) to the triangular system
\begin{eqnarray}\label{sysw10n}
 \left\{\begin{array}{l}
 p_t = p_3 - 6 p p_1,\\
 q_t = 4 q_3 - 6 p q_1 - 3 p_1 q
 \end{array}\right.
\end{eqnarray}
in which the first equation is the \emph{Korteweg--De Vries (KdV) equation}. The second equation can be identified with the second equation of the overdetermined linear problem
\[
\begin{array}{ll}
L\psi = \lambda \psi, & L = D_x^2 - p, \\
\psi_t = A\psi, \qquad & A = 4D_x^3 - 6p D_x - 3p_1,
\end{array}
\]
which is equivalent to the Lax representation $L_t=[A,L]$ of the KdV equation.

\item[\underline{System (\ref{K_2}):}] The Miura transformation
\[
p = u_1 + u^2 - v^2, \quad q^{-1} q_1 = u + v
\]
relates system (\ref{K_2}) to (\ref{sysw10n}).
\end{description}
 \subsection{Linearisable systems}\label{diffsub}

\begin{description}
    
\item[\underline{System (\ref{L_1}):}] The transformation $$u=\frac{p_1}{2 p},\quad v=\frac{q}{\sqrt{p}},$$ 
maps system \eqref{L_1}  into the triangular form
\begin{equation*}
    \left\{\begin{array}{l}
        p_t=p_3+2 q^2 , \\
        q_t=\kappa q_3 ,
    \end{array}  \right.
\end{equation*}
whose symmetries were studied in   \cite{mr1829636}.

\item[\underline{System (\ref{L_2}):}] Miura transformation $$u=\frac{p_1}{2q}-\frac{\delta q^2}{p},\quad v=\frac{q}{\sqrt{p}}$$ 
maps system (\ref{L_2}) into the triangular one:
\begin{eqnarray*}
 \left\{\begin{array}{l} p_t=p_3+4\delta(\kappa-1)qq_2+2(3\kappa-2\delta-\kappa\delta)q_1^2,\\
q_t=\kappa q_3.  \end{array} \right.
\end{eqnarray*}

\item[\underline{System (\ref{L_1pp}):}]
System~(\ref{L_1pp}) is obtained from system~\eqref{L_1} by the scaling transformation $t\mapsto \kappa t$ followed by the limit $\kappa\to\infty$. Consequently, the transformation
\[
u=\frac{p_1}{2p},\qquad v=\frac{q}{\sqrt{p}},
\]
maps system~(\ref{L_1pp}) into the triangular system
\[
p_t=2q^2,\qquad
q_t=q_3.
\]

\item[\underline{System (\ref{L_3pp}):}]
System~(\ref{L_3pp}) is related to system~\eqref{L_2} in the same way, namely by the scaling transformation $t\mapsto \kappa t$ followed by the limit $\kappa\to\infty$. Therefore, the transformation
\[
u=\frac{p_1}{2p}-\frac{\delta q^2}{p},\qquad v=\frac{q}{\sqrt{p}},
\]
maps system~(\ref{L_3pp}) into the triangular system
\[ 
p_t=4\delta qq_2+2(1-\delta)q_1^2,\qquad
q_t=q_3.
\]

\item[\underline{System (\ref{A_8p}):}] The transformation
\begin{eqnarray*}
p=u_1+v^2,\quad q=v    
\end{eqnarray*}
brings system \eqref{A_8p} to
\[
p_t=0,\qquad
q_t=q_3+qp_1+2pq_1,
\]

\item[\underline{System (\ref{A_6p}):}] The transformation
\begin{eqnarray*}
p=u_1+v^2-u^2,\quad q^{-1} q_1=v-u    
\end{eqnarray*}
brings system \eqref{A_6p} to
\[
p_t=0,\qquad
q_t= q_3 + p q_1 + q p_1.
 \]

\item[\underline{System (\ref{M_1}):}] System (\ref{M_1}) possesses a hierarchy of symmetries of all odd orders, which can be generated by a local master symmetry 
\begin{eqnarray*}
\begin{pmatrix}
x u_t+\frac{3}{2} u_2+10 u u_1 -6 u_1 v+(\k-3) (u v_1 + v v_1) +\left((\k+3)u^2 +4(\k-2) u v +3(\k-1) v^2\right)(u-v)
\\ \\
x v_t+\frac{3}{2} \k v_2+ (3\k-1) (u + v) u_1 + 2 \k ( 3 u-5 v) v_1+ \left(3 (\k-1) u^2 +4 (2\k -1) u v -(3\k +1) v^2\right) (u-v) 
\end{pmatrix}    
\end{eqnarray*}
given in \cite{serwa24} (with typos in the second component). The only conserved density for this
system is $u-v$.

\item[\underline{System (\ref{M_2}):}] System (\ref{M_2}) possesses a hierarchy of symmetries of  odd orders, which can be generated by a local master symmetry
\begin{eqnarray*}
\begin{pmatrix}
x u_t+\frac{3}{2} u_2+5 u^2 u_1 +3 u_1 v^2 +(3-\k) u v v_1 +\frac{1}{2} u^5  + (2-\k) u^3 v^2 +(\frac{3}{2}-\k) u v^4
\\ \\
x v_t+\frac{3}{2} \k v_2 +(3\k-1) u u_1 v +3\k u^2 v_1 +5 \k v^2 v_1 +(\frac{3}{2} \k-1) u^4 v +(2\k-1) u^2 v^3 +\frac{\k}{2} v^5
\end{pmatrix}    
\end{eqnarray*}
given in \cite{serwa24}. When $v=0$, it reduces to the Ibragimov-Shabat equation  \cite{mr82b:58007}:
$$
u_t = u_3 + 3u^2 u_2 + 9 u u_1^2 + 3u^4 u_1 .
$$
Introducing new dependent variables $p$ and $q$ by
\[
\frac{p_1}{p}-\frac{u_1}{u}=u^2+v^2,
\qquad
\frac{q_1}{q}-\frac{v_1}{v}=u^2+v^2,
\]
one obtains a differential substitution linking the variables $(u,v)$ and $(p,q)$. 
These differential equations are linear in $p$ and $q$, and of Bernoulli type in $u$ and $v$, and hence can be integrated explicitly.

In the variables $p$ and $q$, system~\eqref{M_2} reduces to the linear system
\begin{equation*}
    \left\{\begin{array}{l}
        p_t=p_3 , \\
        q_t=\kappa q_3 .
    \end{array}  \right.
\end{equation*}

\item[\underline{System (\ref{0M1}):}] System~(\ref{0M1}) and its master symmetry are obtained from system~(\ref{M_1}) by the transformation $u\mapsto -v, v\mapsto -u$ followed by setting $\k=0$. Consequently, it admits a hierarchy of odd-order symmetries generated by a local master symmetry.

\item[\underline{System (\ref{L_4pp}):}] Similarly, after interchanging $u$ and $v$, system~(\ref{L_4pp}) coincides with system~(\ref{M_2}) for $\k=0$. Therefore, it possesses a hierarchy of odd-order symmetries generated by a local master symmetry.

\end{description}

\subsection{Integrable systems corresponding to Cases $16, 18, 22$ and $23$ in Table \ref{table2}}\label{sec64}

The classification of integrable systems~\eqref{sysg} of order $n=5$ corresponding to Cases~$16, 18, 22$ and~$23$ in Table \ref{table2} leads to rather cumbersome equations. However, these systems can be reduced, by differential substitutions, to a simple linear--triangular form.

For example, the system 
\begin{equation}\label{t1}
\left\{
\begin{array}{l}
u_t=u_5 + \left(
  5 u u_3 + 10 u_1 u_2 + v^2 + 10 u^2 u_2 + 
   15 u u_1^2 + 10 u^3 u_1 + u^5\right)_x, \\ [4pt]
v_t=\lambda v_5 + \frac{1}{2} (\lambda - 1) v u_4 + 
 \frac{5}{2} \lambda u v_4 + \frac{5}{2} \lambda v_1 u_3 + 
 5 \lambda u_1 v_3 + 5 \lambda u_2 v_2 + 
 \frac{5}{4} (\lambda - 2) u v u_3 + \frac{5}{2} \lambda u^2 v_3\\ [3pt]
 \qquad + 
 \frac{5}{2} (\lambda - 2) v u_1 u_2 + 5 \lambda u v_1 u_2 + 
 \frac{15}{2} \lambda u u_1 v_2 + \frac{15}{4} \lambda u_1^2 v_1 - 
 \frac{1}{2} v^3 + \frac{5}{4} (\lambda - 4) u^2 v u_2 + 
 \frac{5}{4} \lambda u^3 v_2\\ [3pt] \qquad + \frac{15}{4} \lambda u^2 u_1 v_1 + 
 \frac{5}{16} \lambda u^4 v_1 + \frac{15}{8} (\lambda - 4) u v u_1^2 + 
 \frac{5}{8} (\lambda - 8) u^3 v u_1 + 
 \frac{1}{32} (\lambda - 16) u^5 v
\end{array}  \right.
\end{equation}
with $\lambda\ne 1$ (i.e. satisfying condition~(V.)) admits an infinite hierarchy of symmetries if and only if the following conditions holds:
\begin{description}
    \item[Case 16 and 18:] $l=4,\quad q={\i},\quad \lambda=-4$;
    \end{description} 
    \begin{description}
    \item[Case 22 and 23:] $l=5 ,\quad  q=\alpha\frac{\beta-1}{\alpha-1},\quad\alpha^5=\beta^5=1,\, \beta\ne1,\,\alpha\notin \{1, \beta,\beta^{-1}\}$,
     where the admissible values are
    \[
    \begin{array}{llll}
        \alpha =e^{\frac{2\pi {\i}}{5}},\quad& \beta=e^{\frac{4\pi {\i}}{5}},\quad& q=e^{\frac{2\pi {\i}}{5}}+e^{\frac{4\pi {\i}}{5}},\quad &\lambda=-\frac{13+ 5\sqrt{5}}{22};\\
     \alpha =e^{\frac{4\pi {\i}}{5}},& \beta=e^{\frac{2\pi {\i}}{5}},& q=-1-e^{\frac{6\pi {\i}}{5}},&\lambda=-\frac{13+ 5\sqrt{5}}{22};\\
     \alpha =e^{\frac{2\pi {\i}}{5}},& \beta=e^{\frac{6\pi {\i}}{5}},& q=-1-e^{\frac{8\pi {\i}}{5}},\quad & \lambda=-\frac{13- 5\sqrt{5}}{22};\\
    \alpha =e^{\frac{6\pi {\i}}{5}},& \beta=e^{\frac{2\pi {\i}}{5}},& q=e^{\frac{6\pi {\i}}{5}}+e^{\frac{2\pi {\i}}{5}},& \lambda=-\frac{13- 5\sqrt{5}}{22}.
    \end{array}
    \]   
\end{description}

The differential substitution to the new dependent variables
\[
u=\frac{P_1}{P},\qquad v=\frac{Q}{\sqrt{P}},
\]
transforms system~(\ref{t1}) into the linear--triangular system
\[
P_t=P_5+Q^2,\qquad
Q_t=\lambda Q_5.
\]
Conversely, system~(\ref{t1}) can be recovered from this linear--triangular system by applying the same differential substitution.

According to Table \ref{table2}, in the Case 16 the system (\ref{t1}) admits symmetries of orders $S_1=\{5+4k\,;\, k\in\N\}$, whereas in Case~18 the system is required to admit symmetries of orders
 $S_2=\{ 1+12k\,;\, k\in\N\}\cup \{ 5+12k\,;\, k\in\N\}$. Note that $S_2\subset S_1$. 
Surprisingly, despite this difference in the required symmetry orders, the classification results for Cases~16 and~18 coincide. A similar remarkable phenomenon occurs for Cases~22 and~23.

Our classification results for the integrable systems \eqref{sysg} with (n=5) corresponding to Cases~$16, 18, 22$ and~$23$ are summarised in the following theorem.

\begin{The}\label{thmT}
Every system \eqref{sysg} with (n=5) satisfying conditions (I)--(V) and corresponding to Cases~$16, 18, 22$ or~$23$ can be reduced, by almost invertible transformations, to the following linear triangular systems:
\underline{$W(u)=1,\ W(v)=\frac{5}{2}$}: 
\begin{align}
&\left\{ \begin{array}{l}
p_t=p_5+q^2,\\
q_t=\lambda q_5,
\end{array}\right.\quad u=\frac{p_1}{p},\,\,v=\frac{q}{\sqrt{p}}; \tag{$\rm T_1$} \label{K1_1_5d2}
\end{align}

\underline{$W(u)=1,\ W(v)=2$}:
\begin{align}
&\left\{ \begin{array}{l}
p_t=p_5+qq_1,\\
q_t=\lambda q_5,
\end{array}\right.\quad u=\frac{p_1}{p},\,\,v=\frac{q}{\sqrt{p}}; \tag{$\rm T_2$} \label{K1_1_2}
\end{align}

\underline{$W(u)=1,\ W(v)=\frac{3}{2}$}:
\begin{align}
&\left\{ \begin{array}{l}
p_t=p_5+\alpha qq_2+\beta q_1^2,\\
q_t=\lambda q_5,
\end{array}\right.\quad u=\frac{p_1}{p},\,\,v=\frac{q}{\sqrt{p}}; \tag{$\rm T_3$} \label{K1_1_3d2}
\end{align}

\underline{$W(u)=1,\ W(v)=1$}:
\begin{align}
&\left\{ \begin{array}{l}
p_t=p_5+\alpha qq_3+\beta q_1q_2,\\
q_t=\lambda q_5,
\end{array}\right.\quad u=\frac{p_1}{p},\,\,v=\frac{q}{\sqrt{p}}; \tag{$\rm T_4$} \label{K1_1_1}
\end{align}

\underline{$W(u)=1,\ W(v)=\frac{1}{2}$}:
\begin{align}
&\left\{ \begin{array}{l}
p_t=p_5+\alpha q_2^2+\beta q_1q_3+\delta qq_4,\\
q_t=\lambda q_5,
\end{array}\right.\quad u=\frac{p_1}{p}+\frac{\delta}{2(1-\lambda)}\frac{q^2}{p},\,\,v=\frac{q}{\sqrt{p}}, \tag{$\rm T_5$} \label{K1_1_1d2}
\end{align}
where $\lambda=-4$ in the Cases $16, 18$, and $ \lambda=-\frac{13\pm 5\sqrt{5}}{22}$ in the Cases $22, 23$. 
Here $\alpha, \beta, \delta\in\C$ are arbitrary constants and $\delta\ne\alpha$.
\end{The}

\begin{Rem} We note that system~\eqref{K1_1_1d2} admits a symmetry of order $3$, namely system~\eqref{L_2} with $\kappa=\frac{(1+q)^3}{1+q^3}$, provided that the parameters $\alpha, \beta, \delta$ satisfy
\[
2 \alpha - 2 \beta + 2 \delta + 8 \alpha \kappa - 
 5 \beta \kappa - 4 \delta \kappa + 8 \alpha \kappa^2 - 
 2 \beta \kappa^2 - 7 \delta \kappa^2=0.
\]
\end{Rem}

The linear--triangular integrable systems \eqref{K1_1_5d2}-\eqref{K1_1_1d2}, written in the variables $p,q$, have been studied in \cite{mr99i:35005, mr1829636, ba91, ph02}.

\section{Discussion and Remarks}\label{dissrem}
The classification results presented in this paper are obtained under a set of technical assumptions. 
The evidence emerging from the classification indicates that the resulting lists 
provide a coherent description of the  integrable hierarchies globally.

Restricting attention to S-integrable hierarchies, the classification of odd-order integrable systems of the form
\begin{equation*}\label{sys1homc}
\left\{\begin{array}{l} 
u_t=\lambda_1 u_n+F_1(u_{n-1},v_{n-1},\ldots,u,v),\\
v_t=\lambda_2 v_n+F_2(u_{n-1},v_{n-1},\ldots,u,v),
\end{array}\right.\quad n>1,  \quad \lambda_1\ne\lambda_2,\quad \lambda_1,\lambda_2\in\C.
\end{equation*}
obtained under assumptions~(I)--(V) or~(I)--(V$^\prime$), appears to be essentially complete. More precisely, every S-integrable odd-order system of this form is expected either to occur in the lists of Theorems~\ref{thmA}--\ref{thmApA}, to belong to one of the hierarchies generated by those systems, or to arise from a hierarchy generated by one of the second-order systems classified in \cite{mr1829636,mr89g:58092}. The remaining integrable systems satisfying these assumptions are expected to be C-integrable.

A complete justification of this picture would follow from the verification of the following two statements:
\begin{enumerate}
\item Every S-integrable system of order $7$, $11$, $13$, or $31$, corresponding to Cases~$2, 3, 5, 6, 7, 8, 9, 12$, and~$15$, admits a symmetry of order $3$ or $5$. Such systems would therefore be covered by our present classification. 
\item All integrable hierarchies corresponding to Cases~$16-24$ are C-integrable and can be reduced, by means of differential substitutions, to linear--triangular systems analogous to those presented in Theorem~\ref{thmT}.
\end{enumerate}

There are, however, S-integrable hierarchies that lie outside the assumptions imposed in this paper. For example, the vector modified KdV equation has spectral invariant $\mu_2(m)/\mu_1(m)=1$ (see, e.g., \cite{mr2002j:35275}). Another example is the Ito system \cite{ito82}
\[
\left\{\begin{array}{l} 
u_t=u_3+6uu_1+2vv_1,\\
v_t=2uv_1+2vu_1,
\end{array}\right. .
\]
whose hierarchy has spectral invariant $\mu_2(m)/\mu_1(m)=0$. These hierarchies violate assumptions~(V) and~(V$^\prime$), respectively.
The extension of the global approach developed here to such classes of hierarchies remains an open problem.

The present work is devoted to hierarchies whose seeds are of odd order. Hierarchies beginning with even-order members, or consisting entirely of even-order members, require a more delicate analysis of factorisation properties \cite{Peter09}. The classification of such even-order hierarchies remains open.

We have shown that all S--integrable hierarchies within the class considered here can be related to canonical Drinfeld--Sokolov hierarchies by linear changes of variables or differential substitutions. It would be desirable to construct \emph{native} Lax representations for these systems directly, thereby avoiding the need to invert the differential substitutions relating them to the canonical Drinfeld--Sokolov systems. Such native Lax representations should be related to the canonical ones by gauge transformations, inducing the corresponding Miura-type transformations realised by these differential substitutions. A detailed investigation of this problem is beyond the scope of this paper and will be undertaken elsewhere.

Another noteworthy observation is that the spectral invariant of an integrable hierarchy appears to encode information about the affine Lie algebra underlying its Lax representation. Specifically, suppose that 
$$\frac{\mu_2(m)}{\mu_1(m)}=\frac{(1+q)^m}{1+q^m},$$
where $q$ is a root of unity. In all examples arising from our classification, the order of $q$ divides the Coxeter number of the corresponding affine Lie algebra, while the seeds of the hierarchy are linked to its exponents. Moreover, the order of the recursion operator is divisible by the order of $q$.  One direction of this correspondence follows directly from the general Drinfeld--Sokolov reduction framework and the construction of integrable hierarchies associated with classical affine $W$-algebras \cite{DSKV18}. It is widely expected that the converse also holds, namely that these spectral properties uniquely determine the underlying affine Lie algebra.

Finally, the methodology developed in this paper extends naturally to a broader class of systems, including those with a non-diagonal linear part. Indeed, if the symbol of the linear part has distinct eigenvalues, the system can be formally diagonalised, and the same approach remains applicable.  This technique has been used successfully, for example, in \cite{nw07}.

\section*{Acknowledgements}
The authors thank Peter van der Kamp for useful discussions.
AVM appreciates the hospitality and support received from the
School of Mathematics and Statistics, Ningbo University. 
JPW is supported by the National Natural Science Foundation of China (Grant No. 12571265) and the Ningbo University Research Start-Up Fund.
\medskip

  \begin{center}
    {\bf APPENDIX}
  \end{center}
\appendix

\renewcommand{\theequation}{\Alph{section}.\arabic{equation}}
\setcounter{equation}{0}
\newcommand{\thelemma}{\Alph{section}\arabic{lemma}}
\setcounter{The}{0}

\section{Some results by F. Beukers}\label{FBlemmas}

The results presented in this appendix are due to Frits Beukers (Department of Mathematics, University of Utrecht, Email:f.beukers@uu.nl). 
The statements of Lemmas \ref{FB1} and \ref{Gcubdiag}, together with their proofs, can be found, for example, in\cite{mr99g:35058, wang98, kamp02a}.
Lemmas \ref{Gcubirred} and \ref{Gcubred}  were proved using a method similar to that used for Lemma \ref{Gcubdiag}.

\begin{Lem}\label{FB1} For any integers $m>3$ and $k>1$, and for any choice of nonzero complex parameters $\alpha_i\in\C^*$, $i=0,\ldots,m$, the polynomial
$$
h_{m,k}(x_1,\ldots,x_m)=\alpha_0(x_1+\ldots+x_m)^k-\sum_{i=1}^m\alpha_i x_i^k,\quad \alpha_i\in\C^*,\,\,i=0,\ldots,m
$$
is irreducible over $\C$ .
\end{Lem}

\begin{Lem}\label{Gcubdiag}
Consider the polynomial $g_m(x,y,z)=(x+y+z)^m-x^m-y^m-z^m$. Then 
\begin{itemize}
\item If $m$ is even,  $g_m$ is irreducible over $\C$.
\item If $m$ is odd, $g_m$ admits the factorisation
$$
g_m(x,y,z)=(x+y)(x+z)(y+z)h_m(x,y,z),
$$
where $h_m(x,y,z)$  is irreducible over $\C$..
\end{itemize}
\end{Lem}

\begin{Lem}\label{Gcubirred} For any odd integer $m\ge 3$ and any nonzero parameters $\alpha_1,\alpha_2\in\C^*$ with $\alpha_1\ne\alpha_2$, the polynomials
$$
g_m=\alpha_1(x+y+z)^m-\alpha_1(x^m+y^m)-\alpha_2z^m
$$
and
$$
\hat{g}_m=\alpha_1(x+y+z)^m-\alpha_2(x^m+y^m+z^m)
$$
are irreducible over $\C$.    
\end{Lem}

\begin{Lem}\label{Gcubred} For any integer $m\ge 3$ and  nonzero parameters $\alpha_1,\alpha_2\in\C^*$ with $\alpha_1\ne\alpha_2$,
 consider the polynomial
$$
g_m=\alpha_1(x+y+z)^m-\alpha_1x^m-\alpha_2(y^m+z^m).
$$
\begin{itemize}
\item If $m$ is odd, then $g_m=(y+z) h_m$, where $h_m$ is irreducible over $\C$.
\item If $m$ is even, then $g_m$ is irreducible over $\C$.
\end{itemize}
\end{Lem}

\section{The existence theorem}\label{appc}

In this appendix, we state, without proof, the existence theorem that appears as Theorem 2.3 in \cite{mr99g:35058}, where it is referred to as an implicit function theorem (see also Theorem 2.76, p.27 in \cite{wang98}).
Roughly speaking, the theorem asserts that solving a problem up to a sufficiently high order guarantees the existence of an exact solution. It was formulated in the context of filtered Lie modules and played a key role in the classification of scalar homogeneous equations, where it was used to show that the existence of one symmetry implies the existence of infinitely many symmetries.

Consider a filtered Lie algebra \[\cF = \cF^0 \supseteq \cF^1 \supseteq \cdots \supseteq \cF^l \supseteq \cdots\] and let \(\cV\) be a filtered \(\cF\)-module 
\[\cV = \cV^0 \supseteq \cV^1 \supseteq \cdots \supseteq \cV^l \supseteq \cdots,\qquad \bigcap_{j=0}^{\infty} \cV^j = 0\]
such that the action of \(\cF\) on \(\cV\) satisfies \(X^i \cdot v^j \in \cV^{i+j}\) whenever \(X^i \in \cF^i\) 
and \(v^j \in \cV^j\).

\begin{Def}
We call \(S^0 \in \cF^0/\cF^1\) \emph{relatively \(l\)-prime with respect to \(K^0 \in \cF^0/\cF^1\)} if for all \(j \ge l>0\) and \(X^j \in \cV^j\),
\[
S^0 \cdot X^j \in \operatorname{Im} K^0 \pmod{\cV^{j+1}} \;\Longrightarrow\; X^j \in \operatorname{Im} K^0|_{\cV^j} \pmod{\cV^{j+1}} .
\]
\end{Def}

\begin{Def}
We call \(K^0 \in \cF^0/\cF^1\) \emph{nonlinear injective} if for all \(X^l \in \cV^l\) with \(l > 0\),
\[
K^0 \cdot X^l \in \cV^{l+1} \;\Longrightarrow\; X^l \in \cV^{l+1}.
\]
\end{Def}

\begin{The}\label{themapp}
Let  \(K = K^0 + K^1\) and \(S = S^0 + S^1\), where \(K^0\), \(S^0\in\cF^0/\cF^1\), and \(K^1, S^1 \in \cF^1\). 
Suppose that
\begin{itemize}
\item \([K, S] = 0\),
\item \(K^0\) is nonlinear injective,
\item \(S^0\) is relatively \(l\)-prime with respect to \(K^0\),
\item and there exists some \(\hat{Q} \in \cV^0\) such that
\[
K \cdot \hat{Q} \in \cV^{l} \quad\text{and}\quad S^0 \cdot \hat{Q} \in \cV^1.
\]
\end{itemize}
Then there exists a unique \(Q = \hat{Q} + Q^{l}\) with \(Q^{l} \in \cV^{l}\) such that
\[
K \cdot Q = S \cdot Q = 0.
\]
\end{The}

The following computation is used to prove that $\bg^{0,0}$ is relatively 2-prime with respect to $\bff^{0,0}$ in the proof of Theorem \ref{bbt}.

Let \eqref{sys1sysh} and \eqref{sys1symh} be symbolic representations of system \eqref{sys1hom} and its symmetry \eqref{sys1hom}, respectively. Let
$$
\omega_i(\xi_1)=\lambda_i \xi_1^n,\quad \Omega_i(\xi_1)=\mu_i \xi_1^m
$$
with $n,m\in 2\bbbn+1$. Then, by Lemmas \ref{gdiag}, \ref{gsym} and \ref{gsym2} for every odd $n,m$ we obtain
\begin{eqnarray*}
&&xy(x+y)\,|\,G^{\bullet_1}_{2,0}(x,y),\quad xy(x+y)\,|\,G^{\bullet_2}_{0,2}(x,y),\\
&&y\,|\, G^{\bullet_1}_{1,1}(x,y),\quad x\,|\,G^{\bullet_2}_{1,1}(x,y),\\
&&x+y\,|\, G^{\bullet_1}_{0,2}(x,y),\quad x+y\,|\,  G^{\bullet_2}_{2,0}(x,y),
\end{eqnarray*}
where $\bullet$ stands for either $\omega$ or $\Omega$ (so that, e.g., $G^{\bullet_1}_{2,0}$ denotes either $G^{\omega_1}_{2,0}$ or $G^{\Omega_1}_{2,0}$), and $G^{\omega_1}_{j,k}(x,y)$, $G^{\Omega_1}_{j,k}(x,y)$ are defined by \eqref{Go} and \eqref{GO}.

Assume that \eqref{sys1symh} is the symbolic representation of a $2$-approximate symmetry for system \eqref{sys1sysh}. From relations \eqref{Aijkg} for $i=1,2$ with $j+k=2$ and the above divisibility properties, it  follows that the coefficients $a^i_{j,k}$ are of the form:
\begin{eqnarray*}
&&a_{2,0}^1(x,y)=f^1_{2,0}(x,y)\frac{G^{\omega_1}_{2,0}(x,y)}{xy(x+y)},\quad a_{1,1}^1(x,y)=f^1_{1,1}(x,y)\frac{G^{\omega_1}_{1,1}(x,y)}{y},\\
&&a_{0,2}^1(x,y)=f^1_{0,2}(x,y)\frac{G^{\omega_1}_{0,2}(x,y)}{x+y},\quad a_{2,0}^2(x,y)=f^2_{2,0}(x,y)\frac{G^{\omega_2}_{2,0}(x,y)}{x+y},\\
&&a_{1,1}^2(x,y)=f^2_{1,1}(x,y)\frac{G^{\omega_2}_{1,1}(x,y)}{x},\quad a_{0,2}^2(x,y)=f^2_{0,2}(x,y)\frac{G^{\omega_2}_{0,2}(x,y)}{xy(x+y)},
\end{eqnarray*}
where $f^i_{j,k}$ are polynomials. The corresponding coefficients $A^i_{j,k}$ have analogous expressions, obtained by replacing $G^{\omega_i}_{j,k}$ with $G^{\Omega_i}_{j,k}$.

According to formula \eqref{Aijkg} for $i=1,2$ and $j+k=3$, we have 
\begin{eqnarray*}
&&G^{\omega_1}_{3,0}(\xi_1,\xi_2,\xi_3)A^1_{3,0}(\xi_1,\xi_2,\xi_3)=G^{\Omega_1}_{3,0}(\xi_1,\xi_2,\xi_3)a^1_{3,0}(\xi_1,\xi_2,\xi_3)+R^1_{3,0}(\xi_1,\xi_2,\xi_3),\\
&&G^{\omega_1}_{1,2}(\xi_1,\zeta_1,\zeta_2)A^1_{1,2}(\xi_1,\zeta_1,\zeta_2)=G^{\Omega_1}_{1,2}(\xi_1,\zeta_1,\zeta_2)a^1_{1,2}(\xi_1,\zeta_1,\zeta_2)+R^1_{1,2}(\xi_1,\zeta_1,\zeta_2),
\end{eqnarray*}
where, after substituting the expressions for $a^i_{j,k}$ and $A^i_{j,k}$ with $i=1,2$ and $j+k=2$, the remainder terms $R^1_{3,0}$ and $R^1_{1,2}$ are given by
\begin{eqnarray*}
&&R^1_{3,0}(\xi_1,\xi_2,\xi_3)=\frac{2}{\xi_1\xi_2\xi_3(\xi_1+\xi_2+\xi_3)}\\
&&\qquad \left\langle \frac{f^1_{2,0}(\xi_1 + \xi_2, \xi_3) f^1_{2,0}(\xi_1, \xi_2)(G^{\omega_1}_{2,0}(\xi_1, \xi_2) G^{\Omega_1}_{2,0}(\xi_1 + \xi_2, \xi_3)-G^{\omega_1}_{2,0}(\xi_1 + \xi_2, \xi_3) G^{\Omega_1}_{2,0}(\xi_1, \xi_2))}{(\xi_1+\xi_2)^2}\right\rangle_{\xi}\\
&&\qquad +\left\langle\frac{f^2_{2,0}(\xi_1, \xi_2) f^1_{1,1}(\xi_3, 
\xi_1 + \xi_2)(G^{\omega_2}_{2,0}(\xi_1, \xi_2) G^{\Omega_1}_{1,1}(\xi_3, \xi_1 + \xi_2)-G^{\Omega_2}_{2,0}(\xi_1, \xi_2) G^{\omega_1}_{1,1}(\xi_3, \xi_1 + \xi_2))}{(\xi_1+\xi_2)^2}\right\rangle_{\xi}
\end{eqnarray*}
and
\begin{eqnarray*}
&&\!\!\!
R^1_{1,2}(\xi_1,\zeta_1,\zeta_2)=\left\langle\frac{f^1_{1,1}(\xi_1, \zeta_1 + \zeta_2) f^2_{0,2}(\zeta_1, 
  \zeta_2)(G^{\omega_2}_{0,2}(\zeta_1, \zeta_2) G^{\Omega_1}_{1,1}(\xi_1, \zeta_1 + \zeta_2) - G^{\Omega_2}_{0,2}(\zeta_1, \zeta_2) G^{\omega_1}_{1,1}(\xi_1, \zeta_1 + \zeta_2))}{\zeta_1\zeta_2(\zeta_1+\zeta_2)^2}\right\rangle_\zeta\\
&&\qquad-\left\langle\frac{f^1_{1,1}(\xi_1,\zeta_1) f^1_{1,1}(\xi_1 + \zeta_1,\zeta_2) (G^{\Omega_1}_{1,1}(\xi_1, \zeta_1) G^{\omega_1}_{1,1}(\xi_1 + \zeta_1, \zeta_2) - 
   G^{\omega_1}_{1,1}(\xi_1, \zeta_1) G^{\Omega_1}_{1,1}(\xi_1 + \zeta_1, \zeta_2))}{\zeta_1\zeta_2}\right\rangle_\zeta\\
&&\qquad+\frac{2 f^1_{0,2}(\zeta_1, \zeta_2) f^1_{2,0}(\xi_1, \zeta_1 + \zeta_2) (G^{\omega_1}_{0,2}(\zeta_1, \zeta_2) G^{\Omega_1}_{2,0}(\xi_1,\zeta_1 + \zeta_2) - 
   G^{\Omega_1}_{0,2}(\zeta_1, \zeta_2) G^{\omega_1}_{2,0}(\xi_1,\zeta_1 + \zeta_2))}{\xi_1 (\zeta_1 + \zeta_2)^2 (\xi_1 + \zeta_1 + \zeta_2)}\\
&&\qquad+\left\langle\frac{2 f^1_{0,2}(\xi_1 + \zeta_1, \zeta_2) f^2_{1,1}(\xi_1, \zeta_1) (G^{\omega_2}_{1,1}(\xi_1, \zeta_1) G^{\Omega_1}_{0,2}(\xi_1 + \zeta_1, \zeta_2) - 
   G^{\Omega_2}_{1,1}(\xi_1, \zeta_1) G^{\omega_1}_{0,2}(\xi_1 + \zeta_1, \zeta_2))}{\xi_1 (\xi_1 + \zeta_1 + \zeta_2)}\right\rangle_\zeta
\end{eqnarray*}
By direct computation, we obtain the following statement:
\begin{Lem}\label{quadraTh} Let
$
\omega_i(\xi_1)=\lambda_i \xi_1^n$ and $ \Omega_i(\xi_1)=\mu_i \xi_1^m
$
with $n,m\in 2\bbbn+1$. Then
\begin{itemize}
    \item[{\rm (i)}]  $(\xi_1+\xi_2)(\xi_1+\xi_3)(\xi_2+\xi_3)\mid R^1_{3,0}(\xi_1,\xi_2,\xi_3)$ if and only if
\begin{equation}\label{quadrac1}
2 f^1_{2,0}(x,0) f^1_{2,0}(y, -y) - x^2 y^2 f^1_{1,1}(x, 0) f^2_{2,0}(y, -y)=0 .
\end{equation}
\item[{\rm (ii)}] $(\zeta_1+\zeta_2)\mid R^1_{1,2}(\xi_1,\zeta_1,\zeta_2)$ if and only if
\begin{equation}\label{quadrac2}
2 x^2 f^1_{0,2}(x, -x) f^1_{2,0}(0, y) - y^2 f^1_{1,1}(y, 0) f^2_{0,2}(x, -x)=0.
\end{equation}
\end{itemize}
\end{Lem}

\begin{proof} Setting $\xi_2=-\xi_1+\epsilon$ in $R^1_{3,0}(\xi_1,\xi_2,\xi_3)$ and expanding in a Laurent series at $\epsilon=0$, we obtain
\begin{eqnarray*}
    R^1_{3,0}(\xi_1,-\xi_1+\epsilon,\xi_3)&=&\frac{1}{3\xi_1^2\xi_3^2}\left(2 f^1_{2,0}(0, \xi_3) f^1_{2,0}(\xi_1, -\xi_1)-\xi_1^2 \xi_3^2 f^1_{1,1}(\xi_3, 0) f^2_{2,0}(\xi_1, -\xi_1)\right)\times\\&&\times\left(\omega_1^\prime(\xi_1)\Omega_1^\prime(\xi_3)-\omega_1^\prime(\xi_3)\Omega_1^\prime(\xi_1)\right)+O(\epsilon).
\end{eqnarray*}
Thus, the divisibility of $\xi_1+\xi_2$ (and, by symmetry, by $(\xi_1+\xi_3)(\xi_2+\xi_3)$) holds if and only if the identity \eqref{quadrac1} holds.

Similarly, the Laurent expansion at $\epsilon=0$ of $R^1_{1,2}(\xi_1,\zeta_1,-\zeta_1+\epsilon)$ is of the form
\begin{eqnarray*}
R^1_{1,2}(\xi_1,\zeta_1,-\zeta_1+\epsilon)&=&\frac{1}{\xi_1^2\zeta_1^2}\left(2 \zeta_1^2 f^1_{0,2}(\zeta_1, -\zeta_1) f^1_{2,0}(\xi_1, 0)-\xi_1^2 f^1_{1,1}(\xi_1, 0) f^2_{0,2}(\zeta_1, -\zeta_1)\right)\times\\
&&\times \left(\omega_1^\prime(\xi_1)\Omega_2^\prime(\zeta_1)-\Omega_1^\prime(\xi_1)\omega_2^\prime(\zeta_1)\right)+O(\epsilon),
\end{eqnarray*}
which implies that $(\zeta_1+\zeta_2)\mid R^1_{1,2}(\xi_1,\zeta_1,\zeta_2)$ if and only if the condition \eqref{quadrac2} is satisfied.  
\end{proof}

\section {  Representations of the affine Lie algebras $A_3^{(2)},$ $ A_4^{(2)},$ $ B_2^{(1)},$ $ G_2^{(1)},$ $ D_4^{(3)}$}\label{RepAlg}

In this section, for the convenience of the reader, we provide explicit representations of the affine Lie algebras $A_3^{(2)}, A_4^{(2)}, B_2^{(1)}, G_2^{(1)}$, and $D_4^{(3)}$. For each Lie algebra, we present the system of affine generators $\{e_i,f_i,h_i\}$, $i=0,1,2$. The representations of the Lie algebras $A_3^{(2)}, A_4^{(2)}$, 
and $B_2^{(1)}$ are taken from \cite{DS}.

\begin{itemize}

\item $A_3^{(2)}$: $\qquad  \node{ }{c_0}\!\!\!\Longleftarrow\!\!\!\node{ }{c_2}\!\!\!\Longrightarrow\!\!\!\node{ }{c_1}$
\begin{eqnarray*}
&&e_0=\frac12(e_{1,3}+e_{2,4})\zeta,\quad f_0=2(e_{3,1}+e_{4,2})\zeta^{-1},\quad\,\,\, h_0=e_{1, 1} + e_{2, 2} - e_{3,3} - e_{4,4},\\
&&e_1=(e_{2,1} + e_{4, 3})\zeta,\quad \,\,\,\,f_1=(e_{1,2} + e_{3,4})\zeta^{-1},\quad\,\,\,\,\,\, h_1=-e_{1, 1} + e_{2, 2} - e_{3,3} + e_{4,4},\\
&&e_2=e_{3,2}\zeta,\quad\quad\quad\quad \,\,\,\,\,\,\,\,f_2=e_{2,3}\zeta^{-1},\quad\quad\quad\quad\quad\,\,\, h_2=-e_{2,2}+e_{3,3}.
\end{eqnarray*}
\item $A_4^{(2)}$: $\qquad  \node{ }{c_0}\!\!\!\Longrightarrow\!\!\!\node{ }{c_1}\!\!\!\Longrightarrow\!\!\!\node{ }{c_2}$
\begin{eqnarray*}
&&e_0=e_{1,5}\zeta,\quad\quad\quad\quad \,\,\,f_0=e_{5,1}\zeta^{-1},\quad\quad\quad\quad\quad h_0=e_{1,1}-e_{5,5},\\
&&e_1=(e_{2,1}+e_{5,4})\zeta,\quad f_1=(e_{1,2}+e_{4,5})\zeta^{-1},\quad\,\,\, h_1=-e_{1, 1} + e_{2, 2} - e_{4, 4} + e_{5, 5},\\
&&e_2=(e_{3, 2} + e_{4, 3})\zeta,\quad f_2=2(e_{2, 3} + e_{3, 4})\zeta^{-1},\quad h_2=2 (-e_{2, 2} + e_{4, 4}).
\end{eqnarray*}
\item $B_2^{(1)}$:  
 $\qquad \node{ }{c_0}\!\!\!\Longrightarrow\!\!\!\node{ }{c_2}\!\!\!\Longleftarrow\!\!\!\node{ }{c_1}$
\begin{eqnarray*}
&&e_0=\frac{1}{2} (e_{1, 4} + e_{2, 5})\zeta,\quad f_0=2 (e_{4, 1} + e_{5, 2})\zeta^{-1},\quad h_0=e_{1, 1} + e_{2, 2} - e_{4, 4} - e_{5,5},\\
&&e_1=(e_{2, 1} + e_{5, 4}) \zeta,\quad\,\,\,\, f_1=(e_{1, 2} + e_{4, 5})\zeta^{-1},\quad\,\,\, h_1=-e_{1, 1} + e_{2, 2} - e_{4, 4} + e_{5, 5},\\
&&e_2=(e_{3, 2} + e_{4, 3})\zeta,\quad\,\,\,\, f_2=2 (e_{2, 3} + e_{3, 4})\zeta^{-1},\quad h_2=2 (e_{4, 4} - e_{2,2}).
\end{eqnarray*}
\item $G_2^{(1)}$:  
 $\qquad \node{ }{c_0}\!\!\!-\!\!\-\!\!-\!\!\!-\!\!\!\node{ }{c_1}\!\!\!\equiv\!\Rrightarrow\!\!\!\node{ }{c_2}$
\begin{eqnarray*}
&&e_0=\frac{1}{2}(e_{1,7}+e_{2,8})\zeta,\quad f_0=2(e_{7,1}+e_{8,2})\zeta^{-1},\quad h_0=e_{1, 1} + e_{2, 2} - e_{7, 7} - e_{8,8},\\
&&e_1=(e_{3,2}+e_{7,6})\zeta,\quad\,\,\,\, f_1=(e_{2,3}+e_{6,7})\zeta^{-1},\quad\,\,\, h_1=-e_{2, 2} + e_{3, 3} - e_{6, 6} + e_{7, 7},\\
&&e_2=\frac{1}{2}(2 e_{2, 1} + 2 e_{4, 3} + e_{5, 3} + e_{6, 4} + 2 e_{6, 5} + 2 e_{8, 7})\zeta,\\
&&f_2=(e_{1, 2} + e_{3, 4} + 2 e_{3, 5} + 
 2 e_{4, 6} + e_{5, 6} + e_{7, 8})\zeta^{-1},\\
&&h_2=-e_{1, 1} + e_{2, 2} - 2 e_{3, 3} + 
 2 e_{6, 6} - e_{7, 7} + e_{8, 8}.
\end{eqnarray*}
\item $D_4^{(3)}$: 
$\qquad \node{ }{c_0}\!\!\!-\!\!\-\!\!-\!\!\!-\!\!\!\node{ }{c_1}\!\!\!\Lleftarrow\!\equiv\!\!\!\node{ }{c_2}$  
\begin{eqnarray*}
&&e_0=\left(e_{1, 4} - (1 - {\i} \sqrt{3}) e_{1, 5} + (1 + {\i} \sqrt{3}) e_{2, 6} + (1 + {\i} \sqrt{3}) e_{3, 7} - (1 - {\i} \sqrt{3}) e_{4, 8}\right)\zeta,\\
&&f_0=\frac{1}{4} \left(4 e_{4, 1} - (1 + {\i} \sqrt{3}) e_{5, 1} + (1 - {\i} \sqrt{3}) e_{6, 2} + (1 - {\i} \sqrt{3}) e_{7, 3} - (1 + {\i} \sqrt{3}) e_{8, 4} + 4 e_{8, 5}\right)\zeta^{-1},\\
&&h_0=2 e_{1, 1} + e_{2, 2} + e_{3, 3} - e_{6, 6} - e_{7, 7} - 2 e_{8, 8},\\
&&e_1=\frac{1}{2} \left(2 e_{2, 1} + 
   2 e_{4, 3} + e_{5, 3} + e_{6, 4} + 
   2 e_{6, 5} + 2 e_{8, 7}\right)\zeta,\\
&&f_1=\left(e_{1, 2} + e_{3, 4} + 2 e_{3, 5} + 
 2 e_{4, 6} + e_{5, 6} + e_{7, 8}\right)\zeta^{-1},\\
&&h_1=-e_{1, 1} + e_{2, 2} - 2 e_{3, 3} + 
 2 e_{6, 6} - e_{7, 7} + e_{8, 8},\\
&&e_2=(e_{3, 2} + e_{7, 6})\zeta,\quad f_2=(e_{2, 3} + e_{6, 7})\zeta^{-1},\quad h_2=-e_{2, 2} + e_{3, 3} - e_{6, 6} + e_{7, 7}.
\end{eqnarray*}

\end{itemize}
\bibliographystyle{unsrt}
\bibliography{kdv}
\end{document}

%% file: mac.tex
\def\blue#1{\textcolor{blue}{#1}}
\def\red#1{\textcolor{red}{#1}}

\def \ring {{\cal R}}

\def \ord {\, \mbox{ord}}

\def \max {\, \mbox{max}}
\def \spc {\, \mbox{Span}_{\bbbc}}

\def \hq {{\hat{q}}}
\def \hs {{\hat{s}}}

\def \cR {{\mathcal{R}}}
\def \hR {{\hat{\mathcal{R}}}}
\def \hL {{\hat{\mathcal{L}}}}

\def\N {{\mathbb N}}
\def\bbbn {{\mathbb N}}
\def\bbbc{{\mathbb C}}
\def\C{{\mathbb C}}
\def\bbbz{{\mathbb Z}}

\def\bu{{\bf u}}
\def\bff{{\bf f}}
\def\bg{{\bf g}}
\def\bh{{\bf h}}
\def\zp{{{\mathbb Z}_{\ge 0}}}

\def\cL{{\mathcal{L}}}

\def\hu{{\hat{u}}}
\def\hv{{\hat{v}}}
\def\i{\rm i}

\def\raph{{\, \xmapsto{\ \phi\ }\,}}

\def\cA{{\cal A}}

\def\cL{{\cal L}}

\def\cR{{\cal R}}

\def\cF{{\cal F}}

\def\cV{{\cal V}}

\def\pr#1{{\Pi_{\cR}^{#1}}}
\def\pl#1{{\Pi_{\cL}^{#1}}}

\def\l{\lambda}
\def\k{\kappa}
\def\lb{\overline{\lambda}}

\newtheorem{The}{Theorem}[section]
\newtheorem{Def}[The]{Definition}
\newtheorem{Pro}[The]{Proposition}
\newtheorem{Lem}[The]{Lemma}
\newtheorem{Cor}[The]{Corollary}
\newtheorem{Ex}[The]{Example}
\newtheorem{Rem}[The]{Remark}

\DeclareMathOperator{\lcm}{lcm}

\newcommand\node[2]{\overset{#1}{\underset{#2}{\circ}}}